\documentclass[aps, reprint, prx, nofootinbib, superscriptaddress]{revtex4-2}

\usepackage[T1]{fontenc}
\usepackage[utf8]{inputenc}
\usepackage{CJKutf8}

\makeatletter
\def\l@subsubsection#1#2{}
\makeatother

\makeatletter
\renewcommand\onecolumngrid{%
    \do@columngrid{one}{\@ne}%
    \def\set@footnotewidth{\onecolumngrid}%
    \def\footnoterule{\kern-6pt\hrule width 1.5in\kern6pt}% Optional: fixes footnote line width
}
\makeatother

\usepackage{amsmath, amssymb, amsthm, mathtools}
\usepackage{bm}
\usepackage{physics}
\usepackage{mathrsfs}
\usepackage{esint}
\usepackage{thmtools}
\usepackage{thm-restate}
\usepackage{chngcntr}

\usepackage{enumitem}

\usepackage{graphicx}
\usepackage[dvipsnames]{xcolor}
\usepackage[most]{tcolorbox}

\usepackage{hyperref}
\hypersetup{
  colorlinks=true,
  linkcolor=Blue,
  citecolor=Blue,
  urlcolor=Blue
}

\newtheorem{proposition}{Proposition}

\newtheorem{corollary}{Corollary}

\declaretheorem[name=Theorem]{thm}
\declaretheorem[name=Proposition]{prop}
\declaretheorem[name=Lemma]{lem}
\declaretheorem[name=Corollary]{cor}

\newtheorem{definition}{Definition}

\newcommand{\Var}{\operatorname{Var}}

\begin{document}

%\preprint{APS/123-QED}

% ----- Title -----
\title{Precision Limits of Multiparameter Markovian-Noise Metrology}

\author{Anthony J. Brady}\email{ajbrady.q@gmail.com}
\affiliation{Joint Center for Quantum Information and Computer Science, NIST/University of Maryland, College Park, MD, 20742, USA}
\affiliation{Joint Quantum Institute, NIST/University of Maryland, College Park, MD, 20742, USA}

\author{\begin{CJK}{UTF8}{gbsn}Yu-Xin Wang (王语馨)\end{CJK}}
\affiliation{Joint Center for Quantum Information and Computer Science, NIST/University of Maryland, College Park, MD, 20742, USA}

\author{Luis Pedro Garc\'ia-Pintos}
\affiliation{Quantum and Condensed Matter Physics Group (T-4), Theoretical Division, Los Alamos National Laboratory, Los Alamos, NM, 87545, USA}

\author{Alexey V. Gorshkov}
\affiliation{Joint Center for Quantum Information and Computer Science, NIST/University of Maryland, College Park, MD, 20742, USA}
\affiliation{Joint Quantum Institute, NIST/University of Maryland, College Park, MD, 20742, USA}

\begin{abstract}
    Measuring stochastic signals (``noise metrology'') constitutes a central task in quantum sensing and the characterization of open quantum systems. Here we establish ultimate precision bounds for multiparameter estimation of stochastic signals encoded through Markovian Lindblad dynamics, allowing for arbitrary quantum control and noiseless ancillae. Although Markovianity enforces standard-quantum-limit scaling with sensing time $T$, our bounds reveal Heisenberg-type scaling in the number of dissipative channels, $R$: when the stochastic signal exhibits high-rank correlations across the $R$ channels and the probe is entangled, the average variance (per parameter) scales no better than $\Omega(1/(TR^2))$. For collective $k$-body dissipation, $R=\Theta(N^k)$, signifying super-Heisenberg scaling with the system size $N$. We further show that, when the unknown parameters enter through the dissipative eigenrates, a Rapid Prepare-and-Measure (RPM) protocol that tracks many distinct quantum jumps in parallel attains these limits. In this regime, the estimation problem reduces to a multi-Poisson counting model, providing a conceptually clean route to optimal quantum noise metrology. We illustrate the breadth of the framework with applications to networked noise metrology, collective many-body dissipation, learning Pauli noise, and subdiffraction quantum imaging.
\end{abstract}

\date{\today}

\maketitle

\tableofcontents

%\newpage

%==========
\section{Introduction}
\label{sec:intro}

Precision sensing drives progress and innovation across science and technology. Quantum sensing~\cite{Degen2017:QuSensing, Pezze2018:QuMetrologyAtoms, Lloyd2018:PhotonQuSensing}, a cornerstone of quantum information science and technology, employs quantum resources, such as entanglement~\cite{Huang2024:EntangledMetrology}, quantum memory~\cite{Zaiser2016:QuMemorySensing, Chen2022:ExpSepMemory, Huang2022:QuAdvLearn}, and collective measurements (via multiplexed quantum readout)~\cite{Massar1995:EntanglMeas, Conlon2023:EntanglMeas}, to surpass limitations inherent to classical sensing strategies. 

Formally, quantum sensing aims to estimate unknown parameters $\bm{\theta}=(\theta_1, \dots, \theta_d)$ that imprint on the probe state $\rho$ of a quantum system through a parametric process $\mathcal{E}_{\bm{\theta}}(\rho)$. We prepare and optimize the quantum probe, perform measurements that extract information about $\bm{\theta}$, and construct estimators $\check{\bm{\theta}}$ from the measurement data. Quantum estimation theory~\cite{Helstrom1976, Holevo1982, Paris2009:QFI} then quantifies the ultimate precision (defined as the inverse variance) in this task through the quantum Fisher information (QFI) and the associated quantum Cram\'er-Rao bound~\cite{Braunstein1994:Bures}, which establishes fundamental limits on any quantum sensing protocol.

Much work in quantum sensing focuses on unitary parameter encoding, where coherent signal accumulation supports quadratic precision scaling (the Heisenberg limit) with both system size $N$ and time $T$~\cite{Giovannetti2004:BeatSQL, Giovannetti2006:QuMetrology}. While unitary encoding permits Heisenberg scaling in time through coherent signal accumulation, many frontier tasks target stochastic signals, such as learning noise correlations or characterizing dissipative couplings, where coherent buildup in time is intrinsically absent. 

Nevertheless, quantum resources still deliver strong advantages for sensing stochastic signals, including entanglement-enhanced estimation of correlated noise with quantum sensor networks~\cite{Brady2024:NoiseQSN, Wang2024:ExpSensing, Prabhu2026:ExpSensingQSA, Wang2026:PowerLaw}, exponential complexity advantage in learning large families of noise parameters using quantum memory~\cite{Chen2022:PauliChEst, Chen2024:TightPauliLearn, Caro2024:PTMlearning, Kim2025:LearnResources}, and enhanced information extraction in quantum-optical imaging through multiplexed quantum readout~\cite{Tsang2019:Starlight, Albarelli2020:ImagingPersp, Defienne2024:AdvQuImaging}. These motivations extend across noise characterization and benchmarking of quantum devices~\cite{Hashim2021:RandCompil, Burnett2019:BenchmarkDecoherence, Harper2020:PauliLearn, Lupke2020:Benchmarking, Van2023:PEC, Hashim2025:BenchmarkingRvw, Proctor2025:Benchmarking}, interrogating condensed-matter and many-body systems~\cite{Casola2018:nvManyBody, Rovny2024:nvManyBody, Rovny2022NVcovariance, Ziffer2024:qnsCriticality, Cheng2025:mplex, Rovny2025:nvEntResource, Zhou2025:nvEntResource}, and searches for new physics~\cite{YeZoller2024:Essay, Bass2024:NatRvw}. Despite notable successes, existing results on \emph{quantum noise metrology} largely address specific noise models or single-parameter regimes, and a unified multiparameter theory that yields computable, process-agnostic precision bounds for parametric dissipative processes remains incomplete.

In this work, we establish fundamental limits for multiparameter estimation of stochastic signals encoded in a quantum system via Lindbladian dynamics. We derive universal upper bounds on the QFI matrix for parametric dissipative processes (Theorem~\ref{thm:qfi-matrix}), accounting for arbitrary quantum control and noiseless ancillae. While strict Markovianity forces linear scaling in time $T$ (the standard quantum limit, SQL), our bounds reveal a Heisenberg-type scaling in the number of jump operators (Theorem~\ref{thm:gen-Heisenberg}). Specifically, when the dissipative process exhibits maximal connectivity across jump-index space and the probes are entangled, the precision scales at most as $\order{T R^2}$. Here $R$ counts the number of dissipative channels (independent, local jump operators), which can scale polynomially with the system size $N$ for collective many-body dissipation (Corollary~\ref{cor:super-heisenberg}). This ``super-Heisenberg'' scaling in $N$ mirrors that of the unitary case~\cite{Roy2008:ExpMetrology, Boixo2007:QuEstim} but, uniquely for noise metrology, identifies connectivity in jump space as a metrological resource, generalizing prior insights on networked quantum noise sensing~\cite{Brady2024:NoiseQSN}; see also~\cite{Wang2024:ExpSensing, Prabhu2026:ExpSensingQSA, Wang2026:PowerLaw}.

We further show that a Rapid Prepare-and-Measure (RPM) protocol that quickly tracks many quantum jumps in parallel approaches the fundamental precision limits when the unknown parameters enter through the \emph{eigenrates} of the dissipative process (Lemma~\ref{lem:rpm-fisher}, Theorem~\ref{thm:rpm-qcrb}). In this sense, our RPM protocol extends prior single-parameter fast-reset strategies~\cite{Correa2015:OptThermometry, Sekatski2022:Thermometry, Gardner2025:FastReset, Wang2026:PowerLaw} to the multiparameter setting and provides a conceptually simple route to saturating quantum Cram\'er-Rao bounds.

We organize the remainder of the paper as follows. In Sec.~\ref{sec:setup}, we introduce the estimation setting, the class of parametric Markovian dissipative processes that we analyze in this work, and the collisional purification of the dissipative dynamics that we use to upper bound estimation precision. In Sec.~\ref{sec:results}, we present our precision bounds for general parametric dissipative processes, develop the generalized Heisenberg scaling, and describe protocols that attain these limits in important regimes, particularly for eigenrate estimation. In Sec.~\ref{sec:examples}, we illustrate the framework with applications to networked noise metrology, including an example on collective spin dissipation; super-Heisenberg scaling for many-body dissipators, including an example on collective multipole dissipation; Pauli-noise estimation with and without quantum memory; and subdiffraction incoherent quantum imaging. In Sec.~\ref{sec:outlook}, we close with an outlook.

We also provide detailed appendices that include a derivation of the collisional purification of Markovian dynamics (Appendix~\ref{app:collisional}), full proofs of our main results (Appendix~\ref{app:proofs}), a technical discussion of tightness of our precision bounds via Uhlmann extremality (Appendix~\ref{app:uhlmann-extremality}), and additional details on learning Pauli noise (Appendix~\ref{app:pauli-noise}) and subdiffraction imaging (Appendix~\ref{app:subdiff-imaging}).

%==========
%==========
\section{Setup}
\label{sec:setup}

We study estimation of parameters $\bm{\theta}=(\theta_\alpha)_{\alpha=1}^d$ encoded in a quantum probe through a dissipative, memoryless (Markovian) process. The probe is prepared in an initial state $\psi\triangleq \dyad{\psi}$ which may include ancillae, evolves under Lindbladian evolution $\mathcal{E}_{\bm{\theta}}(\psi)$ with parametric dissipator $\mathscr{D}_{\bm{\theta}}$ and arbitrary control, and is then measured to produce estimators $\check{\bm{\theta}}$. To analyze estimation precision, we use local (asymptotic) quantum-estimation theory based on Fisher information (see Refs.~\cite{Paris2009:QFI, Sidhu2020:QFIrvw, LiuYuanLuWang2020, Meyer2021, Pezze2025:AvMultiQuEst} for background). We also introduce a %collisional
purification of $\mathcal{E}_{\bm \theta}(\psi)$ that allows us to compute upper bounds on the precision for arbitrary Markovian dissipative encodings.

%==========
\subsection{Parameter Estimation}
\label{sec:estimation}

\paragraph*{Estimation precision.}
Let $\check{\bm{\theta}}$ denote the unbiased estimators constructed from measurement outcomes, with uncertainty characterized by the covariance matrix
\begin{equation}
    \bm{\Sigma}_{\alpha\beta}\ \triangleq\ \mathbb{E}\left[(\check{\theta}_\alpha-\theta_\alpha)(\check{\theta}_\beta-\theta_\beta)\right].
\end{equation}
Throughout, we also use the (diagonal) absolute precision,
\begin{equation}
    \label{eq:precision}
    \tau_\alpha \triangleq \frac{1}{\bm{\Sigma}_{\alpha\alpha}},
\end{equation}
as a simple but informative performance metric.

\paragraph*{Classical and quantum Fisher information.}
Given a POVM $\{\Pi_j\}$, outcome probabilities derive from the Born rule $p_j(\bm{\theta})=\Tr(\Pi_j\rho(\bm{\theta}))$ with $\sum_j \Pi_j = I$, and the classical Fisher information matrix reads
\begin{equation}\label{eq:fisher-metric}
  (F_C)_{\alpha\beta}
  = \sum_j \frac{(\partial_\alpha p_j)(\partial_\beta p_j)}{p_j}
  = 4\sum_j (\partial_\alpha \sqrt{p_j})(\partial_\beta \sqrt{p_j}).
\end{equation}
For $\nu$ i.i.d.\ repetitions and unbiased estimators, the classical Cram\'er-Rao bound establishes a (protocol-dependent) precision limit,
\begin{equation}
  \bm{\Sigma} \ \ge\ \frac{1}{\nu}F_C^{-1}.
\end{equation}
For any POVM, the classical Fisher information matrix $F_C$ is bounded from above by the (protocol-independent) QFI matrix $F_Q$, leading to the celebrated quantum Cramér-Rao bound (QCRB)~\cite{Helstrom1976, Holevo1982, Braunstein1994:Bures, Fujiwara1995:QFImetric}
\begin{equation}\label{eq:qcrb}
  \bm{\Sigma} \ \ge\ \frac{1}{\nu}F_C^{-1} \ \ge\ \frac{1}{\nu}F_Q^{-1}.
\end{equation}
A useful inequality is $(F_Q^{-1})_{\alpha\alpha}\ge 1/(F_Q)_{\alpha\alpha}$; hence, correlations in $F_Q$ typically degrade marginal precisions relative to what the diagonal entries alone suggest. (In general, the multiparameter QCRB involving $F_Q$ may not always be achievable due to measurement incompatibility~\cite{Ragy2016compatibility,Pezze2017:Multiparam, LiuYuanLuWang2020, Albarelli2022:Incompatibility, Pezze2025:AvMultiQuEst} but tight, closely related bounds can be formulated, e.g., for the case with two parameters~\cite{Sidhu2021:HolevoBound}.)

\paragraph*{Purification viewpoint.}
For a parametric quantum channel $\rho(\bm{\theta})=\mathcal{E}_{\bm{\theta}}(\psi)$, any Stinespring dilation~\cite{Stinespring1955} produces a purification
$\Psi(\bm{\theta})=U(\bm{\theta})(\psi\otimes \varphi_{\rm env})U^\dagger(\bm{\theta})$ such that
$\rho(\bm{\theta})=\Tr_{\rm env}\Psi(\bm{\theta})$. Since partial trace is parameter-independent,
\begin{equation}
  F_Q[\Psi(\bm{\theta})]\ \ge\ F_Q[\rho(\bm{\theta})]
\end{equation}
via quantum data processing~\cite{Nielsen1996}. The choice of dilation is not unique; in fact, minimizing $F_Q[\Psi]$ over environment unitaries recovers $F_Q[\rho]$ (per Uhlmann's theorem~\cite{Uhlmann1976}). In this work we focus on a specific purification---the \emph{canonical collisional purification} below---whose QFI is computable and yields an upper bound which is often tight; we discuss tightness more generally in Appendix~\ref{app:uhlmann-extremality}. 

%==========
\subsection{Dissipative Dynamics and Collisional Purification}
\label{sec:markov-evol}

\paragraph*{Markovian evolution.}
We consider a quantum system evolving under a Lindbladian generator with parametrized dissipator~\cite{Lindblad1976:Generators, Gorini1976, Breuer2007, Albert2016:GeometryLindbladians},
\begin{equation}
    \label{eq:lindbladian}
  \partial_t \rho(\bm{\theta};t) = \mathcal{L}_{\bm{\theta}, t}(\rho),\quad
  \mathcal{L}_{\bm{\theta}, t}(\rho)\triangleq -i\comm{H_c(t)}{\rho}+\mathscr{D}_{\bm{\theta}}(\rho),
\end{equation}
where $H_c(t)$ is a parameter-independent control Hamiltonian and $\mathscr{D}_{\bm{\theta}}$ is the parametric dissipator. The formal solution of the Lindblad equation is
\begin{equation}
    \label{eq:rho-t}
    \rho(\bm{\theta};t)=\mathcal{T}\exp\!\Big(\int_0^t \dd{s}\,\mathcal{L}_{\bm{\theta}, s}\Big)(\psi),
\end{equation}
where $\psi=\rho(0)$. We assume an initial pure state throughout, $\psi \triangleq \dyad{\psi}$, which may include ancillary degrees of freedom that evolve trivially. By convexity of the QFI, sensors prepared in a mixed initial state satisfy the same upper bounds.

Let $\{L_i\}_{i=1}^R$ denote jump operators in a local basis (independent of $\bm{\theta}$). Parameter dependence enters through the Hermitian positive semidefinite matrix $\Gamma(\bm{\theta})\in\mathbb{C}^{R\times R}$,
\begin{equation}
    \label{eq:gamma-matrix}
    \Gamma(\bm{\theta})=V(\bm{\theta}) \Lambda(\bm{\theta}) V^\dagger(\bm{\theta}),
    \quad
    \Lambda={\rm diag}(\gamma_k),\quad \gamma_k \geq 0,
\end{equation}
where $\Lambda \in \mathbb{R}_+^{R\times R}$ and $V\in\mathbb{C}^{R\times R}$ and $k=1,\dots, R$.
From which we define the canonical jump operators
\begin{equation}
    \label{eq:can-jump}
    J_k(\bm{\theta}) = \sum_{i=1}^R V_{ik}(\bm{\theta}) L_i.
\end{equation}
Thus the eigenrates $\gamma_k$ and the eigenframe $V$ both depond on $\bm{\theta}$. The dissipator may be written in the physical frame and in the canonical frame, respectively,
\begin{align}
    \mathscr{D}_{\bm{\theta}}(\rho)
    &= \sum_{i,j=1}^R\Gamma_{ij}(\bm{\theta})
    \left(L_i\rho L_j^\dagger-\tfrac12\acomm{L_j^\dagger L_i}{\rho}\right)
    \label{eq:dissip-phys}
    \\
    &=\sum_{k=1}^R \gamma_k(\bm{\theta})
    \left(J_k(\bm{\theta})\rho J_k^\dagger(\bm{\theta})-\tfrac12\acomm{J_k^\dagger J_k(\bm{\theta})}{\rho}\right).
    \label{eq:dissip-can}
\end{align}

\paragraph*{Collisional purification.}
We realize the Markovian dissipator by coupling the probe to a memoryless bosonic bath through a purification of the Lindbladian evolution; see Appendix~\ref{app:collisional} for further details. Let $U(\bm{\theta};t)=\mathcal{T}{\rm exp}(-i\int_0^t\dd{s}H_{\rm tot}(\bm{\theta}; s))$ act on the joint system-environment Hilbert space $\mathscr{H}_S\otimes\mathscr{H}_{\rm env}$, with
\begin{equation}
    H_{\rm tot}(\bm{\theta}; t)=H_c(t)+H(\bm{\theta};t),
\end{equation}
and interaction Hamiltonian
\begin{equation}
    \label{eq:collisional-H}
    H(\bm{\theta};t)
    = \sum_{k=1}^R \sqrt{\gamma_k(\bm{\theta})}\,J_k(\bm{\theta})\otimes e_k^\dagger(t) + {\rm h.c.}
\end{equation}
We take the bath operators $\{e_k(t)\}$ as bosonic annihilation operators satisfying
$\comm*{e_k(t)}{e_{k'}^\dagger(t')}=\delta_{kk'}\delta(t-t')$ and assume the bath state $\varphi_{\rm env}$ is the (memoryless, Gaussian) vacuum state, such that
\begin{equation}
    \Tr(\varphi_{\rm env}e_k(t)e_{k'}^\dagger(t'))=\delta_{kk'}\delta(t-t'),
\end{equation}
with all other first and second moments vanishing. Defining the global state
\begin{equation}
    \Psi(\bm{\theta};t)
    = U(\bm{\theta};t)\,(\psi\otimes\varphi_{\rm env})\,U^\dagger(\bm{\theta};t),
\end{equation}
we have $\rho(\bm{\theta};t)=\Tr_{\rm env}\Psi(\bm{\theta};t)$, and this dilation reproduces the Lindblad evolution in Eq.~\eqref{eq:rho-t}.

Operationally, the purification corresponds to a collisional model in which the system couples at each instant to a fresh bath element prepared in the vacuum, thereby enforcing Markovianity. This dilation is not unique: Taking $U(\bm{\theta};t)\mapsto (I\otimes U_E(\bm{\theta}; t))\,U(\bm{\theta};t)$ leaves the reduced evolution unchanged. 

In what follows, we employ the canonical ($U_E=I$) collisional purification to derive precision bounds and infer protocols. The resulting bounds need not be universally tight, precisely because different choices $U_E$ correspond to different purifications of the same reduced dynamics; we do, however, identify regimes in which the bounds are tight. Appendix~\ref{app:uhlmann-extremality} further discusses the general issue of tightness from the perspective of Uhlmann's theorem.

%==========
%==========
\section{Precision Limits and Protocols}
\label{sec:results}

In this Section, we present our main results (Theorems~\ref{thm:qfi-matrix} and~\ref{thm:gen-Heisenberg}). In particular, we analyze the QFI matrix of the collisional purification of parametric dissipative dynamics. We first establish the temporal structure of the QFI, showing that, for memoryless processes, the attainable precision grows at most linearly with sensing time (Lemma~\ref{lem:qfi-flow}). We then derive a general closed form for the \emph{QFI matrix flow} (Theorem~\ref{thm:qfi-matrix}), which upper bounds the metrological performance of arbitrary dissipatively encoded processes and makes explicit the roles of the eigenrates $\gamma_r$ and the eigenframe $V$; this key result holds under the most general adaptive strategy compatible with continuous-time Lindblad dynamics (see, e.g., Ref.~\cite{Albarelli2022:Incompatibility} for a complementary multiparameter, channel-based approach). Building on this, we prove a \emph{Generalized Heisenberg Limit} in the jump-index dimension $R$ (Theorem~\ref{thm:gen-Heisenberg}) and identify structural conditions required to approach it. Finally, we specialize to the practically relevant case of \emph{eigenrate estimation} (parameters encoded only in $\gamma_r$; Proposition~\ref{prop:rate-qfi}). We provide sufficient conditions and an explicit RPM protocol for attaining the multiparameter QCRB (Theorem~\ref{thm:rpm-qcrb}), extending single-parameter, fast-reset strategies~\cite{Correa2015:OptThermometry, Sekatski2022:Thermometry, Gardner2025:FastReset, Wang2026:PowerLaw} to multiparameter estimation.

%==========
\subsection{Collisional QFI and Generalized Heisenberg Limit}
\label{sec:collisional-qfi}

\paragraph*{Temporal structure and flow.}
For parametric memoryless processes, the achievable precision is bounded by a linear scaling in time. This behavior is expected when estimating dissipative processes~\cite{Fujiwara2008:fibre, Escher2011:Framework, Dobrza2012ElusiveHeisenberg, Kolodynski2013:EfficientTools, Catana2015, Guta2017, Gorecki2023:CausalAdaptive} and is well known in single-parameter Markovian settings~\cite{Correa2015:OptThermometry, Sekatski2022:Thermometry, Wan2022:AdaptiveMarkov, Das2025:QFItimeScalings, Gardner2025:FastReset}.
Though not our main result, we derive this scaling using collisional purification techniques. For clarity, we consider single-parameter estimation and a single jump operator; the same time-scaling extends to multiple parameters and multiple jump operators by similar reasoning.

\begin{restatable}[QFI Flow]{lem}{QFIFlow}
    \label{lem:qfi-flow}
    Consider a parametric dissipator
    \begin{equation}
        \mathscr{D}_\theta(\rho)
        =\gamma_\theta\!\left(L\rho L^\dagger-\tfrac{1}{2}\acomm{L^\dagger L}{\rho}\right),
    \end{equation}
    where the rate $\gamma_\theta$ encodes the parameter $\theta$.
    Then the QFI of the collisional purification $F_Q(t)$ satisfies
    \begin{equation}
        \label{eq:qfi-flow}
        F_Q(t)=\int_{0}^{t}\dd{s}\,\dot{F}_Q(s)\ \le\ 4t (\partial_\theta \sqrt{\gamma_\theta})^2\,\norm{L}^{2},
    \end{equation}
    where $\dot{F}_Q\triangleq \partial_t F_Q(t)$ is the QFI flow and $\norm{\cdot}$ denotes the spectral norm.
\end{restatable}

\begin{proof}[Proof sketch]
    The linear time scaling follows directly from the memoryless assumption of the bath (viz., $\Tr(\varphi_{\rm env}e(s)e^\dagger(s'))=\delta(s-s')$) and can be recovered within the collisional purification; we provide a full proof in Appendix~\ref{app:lem-qfi-flow} by modeling the environment as a product of independent bath elements in the discrete-time approximation.
\end{proof}

In words, estimation precision for parametric Markovian dynamics obeys the SQL in time $t$. This is consistent with temporal locality of Markovian evolution, for which coherent buildup across disjoint time intervals is absent; the QFI flow thus represents the natural object to analyze in the Markovian setting~\cite{Lu2010:QFIflow}. This perspective is conceptually related to Fisher-information-based speed limits for stochastic quantum dynamics~\cite{Pires2016:QSL, Ito2020:StochSpeedLimit, Luispe2022SpeedLimits}, though the QFI flow plays a metrological (rather than dynamical) role in our work.

\paragraph*{Collisional QFI matrix.}
We now turn to the general multiparameter problem with multiple (potentially many-body) jump operators. We derive the QFI matrix flow for the collisional purification, obtaining a general expression that upper bounds the metrological performance of arbitrary parametric dissipative processes. For related information-geometric treatments of dissipative processes, see Refs.~\cite{Catana2015, Guta2017}.

Define the (disconnected) correlator
\begin{equation}
    \label{eq:disc-corr}
    \mathcal{C}_{ij}(t)\ \triangleq\ \Tr\!\big(L_i^\dagger L_j\,\rho(t)\big),
\end{equation}
with $\rho(t)$ the time-evolved state of Eq.~\eqref{eq:rho-t}, which may include (noiseless) ancillary degrees of freedom, and the dissipative amplitude
\begin{equation}
    \label{eq:dissip-amplitude}
    \zeta(\bm{\theta})\ \triangleq\ V\Lambda^{1/2},\qquad
    \Gamma=\zeta\zeta^\dagger,
\end{equation}
where $\Gamma\in\mathbb{C}^{R\times R}$ is the Hermitian dissipation matrix in the local jump basis $\{L_a\}$ [Eq.~\eqref{eq:gamma-matrix}].

\begin{restatable}[Collisional QFI Matrix]{thm}{QFIM}
    \label{thm:qfi-matrix}
    Consider the collisional purification (Sec.~\ref{sec:markov-evol}) of the dissipative process $\mathscr{D}_{\bm{\theta}}$ [Eqs.~\eqref{eq:dissip-phys}--\eqref{eq:dissip-can}]. Given the correlators $\mathcal{C}(t)$ in Eq.~\eqref{eq:disc-corr}, the collisional QFI matrix flow obeys
    \begin{equation}
        \label{eq:qfim-flow}
        (\dot{F}_Q)_{\alpha\beta}(t)
        = \Re\,\Tr\left(\mathcal{K}_{\alpha\beta}\,\mathcal{C}(t)\right),
    \end{equation}
    where the trace runs over the $R$-dimensional jump-index space and
    \begin{align}
        \label{eq:dissip-kernel}
        \mathcal{K}_{\alpha\beta}
        &\triangleq 4 \big(\partial_\beta \zeta\big)\big(\partial_\alpha \zeta\big)^\dagger \\
        &= \big(\partial_\beta \Gamma + \acomm{K_\beta}{\Gamma}\big)\,
        \Gamma^{-1}\big(\partial_\alpha \Gamma + \acomm{K_\alpha}{\Gamma}\big)^\dagger,
    \end{align}
    with eigenframe connection $K_\alpha\triangleq (\partial_\alpha V)V^\dagger$ in the local basis determined by $\{L_a\}$.
\end{restatable}

\begin{proof}[Proof sketch]
    The result follows from the collisional purification by expressing the QFI flow as a covariance of the joint system-bath generator associated with the collisional interaction Hamiltonian [Eq.~\eqref{eq:collisional-H}]. The linear-in-time structure again originates from the memoryless assumption, in direct analogy with Lemma~\ref{lem:qfi-flow}. See Appendix~\ref{app:thm-collision-qfi} for the full proof.
\end{proof}

For Hermitian operators, $L_i^\dagger=L_i$, we can tighten Theorem~\ref{thm:qfi-matrix} by replacing the disconnected correlator $\mathcal{C}$ with its connected version, obtained by taking ${L_i \to L_i - \Tr(L_i\rho(t))}$. We state this precisely in Proposition~\ref{prop:connected-qgt} of Appendix~\ref{app:thm-collision-qfi};
see also Ref.~\cite{Wang2026:PowerLaw}.

The dissipation kernel $\mathcal{K}$ captures metrological sensitivity arising from correlations in the parametric dissipator $\mathscr{D}_{\bm{\theta}}$, while the correlator $\mathcal{C}$ captures sensitivity arising from quantum correlations of the probe state. The separation of terms appearing in the kernel follows from the fundamental rate equation
\begin{equation}
    \label{eq:gamma-rate-eq}
    \partial_\alpha\Gamma = V(\partial_\alpha\Lambda)V^\dagger + \comm{K_\alpha}{\Gamma},
\end{equation}
which decomposes parameter dependence into eigenrate variations $\partial_\alpha\Lambda$ and eigenframe rotations $\comm{K_\alpha}{\Gamma}$; see Appendix~\ref{app:thm-collision-qfi} for details.

Finally, observe that Hamiltonian control $H_c(t)$ appears only through the state $\rho(t)$ per Eq.~\eqref{eq:rho-t}. Thus control can modify the constants in the bound by shaping $\rho(t)$, but it does not change the linear time scaling. The scaling with the jump-index dimension $R$ is controlled by the joint correlation structure entering both $\mathcal{K}$ and $\mathcal{C}$, which we quantify below through a generalized Heisenberg limit (Theorem~\ref{thm:gen-Heisenberg}).

\paragraph*{Generalized Heisenberg limit.}
We establish an upper bound on the average QFI (per parameter) in terms of the number of local jump operators $R$ and the connectivity of both the dissipation kernel $\mathcal{K}$ and the quantum correlator matrix $\mathcal{C}$. We first introduce notation that we use in the statement of our result (Theorem~\ref{thm:gen-Heisenberg}).

\begin{definition}[Parameter Averaging]
    \label{def:param-avg}
    Given a collection of matrices $\{M_{\alpha\beta}\}$ with $M_{\alpha\beta}\in\mathbb{C}^{m\times m}$ and parameter indices $\alpha,\beta=1,\dots,d$, we define the parameter-averaged matrix
    \begin{equation}
        \label{eq:param-avg}
        \bar{M}\ \triangleq\ \frac{1}{d}\sum_{\alpha=1}^d M_{\alpha\alpha}\ \in\ \mathbb{C}^{m\times m}.
    \end{equation}
\end{definition}
From which we define the average QFI  $\bar{F}_Q(t)$ ($m=1$) and the average dissipation kernel $\bar{\mathcal{K}}$ ($m=R$).

\begin{definition}[Max Norms]
    \label{def:max-norms}
    For any matrix $M\in\mathbb{C}^{R\times R}$, define the max-entry norm
    \begin{equation}
        \label{eq:max-entry}
        \norm{M}_{\max} \triangleq \max_{a,b\in[R]} \abs{M_{ab}}.
    \end{equation}
    Similarly, for a collection of operators $\{L_a\}_{a=1}^R$ on the probe Hilbert space, define
    \begin{equation}
        \label{eq:max-op}
        \norm{L}_{\max}\ \triangleq\ \max_{a\in[R]} \norm{L_a},
    \end{equation}
    where $\norm{\cdot}$ denotes the spectral norm.
\end{definition}

\begin{definition}[Support and Row Connectivity]
    \label{def:row-connect}
    For any matrix $M\in\mathbb{C}^{R\times R}$ define its support
    \begin{equation}
        \label{eq:supp}
        {\rm supp}(M)\ \triangleq\ \{(a,b)\in[R]^2:\ M_{ab}\neq 0\},
    \end{equation}
    and its row connectivity
    \begin{equation}
        \label{eq:row-connect}
        r_M\ \triangleq\ \max_{a\in[R]}\abs{\{b\in[R]:(a,b)\in{\rm supp}(M)\}}.
    \end{equation}
\end{definition}
Intuitively, ${\rm supp}(M)$ records which jump-index pairs $(a,b)$ contribute to, e.g., contractions such as $\Tr(\mathcal{K} \mathcal{C})$, while $r_M$ measures the maximal number of nonzero entries in any row (e.g., $r_{\mathcal{C}}$ quantifies the connectivity of probe correlations in jump-index space).

\begin{restatable}[Generalized Heisenberg Limit]{thm}{GHL}
    \label{thm:gen-Heisenberg}
    Consider a parametric dissipative encoding $\mathscr{D}_{\bm{\theta}}$ with $R$ distinct dissipation channels, and fix a local jump-operator basis $\{L_i\}_{i=1}^R$. For a sensing protocol consisting of $\nu$ i.i.d. repetitions of a single-shot experiment of duration $t$ and total sensing time $T=\nu t$, the average QFI (per parameter) satisfies
    \begin{equation}
        \label{eq:avg-qfi-bound}
        \bar{F}_Q(T) \leq
        T\, \norm{\bar{\mathcal{K}}}_{\max}\,\norm{L}_{\max}^2\,
        R\,\min(r_{\bar{\mathcal{K}}},\,r_{\mathcal{C}}),
    \end{equation}
    where $r_{\bar{\mathcal{K}}}$ denotes the row connectivity of $\bar{\mathcal{K}}$ and $r_{\mathcal{C}}\ \triangleq\ \sup_{s\in[0,t]} r_{\mathcal{C}(s)}$ denotes any uniform upper bound on the row connectivity of $\mathcal{C}(s)$ over a single shot. Ultimately, $\bar{F}_Q(T) \leq \order{TR^2}$ since $r_{\bar{\mathcal{K}}},r_{\mathcal{C}}\le R$.
\end{restatable}

\begin{proof}[Proof sketch.]
    The result derives from Theorem~\ref{thm:qfi-matrix} and the requirement that $\mathcal{K}$ and $\mathcal{C}$ must have overlapping support in the jump-index space (of linear dimension $R$) via $\Tr\,(\mathcal{K}_{\alpha\beta}^\top \mathcal{C}(t))$. See Appendix~\ref{app:thm-gen-heisenberg} for the full proof.
\end{proof}

\subparagraph*{Discussion.}
We distinguish scaling regimes by the connectivity pair $(r_{\bar{\mathcal{K}}}, r_{\mathcal{C}})$, which provides necessary structural conditions to attain $\order{TR^2}$ scaling:
\begin{enumerate}[label=(\roman*)]
    \item \emph{Sparse-sparse.} If $r_{\bar{\mathcal{K}}}=\order{1}$ and $r_{\mathcal{C}}=\order{1}$, then
    \begin{equation}
        \bar{F}_Q(T) \leq \order{TR},
    \end{equation}
    corresponding to uncorrelated dissipation probed by uncorrelated, independent sensors.

    \item \emph{Sparse-dense.} If $r_{\bar{\mathcal{K}}}=\order{1}$ but $r_{\mathcal{C}}=\order{R}$, then still
    \begin{equation}
        \bar{F}_Q(T) \leq \order{TR}.
    \end{equation}
    That is, strong probe correlations cannot overcome an effectively uncorrelated stochastic signal (no-go entanglement advantage).

    \item \emph{Dense-sparse.} If $r_{\bar{\mathcal{K}}}=\order{R}$ but $r_{\mathcal{C}}=\order{1}$, then again
    \begin{equation}
        \bar{F}_Q(T) \leq \order{TR}.
    \end{equation}
    That is, strong correlations in the stochastic signal do not yield Heisenberg scaling without correspondingly strong probe correlations.

    \item \emph{Dense-dense.} If $r_{\bar{\mathcal{K}}}=\order{R}$ and $r_{\mathcal{C}}=\order{R}$, then the bound is compatible with
    \begin{equation}
        \bar{F}_Q(T) \leq \order{TR^2}.
    \end{equation}
    Conversely, joint high-connectivity is necessary to even allow quadratic scaling in $R$; see Sections~\ref{sec:qsn-scaling} and~\ref{sec:super-heisenberg} for concrete examples.
\end{enumerate}
The SQL regimes in (i)-(iii) are consistent with scenarios where one estimates many independent rates in parallel, or where classical correlations in the stochastic signal can be leveraged through multiplexed readout. The Heisenberg-in-$R$ regime of (iv) is unique to noise metrology because enhancement requires high connectivity in both the quantum probes as well as the stochastic signal, as noted in Ref.~\cite{Brady2024:NoiseQSN} for a network of quantum sensors.

As of now, the jump-index dimension $R$ is treated abstractly. Physically, $R$ denotes the number of independent dissipation channels---equivalently, the number of linearly independent jump operators [Eqs.~\eqref{eq:dissip-phys}--\eqref{eq:dissip-can}]. Relating $R$ to the system size $N$ requires additional locality assumptions on the canonical jump operators in $\mathscr{D}_{\bm{\theta}}$, which we defer to Sec.~\ref{sec:super-heisenberg}.

\paragraph*{Eigenrate QFI.}
In much of this work, we consider the case where the parameters enter solely through the eigenrates $\gamma_k(\bm{\theta})$ of the dissipation matrix $\Gamma$, i.e., the eigenframe is parameter-independent, $\partial_\alpha V=0$, so that $K_\alpha=(\partial_\alpha V)V^\dagger=0$ in Theorem~\ref{thm:qfi-matrix}. This is a physically important setting, e.g., learning error rates, thermometry, decay/dephasing estimation etc., and it yields an explicit form of the QFI matrix flow.

\begin{proposition}[Eigenrate QCRB]
    \label{prop:rate-qfi}
    Suppose the parameters $\bm{\theta}$ are encoded solely in the eigenrates $\gamma_k$, so that $\partial_\alpha V=0$
    (and hence $K_\alpha=0$ in Theorem~\ref{thm:qfi-matrix}). Then
    \begin{equation}
        \mathcal{K}_{\alpha\beta}
        = (\partial_\beta \Gamma)\,\Gamma^{-1}\,(\partial_\alpha \Gamma).
    \end{equation}
    Let $\Gamma=V\Lambda V^\dagger$ with $\Lambda=\mathrm{diag}(\gamma_1,\dots,\gamma_R)$ and let $J_k$ denote the canonical
    jump operators [Eq.~\eqref{eq:can-jump}]. In this eigenbasis, the QFI matrix flow takes the explicit form
    \begin{equation}
        \label{eq:eigenrate-qfi}
        (\dot{F}_Q)_{\alpha\beta}(t)
        = 4\sum_{k=1}^R \big(\partial_\alpha \sqrt{\gamma_k}\big)\big(\partial_\beta \sqrt{\gamma_k}\big)\,
        \Tr\!\big(J_k^\dagger J_k\,\rho(t)\big).
    \end{equation}
    For a protocol of $\nu$ i.i.d.\ repetitions, each of duration $t$ (total time $T=\nu t$), the multiparameter QCRB implies the following precision bound on the diagonal precision $\tau_\alpha$ defined in Eq.~\eqref{eq:precision} (cf.~\cite{Sekatski2022:Thermometry}),
    \begin{equation}
        \label{eq:eigenrate-precision}
        \tau_\alpha \le T \sum_{k=1}^R \frac{(\partial_\alpha\gamma_k)^2}{\gamma_k}\norm{J_k}^2.
    \end{equation}
\end{proposition}
\begin{proof}
    The explicit form~\eqref{eq:eigenrate-qfi} follows from Theorem~\ref{thm:qfi-matrix}, while the precision bound~\eqref{eq:eigenrate-precision} follows from the facts $(\partial_\alpha\sqrt{\gamma_k})^2 = (\partial_\alpha\gamma_k)^2/\gamma_k$ and $\Tr\,(J_k^\dagger J_k\rho(s)) \leq \norm{J_k}^2$.
\end{proof}

Proposition~\ref{prop:rate-qfi} extends the single-parameter eigenrate bounds of Refs.~\cite{Sekatski2022:Thermometry, Gardner2025:FastReset} to the multiparameter setting and applies broadly to finite- and infinite-dimensional systems (spins and bosons).\footnote{Though for bosons, care must be taken with the operator-norm bound.} Moreover, for Hermitian jump operators, $J_k^\dagger=J_k$, Proposition~\ref{prop:rate-qfi} can actually be tightened by replacing $\Tr\,(J_k^\dagger J_k\,\rho(t))$ with the variance $\Var_{\rho(t)}(J_k)$; cf. Proposition~\ref{prop:connected-qgt} in Appendix~\ref{app:thm-collision-qfi}.

As of now, the physical setting and the canonical jump operators $J_k$ are left unspecified; they may be strictly local, collective (sums of local operators), or many-body operators. The structure and form inevitably determines the precision scaling with $R$. We discuss structural conditions, as well as concrete examples, in Sec.~\ref{sec:examples}. 

%==========
\subsection{Operational Scenarios}

To apply the collisional QCRB inferred from the QFI-matrix relations, we specify the operating scenarios relevant to experiments conducted over a long total sensing time $T=\nu t$, comprised of $\nu$ i.i.d.\ repetitions of a single-shot protocol of duration $t$. The scenarios below correspond to different experimental constraints and the resulting optimization.

\paragraph*{Scenario 1: Unlimited shots.}
Here the total sensing time $T$ is fixed and the number of repetitions $\nu=T/t$ can be taken large
by choosing small $t$ (subject to experimental overheads discussed below). The relevant figure of merit is the (average) QFI flow, $F_Q(t)/t$, so the time-optimal performance obeys
\begin{equation}
    \label{eq:F-unlim-shots}
    F_Q^{\rm opt}(T)
    = T\,\max_{t}\frac{F_Q(t)}{t}.
\end{equation}

\paragraph*{Scenario 2: Shot limited.}
If the number of repetitions $\nu$ is constrained, e.g., by slow state preparation or long measurement overhead, one instead optimizes the per-shot sensitivity,
\begin{equation}
    \label{eq:F-shot-lim}
    F_Q^{\rm opt}(\nu)
    = \nu\,\max_{t}F_Q(t).
\end{equation}
A subtlety is that the optimizing time $t$ can depend on the unknown parameters themselves, e.g., the rates $\gamma_k$, so this scenario is naturally addressed by adaptive protocols.

\paragraph*{Scenario 3: Short time / weak signal.}
In many practical settings, the single-shot duration is constrained to be small, $t=\delta t$, so that only the initial-time behavior is accessible; mathematically, $\delta t \ll 1/\big(\sum_k \gamma_k\norm{J_k}^2\big)$. In this regime,
\begin{equation}
    \label{eq:F-short-time}
    F_Q(T) \approx T\,\dot{F}_Q(0^+),
    \quad
    \dot{F}_Q(0^+)\triangleq \lim_{\delta t\to 0^+}\frac{F_Q(\delta t)}{\delta t},
\end{equation}
for $\nu$ repetitions ($T=\nu\delta t$). Thus, the initial QFI growth rate controls the attainable precision.

In this work, we focus on Scenarios~1 and~3, for which the central objective is to bound and attain a large Fisher-information rate, $F_Q(t)/t$. In the Markovian setting, the optimal operating point in both scenarios proves to be the continuous short-time limit, $\delta t \to 0^+$. An RPM strategy realizes this by rapidly repeating short-time experiments and tracking many quantum jumps in parallel, thereby attaining the eigenrate QCRB of Proposition~\ref{prop:rate-qfi}.

We note that Scenario~2 can exhibit qualitatively different behavior, e.g., exponential entanglement advantage~\cite{Wang2024:ExpSensing, Prabhu2026:ExpSensingQSA}, which we do not pursue here.

%==========
\subsection{Quickly Tracking Quantum Jumps}
\label{sec:rpm}

In the eigenrate-estimation setting (viz., Proposition~\ref{prop:rate-qfi}), the parameter dependence enters solely through the canonical rates $\gamma_k(\bm{\theta})$, so that the canonical jump operators $J_k$ are known and only the jump statistics depend on $\bm{\theta}$. The estimation problem in this regime is ``quasi-classical''~\cite{Matsumoto2002}: if we can resolve which jump channel occurred, then the measurement record becomes a multi-Poisson model whose Fisher information matches the collisional eigenrate QFI of Proposition~\ref{prop:rate-qfi} [Eq.~\eqref{eq:eigenrate-qfi}]. 

In this Section, we prove that an RPM protocol, which rapidly tracks many quantum jumps in parallel, reduces the problem to a multi-Poisson counting model in the sense above and saturates the parameter-averaged QCRB (Theorem~\ref{thm:rpm-qcrb}). This extends single-parameter, fast-reset strategies (cf.~\cite{Correa2015:OptThermometry, Sekatski2022:Thermometry, Gardner2025:FastReset, Wang2026:PowerLaw}) to the multiparameter setting and highlights the need to resolve \emph{each} jump channel, rather than a single or aggregate jump record which often suffices for single-parameter estimation. 

\paragraph*{Distinguishable jump basis.}
Fix a pure probe state $\ket{\psi}$ and define the normalized jump vectors
\begin{equation}
    \label{eq:jkfrak}
    \ket{\mathfrak{j}_k}\ \triangleq\ \frac{1}{\sqrt{\mu_k}}J_k\ket{\psi},\quad
    \mu_k\ \triangleq\ \expval{J_k^\dagger J_k}{\psi}.
\end{equation}
We say that $\ket{\psi}$ admits a \emph{distinguishable jump basis} if the vectors
$\{\ket{\mathfrak{j}_k}\}_{k=1}^R$ are mutually orthonormal and orthogonal to $\ket{\psi}$, i.e.
\begin{equation}
    \label{eq:jump-orth}
    \ip{\mathfrak{j}_k}{\mathfrak{j}_\ell}=\delta_{k\ell},
    \qquad
    \ip{\psi}{\mathfrak{j}_k}=0\ \ \forall\ k.
\end{equation}
(Note that $\ip{\psi}{\mathfrak{j}_k}=\expval{J_k}{\psi}=0$.) In this case, the projective POVM
\begin{equation}
    \label{eq:jump-povm}
    \bm{\mathcal{J}}\ \triangleq\ \Big\{\dyad{\psi},\ \dyad{\mathfrak{j}_1}, \ \dots, \ \dyad{\mathfrak{j}_R},\ \Pi_\perp\Big\},
\end{equation}
where $\Pi_\perp \triangleq I-\dyad{\psi}-\sum_{k=1}^R \dyad{\mathfrak{j}_k}$ represents projection outside the jump space, resolves whether a jump occurred (to leading order $\delta t$) and in which channel $k$. 

We emphasize that if the conditions in~\eqref{eq:jump-orth} do not hold on the system alone, they can often be enforced by enlarging the Hilbert space via noiseless ancillae.

\paragraph*{Rapid prepare-and-measure.}
The RPM protocol prepares $\psi$, evolves for a short time $\delta t$, measures using the POVM~\eqref{eq:jump-povm}, and repeats for $\nu$ i.i.d.\ shots (total time $T=\nu\delta t$). Here we show that, for fixed $T$ and in the limit $\delta t\to 0^+$, the resulting RPM Fisher information matches the collisional eigenrate QFI [Eq.~\eqref{eq:eigenrate-qfi} of Proposition~\ref{prop:rate-qfi}].

Assume eigenrate-only encoding, so that $J_k$ are parameter-independent and all parameter dependence enters through $\gamma_k(\bm{\theta})$. Within a single time-bin $\delta t$, the Lindblad evolution (without control, $H_c=0$) gives
\begin{multline}
\label{eq:short-time-rho}
    \rho(\bm{\theta};\delta t)
    = 
    \psi
     + \delta t\sum_{k=1}^R \gamma_k(\bm{\theta})
    \left(J_k\psi J_k^\dagger-\tfrac12\acomm{J_k^\dagger J_k}{\psi}\right) \\
    + \order{\delta t^2},
\end{multline}
where $\psi\triangleq\dyad{\psi}$. Using~\eqref{eq:jump-povm} and~\eqref{eq:jump-orth}, the single-bin detection probabilities are thus
\begin{align}
    p_k(\bm{\theta};\delta t)
     &= \delta t\,\gamma_k(\bm{\theta})\,\mu_k + \order{\delta t^2}, \label{eq:pk-shot} \\
    p_0(\bm{\theta};\delta t) &= 1-\sum_k p_k(\bm{\theta};\delta t)+\order{\delta t^2}, \label{eq:no-jump-p}
\end{align}
where $p_k$ is the $k$th jump probability and $p_0$ collects the no-jump outcomes, including $\ket{\psi}$ and $\Pi_\perp$ (cf.~\cite{Radaelli2026:QuJumpUnravel}).

\begin{restatable}[RPM Fisher Matrix]{lem}{}
    \label{lem:rpm-fisher}
    Assume eigenrate-only encoding and suppose $\ket{\psi}$ admits a distinguishable jump basis~\eqref{eq:jump-orth}.
    Consider the RPM protocol with bin size $\delta t$ and $\nu=T/\delta t$ i.i.d.\ shots (total time $T$). Then, in the limit
    $\delta t\to 0^+$ with $T$ fixed,
    \begin{equation}
        \label{eq:rpm-fisher}
        (F_{\rm rpm})_{\alpha\beta}(T)
        =
        4T\, \sum_{k=1}^R
        (\partial_\alpha\sqrt{\gamma_k})(\partial_\beta\sqrt{\gamma_k})\,\expval{J_k^\dagger J_k}{\psi}.
    \end{equation}
    Equivalently,
    \begin{equation}\label{eq:Frpm-limit}
        (F_{\rm rpm})_{\alpha\beta}(T)
        =
        T\,(\dot F_Q)_{\alpha\beta}(0^+),
    \end{equation}
    where $(\dot F_Q)_{\alpha\beta}(0^+)$ is the eigenrate QFI flow of Proposition~\ref{prop:rate-qfi} [Eq.~\eqref{eq:eigenrate-qfi}]
    evaluated at $\rho(0)=\psi$.
\end{restatable}
We provide a short intuitive proof below, based on multi-Poisson counting statistics of quantum jumps. % Also see Tsang's "Poisson Quantum Information"

\begin{proof}
    Use the single-bin jump outcomes from Eqs.~\eqref{eq:pk-shot}--\eqref{eq:no-jump-p}. Over $\nu=T/\delta t$ i.i.d.\ bins, the jump counts $N_k$ converge jointly (as $\delta t\to 0^+$ with $T$ fixed) to
    independent Poisson random variables with means
    \begin{equation}
        \lambda_k(\bm{\theta})=T\,\gamma_k(\bm{\theta})\,\mu_k.
    \end{equation}
    For independent Poisson counts, the classical Fisher information is additive over $k$ and equals
    \begin{align}
        (F_{\rm rpm})_{\alpha\beta}(T)
        &=
        \sum_{k=1}^R \frac{(\partial_\alpha \lambda_k)(\partial_\beta \lambda_k)}{\lambda_k} \\
        &=
        T\sum_{k=1}^R \frac{(\partial_\alpha\gamma_k)(\partial_\beta\gamma_k)}{\gamma_k}\,\mu_k \\
        &=
        4T\sum_{k=1}^R (\partial_\alpha\sqrt{\gamma_k})(\partial_\beta\sqrt{\gamma_k})\,\mu_k,
    \end{align}
    which proves Eq.~\eqref{eq:rpm-fisher}. Finally, Proposition~\ref{prop:rate-qfi} implies
    $(\dot F_Q)_{\alpha\beta}(0^+)=4\sum_k(\partial_\alpha\sqrt{\gamma_k})(\partial_\beta\sqrt{\gamma_k})\,\mu_k$,
    which proves Eq.~\eqref{eq:Frpm-limit}.
\end{proof}

Lemma~\ref{lem:rpm-fisher} demonstrates that, as $\delta t\to 0^+$, the RPM protocol reduces eigenrate estimation to a multi-Poisson counting model if $\ket{\psi}$ admits a distinguishable jump basis. Hence standard efficient estimators, such as the MLE, attain the classical Cram\'er-Rao bound asymptotically as $T\to\infty$~\cite{Kay1993, Helstrom1976}.

We now connect RPM to the physical setting involving a large sensing time $T$ and unlimited shots (\emph{Scenario 1}; see Eq.~\eqref{eq:F-unlim-shots}). As we show below, the optimal average (per parameter) QFI rate, $\max_{t}\sup_{\psi} \bar{F}_Q(t;\psi)/t$, coincides with the corresponding optimal RPM rate, and is attained by a RPM protocol in the limit $\delta t\to 0^+$. We point to Appendix~\ref{app:thm-rpm-qcrb} for a formal proof.

\begin{restatable}[RPM Achieves the Eigenrate QCRB]{thm}{RPM}
    \label{thm:rpm-qcrb}
    Assume eigenrate-only encoding (Proposition~\ref{prop:rate-qfi}) and define the optimal average (per parameter) QFI rate
    \begin{equation}
        \mathcal{R}^\star \triangleq \max_{t}\ \sup_{\psi}\ \frac{\bar{F}_Q(t;\psi)}{t},
    \end{equation}
    where $\bar{F}_Q(t)\triangleq\Tr F_Q(t)/d$ and $F_Q$ is the collisional eigenrate QFI matrix~\eqref{eq:eigenrate-qfi}.
    Then $\mathcal{R}^\star=\lambda_{\max}(A)$, with positive operator $A=\sum_{k=1}^R c_k J_k^\dagger J_k$ and $c_k = \sum_{\alpha=1}^d 4(\partial_\alpha\sqrt{\gamma_k})^2/d$. If a maximizing state $\ket{\psi^\star}$ admits a distinguishable jump POVM~\eqref{eq:jump-povm}, then RPM with $\psi^\star$ achieves $\mathcal{R}^\star$ as $\delta t\to 0^+$:
    \begin{equation}
        \lim_{\delta t\to 0^+}\frac{\bar{F}_{\rm rpm}(\delta t;\psi^\star)}{\delta t}=\mathcal{R}^\star.
    \end{equation}
    Consequently, in the unlimited-shots scenario~\eqref{eq:F-unlim-shots}, RPM saturates the optimal eigenrate QCRB per parameter.
\end{restatable}

A notable feature of the RPM protocol is that optimality does not require coherent control during sensing. Rather, control enters only through our ability to prepare an optimal probe state, measure in the jump basis, and rapidly repeat the protocol.

Finally, we note that Theorem~\ref{thm:rpm-qcrb} straightforwardly extends to any quadratic form $F_Q[u]=\sum_{\alpha,\beta} u_\alpha u_\beta (F_Q)_{\alpha\beta}$, not just the per-parameter average.

%==========
%==========
\section{Examples}
\label{sec:examples}

In this Section, we illustrate how the QFI bounds apply to an array of tasks in quantum noise metrology. Across all settings, the Markovian assumption enforces SQL scaling in the total sensing time $T$, while nontrivial enhancements arise from (i) how many dissipative channels are available \emph{and} identifiable, and (ii) how densely the parameter dependence and probe correlations are spread across those channels. We focus on four representative use-cases: Sec.~\ref{sec:qsn-scaling} analyzes Heisenberg-in-$N$ scaling in networked quantum sensors; Sec.~\ref{sec:super-heisenberg} analyzes super-Heisenberg scaling when collective $K$-body dissipation channels are present; Sec.~\ref{sec:pauli-learning} analyzes learning Pauli noise rates as a high-dimensional multiparameter sensing problem; and Sec.~\ref{sec:subdiff} analyzes subdiffraction incoherent imaging as a joint eigenrate/eigenframe estimation problem.

%==========
\subsection{Networked Quantum Noise Metrology}

\label{sec:qsn-scaling}

A quantum sensor network consists of $N$ spatially distributed sensors, each coupled locally to a ``sample'' (e.g., condensed matter system) or classical field~\cite{Proctor2018:qsn, Zhuang2018:dqs, Eldredge2018:SecureQSNs, Ge2018:qsnLinNetworks,Zhang2021:DQSrvw}. For stochastic signals, the local jump operators $\{L_i\}_{i=1}^R$ represent the \emph{single-sensor} (single-body) operators and $\Gamma(\bm{\theta})$ encodes spatial correlations of the stochastic signal across the sensor array. In this case, the number of dissipation channels (determined by the rank of $\Gamma$) scales as $R= q_1 N$, where $q_1$ is the number of distinct single-body operators per sensor (see Sec.~\ref{sec:super-heisenberg}); we take $q_1=1$ below for clarity. Theorem~\ref{thm:gen-Heisenberg} then suggests that $N^2$ scaling is permissible, in principle, provided the dissipation kernel and quantum correlator matrix of the network are sufficiently dense.

\begin{corollary}[Quantum Sensor Network Scaling]
\label{cor:qsn-scaling}
    Consider a quantum sensor network of $N$ sensors with single-body jump operators $\{L_i\}_{i=1}^N$ (so $R=N$). Assume that, for the parameter family of interest, the dissipation kernel is dense across channels (row-connectivity $r_{\bar{\mathcal{K}}}=\Theta(N)$) and that the probe state realizes dense jump correlations (row-connectivity $r_{\mathcal{C}}=\Theta(N)$). Then for total sensing time $T=\nu t$, there exists a constant $c$ independent of $N$ and $T$, such that the average precision (per parameter) satisfies (cf.~\cite{Brady2024:NoiseQSN})
    \begin{equation}
    \label{eq:FQ-scaling-K2kappa}
      \bar{\tau}\le c\,TN^2.
    \end{equation}
\end{corollary}

\begin{proof}[Proof sketch.]
    For single-body operators, $R = N$. Under connectivity assumptions, Theorem~\ref{thm:gen-Heisenberg} indicates that $\bar{\tau}\le c\,TR^2=c\,TN^2$ with a constant $c$ absorbing operator-norm and kernel prefactors.
\end{proof}

The $N^2$ scaling originates from the joint effect of (i) spatial correlations in the stochastic signal, encoded in $\Gamma(\bm{\theta})$, and (ii) quantum correlations between sensors, encoded in the correlator $\mathcal{C}$, as emphasized in Ref.~\cite{Brady2024:NoiseQSN}. Networked quantum noise metrology arguably provides one of the most experimentally relevant settings for quantum enhancements, e.g., via superdecoherence~\cite{Monz2011:Superdecoherence, Shammah2018:CollectiveDecoherence}, when interrogating many-body noise correlations~\cite{Rovny2025:nvEntResource,Zhou2025:nvEntResource, Wang2026:PowerLaw}, and distributed sensing of stochastic forces and fields (including waveform estimation~\cite{Tsang2011:Waveform, Ng2016:WaveformSpectrum, Gardner2025:Waveform}) with bosonic degrees of freedom~\cite{Brady2022:cavityDMSearch, Brady2023:OmechArrayDMsearch, Xia2023:OmechDQS, Guo2020:dqsCVnetwork, Gilmore2021:IonEFieldQSN}.

Corollary~\ref{cor:qsn-scaling} provides an upper bound but does not, by itself, guarantee achievability. In the important class of eigenrate-estimation problems, however, Theorem~\ref{thm:rpm-qcrb} provides a sufficient route to saturation whenever the canonical jump operators admit a probe with a distinguishable jump basis. We illustrate this with a simple example involving collective spin dissipation.

\paragraph*{Example: Collective spin dissipation.}
Consider collective spin dissipation along each direction $(x,y,z)$, 
\begin{equation}
    \mathscr{D}_{\bm{\theta}}(\rho) = \sum_{\alpha\in\{x,y,z\}} \gamma_\alpha \left(S_\alpha \rho S_\alpha - \tfrac{1}{2} \acomm{S_\alpha^2}{\rho}\right),
\end{equation}
with canonical jump operators
\begin{equation}
    \label{eq:global-spin}
    S_\alpha \triangleq \sum_{i=1}^N \sigma_i^{(\alpha)},
    \qquad \alpha\in\{x,y,z\}.
\end{equation}
We now argue that we can estimate the rates $\bm{\theta}=(\gamma_x,\gamma_y,\gamma_z)$ in parallel, with precision scaling $\Theta(N^2)$ for each using (isotropic, quasi-classical) 3D quantum compass solutions~\cite{Vasilyev2024:QuCompass}. 

The (isotropic) 3D quantum compass $\ket{\psi_\circledast}$ has uniform collective-spin fluctuations and vanishing means, $\expval{S_\alpha}{\psi_\circledast}=0$~\cite{Vasilyev2024:QuCompass}. Consequently, the off-diagonal symmetrized correlators vanish, $\expval{\{S_\alpha,S_\beta\}}{\psi_\circledast}=0$ for $\alpha\neq\beta$, which together with $\comm{S_\alpha}{S_\beta}= 2i\epsilon_{\alpha\beta\gamma}S_\gamma$ implies
\begin{equation}
    \expval{S_\alpha S_\beta}{\psi_\circledast}=0
    \quad (\alpha\neq\beta).
\end{equation}
Thus the states $\{\ket{\psi_\circledast},S_x\ket{\psi_\circledast},S_y\ket{\psi_\circledast},S_z\ket{\psi_\circledast}\}$ are mutually orthogonal, i.e., $\ket{\psi_\circledast}$ admits a distinguishable jump basis in the sense of Eqs.~\eqref{eq:jump-orth} and~\eqref{eq:jump-povm}. Moreover, the quantum compass has collective spin fluctuations
\begin{equation}
    \expval{S_\alpha^2}{\psi_\circledast}=\Theta(N^2),
\end{equation}
so an RPM protocol with a 3D quantum compass achieves $\Theta(TN^2)$ precision scaling for each $\gamma_\alpha$.

A GHZ probe is not optimal for this multiparameter task, since it is aligned along a preferred axis and therefore suppresses sensitivity in the longitudinal direction; in that case multiple probe settings may be required, with the GHZ aligned along a different axis in each setting. By contrast, separable probes can estimate the rates simultaneously but only with scaling $\Theta(TN)$.

%==========
%==========
\subsection{Super-Heisenberg Scaling}
\label{sec:super-heisenberg}

Suppose the jump operators consist of many-body terms, e.g., $k$-local operators, so that the channel dimension $R$ grows superlinearly with the system size $N$. Moreover, suppose the canonical jump operators are genuinely collective, i.e., sums of $k$-local operators. In this setting, the Heisenberg-in-$R$ scaling from Theorem~\ref{thm:gen-Heisenberg} translates into super-Heisenberg scaling in $N$, analogous in spirit to unitary metrology with nonlinear collective generators~\cite{Boixo2007:QuEstim, Roy2008:ExpMetrology}.

\begin{corollary}[Super-Heisenberg Scaling]
    \label{cor:super-heisenberg}
    Consider a probe of $N$ subsystems. Suppose the Lindbladian includes jump operators that act
    nontrivially on at most $K$ subsystems, and that for each support size $k\le K$ there are
    $q_k$ independent operators per support. Then the number of dissipation channels obeys
    \begin{equation}
    \label{eq:R-scaling}
      R \le \sum_{k=1}^{K} q_k \binom{N}{k}
      = \Theta\!\left(N^K\right)\quad(N\gg 1,\;K\;\text{fixed}).
    \end{equation}
    Moreover, if the dissipation kernel and correlator are dense across the active channels ($r_{\bar{\mathcal{K}}}, \, r_{\mathcal{C}}=\Theta(R)$), then there exists a
    constant $c$ such that
    \begin{equation}
      \bar{\tau}\le c\,TR^2
      = \Theta\!\left(TN^{2K}\right).
    \end{equation}
\end{corollary}

\begin{proof}
    The bound \eqref{eq:R-scaling} follows by counting supports and the number of independent operators per support. Under the connectivity assumptions, Theorem~\ref{thm:gen-Heisenberg} indicates that $\bar{\tau}\le c\,T\,R^2$, and substituting \eqref{eq:R-scaling} yields the claim.
\end{proof}

For finite-dimensional subsystems of local dimension $d_{\rm loc}$, the space of traceless operators on $k$ subsystems has dimension $d_{\rm loc}^{2k}-1$. Hence, for fixed $k$, a Lindblad dissipator can involve at most $d_{\rm loc}^{2k}-1$ linearly independent traceless jump operators, so $q_k\le d_{\rm loc}^{2k}-1$ (up to physical restrictions on the allowed jump algebra).

For bosonic modes, the local operator space is infinite-dimensional; a finite $q_k$ arises by restricting the operator class. For instance, if each $k$-mode jump operator is a linear combination of normal-ordered monomials in $\{a_j,a_j^\dagger\}$ of bounded total degree $\le r$, then $q_k$ is finite and scales as $q_k=\order{r^{2k}}$. Here $k$ counts the number of modes on which the operator acts, while $r$ bounds the polynomial degree. For instance, single-photon loss has $(k,r)=(1, 1)$, single-mode two-photon loss has $(k,r)=(1, 2)$, Kerr-type processes have $r=2$, and higher-order nonlinear dissipation has $r>2$. Combining this with \eqref{eq:R-scaling} indicates super-Heisenberg bounds with exponent set by the locality $K$, up to local factors $q_k$. Finally, we remark that if $K$ grows with $N$, then $R$ can, in principle, scale exponentially with $N$~\cite{Roy2008:ExpMetrology, Boixo2007:QuEstim}.

Regarding achievability of Corollary~\ref{cor:super-heisenberg}, it is useful to consider eigenrate-estimation problems in which collective channels are resolved individually. In this setting, Lemma~\ref{lem:rpm-fisher} and Theorem~\ref{thm:rpm-qcrb} imply that, when the canonical jump operators admit a probe with a distinguishable jump basis, an RPM protocol attains the corresponding eigenrate QCRB. For collective $k$-body channels whose jump amplitudes scale extensively with system size, this yields super-Heisenberg scaling $\Theta(TN^{2k})$, as we now demonstrate with an example.

\paragraph*{Example: Collective multipole dissipation.}
We extend the collective spin dissipation example to collective multipole dissipation, which showcases a concrete realization of super-Heisenberg scaling. This provides a family of collective channels for which we can, in principle, estimate $(N+1)^2$ dissipation rates in parallel, with rank-$k$ multipoles admitting super-Heisenberg scaling $\Theta(N^{2k})$.

Consider $N$ spins with collective operators $S_\alpha$ [Eq.~\eqref{eq:global-spin}] and $S_\pm \triangleq S_x \pm i S_y$. The fully symmetric subspace carries spin quantum number $s=N/2$, and the collective irreducible spherical tensors, $T_q^{(k)}$, on this sector have multipole order $k=0,1,\dots,N$ and components $q=-k,\dots,k$; see Section 15.3 of Shankar~\cite{Shankar2012QM}. These operators correspond to degree-$k$ polynomials in the collective spin operators $S_z$ and $S_{\pm}$. There are $\sum_{k=0}^N (2k+1)=(N+1)^2$ tensor operators in total. For example, up to normalization conventions,
\begin{equation}
    T_0^{(0)}= S^2, \quad T_0^{(1)} = S_z,
    \quad
    T_{\pm1}^{(1)} = \mp \frac{1}{\sqrt2} S_\pm,
\end{equation}
and
\begin{align}
    T_0^{(2)} &= \frac{1}{\sqrt6}\left(3S_z^2 - S^2\right), \\
    T_{\pm1}^{(2)} &= \mp \frac12 \left(S_z S_\pm + S_\pm S_z\right), \\
    T_{\pm2}^{(2)} &= \frac12 S_\pm^2,
\end{align}
where $S^2 = S_x^2 + S_y^2 + S_z^2$.

For the moment, we take the non-scalar tensors ($k\neq 0$) as the canonical jump operators of the parametrized dissipator; we include the scalar tensor $T_0^{(0)}$ thereafter. Physically, this model describes an ensemble of identical spins subject to genuinely collective, rotationally structured dissipation~\cite{Shammah2018:CollectiveDecoherence}. It should therefore be understood as an effective description of dissipative dynamics on the symmetric $s=N/2$ manifold, which may arise from coupling to a common reservoir, cavity-mediated decay, or engineered reservoir interactions.

The tensors furnish a symmetry-adapted operator basis for collective multipole relaxation, which spans the operator space on the fully symmetric ($s=N/2$) subspace~\cite{Omanakuttan2024:SpinCat}. We thus assume the dissipator diagonalizes in this basis,
\begin{equation}
    \mathscr{D}_{\bm{\theta}}(\rho)
    =
    \sum_{k=1}^{N}\sum_{q=-k}^{k}
    \gamma_q^{(k)}
    \left(
        T_{q}^{(k)} \rho T_{q}^{(k)\dagger}
        - \tfrac{1}{2}\acomm{T_{q}^{(k)\dagger}T_{q}^{(k)}}{\rho}
    \right),
\end{equation}
where $\gamma_q^{(k)}=\gamma_q^{(k)}(\bm{\theta})$ denotes the parametrized multipole dissipation rates. For concreteness, we take the estimands to be
\begin{equation}
    \bm{\theta}=(\gamma_q^{(k)}), \qquad 1\le k\le N,\ \ -k\le q\le k,
\end{equation}
so that there are $(N+1)^2-1$ parameters to estimate not including $\gamma_0^{(0)}$, which we eventually circle back to.

We now construct a method to simultaneously estimate the multipole rates. Consider the symmetric Dicke states $\{\ket{D_n^N}\}_{n=0}^N$~\cite{Toth2012:dicke}, which satisfy
\begin{equation}
    S_z \ket{D_n^N} = (N-2n)\ket{D_n^N},
\end{equation}
where $n$ counts the number of spins in the $\ket{\downarrow}$ state. Dicke states form a basis of the fully symmetric subspace, with normalized projector
\begin{equation}
    \label{eq:Psym}
    P_{\rm sym} \triangleq \frac{1}{N+1}\sum_{n=0}^N \dyad{D_n^N}.
\end{equation}
On this sector, the tensor operators are orthogonal,
\begin{equation}
    \label{eq:T-orthogonal}
    \Tr\!\left(P_{\rm sym}\,T_q^{(k)\,\dagger}T_{q'}^{(k')}\right)
    =
    C^{(k)}(N)\,\delta_{qq'}\delta_{kk'},
\end{equation}
with $C^{(k)}(N)=\Theta(N^{2k})$ for $k\neq 0$, while for the scalar operator $T_0^{(0)}=S^2$, we have $C^{(0)}=\Theta(N^4)$. Moreover,
\begin{equation}
    \label{eq:T-trace}
    \Tr\!\left(P_{\rm sym}\,T_q^{(k)}\right)=0
    \qquad (k\neq 0).
\end{equation}
We make use of these properties to construct an orthogonal jump basis in the sense of Eqs.~\eqref{eq:jump-orth} and~\eqref{eq:jump-povm}. 

To do so, we introduce an ancillary system of dimension $N+1$; for instance, if the ancilla consists of $L$ spins, then we require only $L=\lceil\log_2(N+1)\rceil$ ancillary spins. Define the Dicke-Choi state (cf.~\cite{Brady2024:NoiseQSN}),
\begin{equation}
    \ket{\psi_{\rm sym}}
    \triangleq
    \frac{1}{\sqrt{N+1}}
    \sum_{n=0}^N
    \ket{n}_{\rm anc}\ket{D_n^N}_{\rm sys},
\end{equation}
so that the reduced system state is the normalized projector onto the symmetric subspace,
\begin{equation}
    \Tr_{\rm anc}(\dyad{\psi_{\rm sym}})=P_{\rm sym}.
\end{equation}
From Eqs.~\eqref{eq:T-orthogonal} and~\eqref{eq:T-trace}, it follows that $\ket{\psi}_{\rm sym}$ and the jump vectors
\begin{equation}
    \ket*{\mathfrak{j}_q^{(k)}}
    \triangleq
    (C^{(k)}(N))^{-1/2}\,
    (I_{\rm anc}\otimes T_q^{(k)})\ket{\psi_{\rm sym}},
\end{equation}
form a distinguishable jump basis. Therefore, by Lemma~\ref{lem:rpm-fisher}, an RPM protocol with the Dicke-Choi probe can estimate the $(N+1)^2-1$ rates $\gamma_q^{(k)}$ ($k\neq 0$) in parallel. Furthermore, since $\expval*{T_q^{(k)\dagger} T_q^{(k)}}{\psi_{\rm sym}}=C^{(k)}(N)=\Theta(N^{2k})$, the achievable precision for each parameter $\gamma_q^{(k)}$ scales as
\begin{equation}
    \tau(\gamma_q^{(k)})=\Theta(TN^{2k})
    \qquad
    (k\neq 0).
\end{equation}
The prefactor depends on the $\order{1}$ constants in $C^{(k)}(N)$ together with a factor of $1/\gamma_q^{(k)}$ per Eq.~\eqref{eq:rpm-fisher}.

We now include the scalar relaxation rate $\gamma_0^{(0)}$ generated by $T_0^{(0)}=S^2$. To make this channel identifiable, the probe must have support on more than one total-spin sector, equivalently $\Var(S^2)>0$. Take $N$ even for simplicity, and let $\ket{\mathfrak{s}}_{\rm sys}$ denote a spin singlet, so that
\begin{equation}
    T_q^{(k)}\ket{\mathfrak{s}}_{\rm sys}=0
    \qquad \forall\ k.
\end{equation}
Introduce an extra ancillary state $\ket{\chi}_{\rm anc}$ orthogonal to $\{\ket{n}_{\rm anc}\}_{n=0}^N$, and let
\begin{equation}
    \ket*{\widetilde{\psi}_{\rm sym}}
    =
    \frac{1}{\sqrt2}
    \left(
        \ket{\psi_{\rm sym}}
        +
        \ket{\chi}_{\rm anc}\ket{\mathfrak{s}}_{\rm sys}
    \right).
\end{equation}
For $k\neq 0$, the non-scalar jump states remain unchanged. We handle the scalar component by introducing the centered jump vector
\begin{equation}
    \ket*{\mathfrak{j}_0^{(0)}}
    \triangleq
    \frac{I_{\rm anc}\otimes \Delta T_0^{(0)}\ket*{\widetilde{\psi}_{\rm sym}}}
    {\sqrt{\expval*{(I_{\rm anc}\otimes\Delta T_0^{(0)})^2}_{\widetilde{\psi}_{\rm sym}}}},
\end{equation}
where $\Delta T_0^{(0)} \triangleq T_0^{(0)}-\expval*{T_0^{(0)}}_{\widetilde{\psi}_{\rm sym}}$. Since $T_0^{(0)}=S^2=N(N+2)$ on the symmetric sector and $S^2=0$ on the singlet, we have
\begin{equation}
    \expval*{(I_{\rm anc}\otimes\Delta T_0^{(0)})^2}_{\widetilde{\psi}_{\rm sym}}
    =
    \Theta(N^4).
\end{equation}
Thus the rate $\gamma_0^{(0)}$ is also estimable with precision
\begin{equation}
    \tau(\gamma_0^{(0)})=\Theta(TN^4).
\end{equation}
This centered construction is consistent with the Hermitian-jump refinement of the QFI matrix, in which the disconnected correlator is replaced by the connected one (i.e., the variance); see the discussion following Theorem~\ref{thm:qfi-matrix} and Proposition~\ref{prop:connected-qgt} of Appendix~\ref{app:thm-collision-qfi}.

Altogether, this protocol enables estimation of the $(N+1)^2$ collective dissipation rates in parallel, with precision scaling $\Theta(TN^{2k})$ for the $k\neq 0$ multipoles and $\Theta(TN^4)$ for the scalar $k=0$ channel.

It is plausible to estimate all multipole dissipation rates in parallel with separable probes. However, due to the independence of the local probes, the estimation precision of $\gamma_q^{(k)}$ scales at best as $\order{TN^k}$ for each $k$, indicating substantially reduced precision relative to the Dicke-Choi construction. %More quantitatively, for a fixed multipole rank $k$, the collective jump $T_q^{(k)}$ has support on $R_k=\Theta(N^k)$ $k$-body operators in the Pauli basis. However, for separable probes, the correlator remains sparse in that basis, with row connectivity $r_{\mathcal C}=\order{1}$, and Theorem~\ref{thm:gen-Heisenberg} implies the stated scaling.

%==========
\subsection{Learning Pauli Noise Rates}
\label{sec:pauli-learning}

In this Section, we consider estimation of Pauli noise rates for the Pauli noise channel~\cite{Fujiwara2003:PauliNoise, Flammia2020:PauliLearn}, a fundamental task in characterization and benchmarking of quantum devices~\cite{Harper2020:PauliLearn, Hashim2021:RandCompil, Van2023:PEC, Hashim2025:BenchmarkingRvw, Proctor2025:Benchmarking}. This example is somewhat atypical from the viewpoint of quantum metrology: there is no Heisenberg-type scaling enhancement because the noise acts independently across Pauli channels (Theorem~\ref{thm:gen-Heisenberg}). Nevertheless, there is an \emph{exponential} separation in learning complexity once we constrain quantum resources---viz., learning with quantum memory versus without~\cite{Chen2022:ExpSepMemory, Chen2022:PauliChEst, Chen2024:TightPauliLearn, Caro2024:PTMlearning, Kim2025:LearnResources}.

Estimating (or learning) Pauli noise rates is a canonical instance of eigenrate estimation (cf.\ Proposition~\ref{prop:rate-qfi}): The channel structure is known (diagonal in the Pauli basis), while the Pauli error rates are unknown. Moreover, the RPM protocol achieves the QCRB for this task (Lemma~\ref{lem:rpm-fisher} and Theorem~\ref{thm:rpm-qcrb}) provided the Pauli channels are distinguishable, as we show below.

Before proceeding, we clarify two points of terminology. Throughout this section, we use the term \emph{learning} to emphasize reconstruction of an unstructured, high-dimensional noise model, although we formulate our analysis in metrological terms through quantum estimation theory rather than statistical learning theory~\cite{Chen2024:TightPauliLearn}. In the asymptotic small-error regime relevant here, these viewpoints largely align, since the inverse Fisher information matrix controls local learning complexity~\cite{Kwon2026:FisherComplexity}. We also fix the resource terminology: Specifically, we refer to a \emph{quantum memory} here as noiseless ancillary qubits that may be initially entangled with the system qubits and measured jointly with them, but do not themselves undergo the parametrized dissipative dynamics. By contrast, \emph{without quantum memory} means that no entangled ancillae are available, although classical registers and adaptive protocols are allowed. %\footnote{In the RPM protocol with quantum memory, the measurement record reduces exactly to independent Poisson counts. Consequently, standard Poisson concentration yields finite-sample $(\epsilon, \delta)$ guarantees. Thus the Fisher-information perspective upgrades directly to a fixed-confidence learning statement.}

Consider an $N$-qubit probe state $\psi\triangleq \dyad{\psi}$ subject over a short time $\delta t$ to the Pauli channel with parameters $\bm r=\{r_a\}$,
\begin{equation}
    \mathcal{E}_{\bm r}(\psi)
    =
    \left(1-\sum_{a=1}^{R} r_a\right)\psi
    +\sum_{a=1}^{R} r_a\, P_a \psi P_a,
\end{equation}
where $\{P_a\}_{a=1}^{R}$ are the non-identity Pauli strings and $R=4^N-1$.
In the Markovian regime, this channel arises from a Lindblad generator with jump operators $J_a=P_a$ and rates $\gamma_a$, so that $r_a=\gamma_a\,\delta t + \order{\delta t^2}$. Hence we focus on weak noise estimation (cf. the seminal work of~\cite{Flammia2020:PauliLearn}), though the main conclusions do not rely on this assumption~\cite{Chen2022:PauliChEst, Chen2024:TightPauliLearn}. We take the estimands as $\bm{\theta}=(\gamma_a)_{a=1}^R$.

\paragraph*{Complexity heuristics.}
We provide a heuristic complexity analysis based on Projection-Valued Measures (PVMs), which makes the underlying dimensional obstruction for learning without quantum memory transparent. This PVM-based argument serves as an intuitive proxy for the full problem. However, the exponential complexity barrier persists under more refined analysis and general measurements: Appendix~\ref{app:pauli-noise} proves that the same conclusion holds for arbitrary no-memory protocols in the short-time regime, including arbitrary POVMs supplemented by classical memory~\cite{Chen2024:TightPauliLearn}; cf.~\cite{Tu2025:LearnMixedU, Kwon2026:FisherComplexity} for related Fisher-information based analysis.

We quantify learning complexity through the average absolute precision, $\bar{\tau}$ (see Definition~\ref{def:param-avg}). For a protocol with total Fisher information matrix $F_{\rm tot}$ accumulated over total sensing time $T$, the multiparameter Cram\'er-Rao bound implies
\begin{equation}
    \label{eq:prec-trace-bound}
    \bar{\tau}\le \frac{1}{R}\Tr(F_{\rm tot}).
\end{equation}
Indeed, given $\tau_a = 1/\bm{\Sigma}_{aa}$ and using $\bm{\Sigma} \geq F_{\rm tot}^{-1}$ together with $(F_{\rm tot}^{-1})_{aa}\ge 1/(F_{\rm tot})_{aa}$, it follows that $\tau_a \le (F_{\rm tot})_{aa}$ and hence $\bar{\tau}\le \Tr(F_{\rm tot})/R$.

To upper bound $\Tr(F_{\rm tot})$, decompose the protocol into $\nu$ experimental configurations indexed by $j$, each run for time $t_j$ with $T=\sum_{j=1}^\nu t_j$. Let $F^{(j)}$ denote the Fisher matrix of configuration $j$, so that $F_{\rm tot}=\sum_{j=1}^\nu F^{(j)}$ and $\Tr(F_{\rm tot})=\sum_{j=1}^\nu \Tr(F^{(j)})$.

Estimating Pauli error rates is an eigenrate estimation problem and Proposition~\ref{prop:rate-qfi} implies a uniform spectral bound on the QFI rate. Since $F^{(j)}\leq F^{(j)}_{Q}$, we have
\begin{equation}
    \lambda_{\max}\left(F^{(j)}\right)
    \leq \lambda_{\max}\left(F^{(j)}_{Q}\right)
    \leq t_j\,\gamma_{\min}^{-1},
\end{equation}
where $\gamma_{\min}\triangleq \min_a \gamma_a$ and $\|J_a\|=\|P_a\|=1$ for Pauli jumps. Therefore,
\begin{equation}
    \Tr\,(F^{(j)})
    \leq \lambda_{\max}(F^{(j)})\,\rank(F^{(j)})
    \le t_j\gamma_{\min}^{-1}\,\rank(F^{(j)}).
\end{equation}

Suppose that each experiment ends with a \emph{projective} measurement on the available quantum Hilbert space $\mathscr{H}$ of dimension $D\triangleq \abs{\mathscr{H}}$. Then each configuration has at most $D$ outcomes, so $\rank(F^{(j)})\le D-1$. Combine these bounds to obtain
\begin{equation}
    \Tr(F_{\rm tot})
    \le \sum_{j=1}^\nu t_j \gamma_{\min}^{-1}(D-1)
    = T\gamma_{\min}^{-1}(D-1),
\end{equation}
and substitute into Eq.~\eqref{eq:prec-trace-bound} to reckon\footnote{This bounds the absolute precision; a bound on the (dimensionless) average relative precision reads $\bar{\tau}_{\rm rel} \leq T\gamma_{\rm max} (D-1)/R$.}
\begin{equation}
    \label{eq:prec-d-bound}
    \bar{\tau}\le T\gamma_{\min}^{-1}\,\frac{D-1}{R}.
\end{equation}

\paragraph*{Learning without quantum memory.}
Without quantum memory and assuming projective measurements, the available Hilbert space is $\mathscr{H}=\mathscr{H}_S$ with $D_S=2^N$, while $R=4^N-1$. Thus Eq.~\eqref{eq:prec-d-bound} yields
\begin{equation}
    \bar{\tau}\le T\gamma_{\min}^{-1}\,\frac{2^N-1}{4^N-1}
    = \Theta\!\left(T\gamma_{\min}^{-1}\,2^{-N}\right).
\end{equation}
Equivalently, achieving a target $\bar{\tau}$ without a quantum memory requires an exponential sensing time,
\begin{equation}
    T_{\text{no-mem}} \ge \Omega\!\left(\bar{\tau}\gamma_{\min}2^N\right),
\end{equation}
consistent with tight complexity bounds~\cite{Chen2024:TightPauliLearn, Kim2025:LearnResources}.

Allowing a large POVM $\{E_x\}_{x=1}^{m}$ with $m\gg D$ can, in principle, render $F_{\rm tot}$ full rank using fewer measurement settings. However, when only system degrees of freedom and classical registers are available, these measurements do not circumvent the exponential learning barrier; see Appendix~\ref{app:pauli-noise} for proof. Intuitively, although a large POVM increases the number of outcomes, the resulting information is spread across exponentially many Pauli directions, so that $F_{\rm tot}$ cannot provide uniformly strong information about all parameters. Equivalently, the root-score vectors $\{\nabla_{\bm{\theta}}\sqrt{p_x}\}_x$ [Eq.~\eqref{eq:fisher-metric}] cannot resolve all directions in parameter space with comparable strength~\cite{Tu2025:LearnMixedU, Kwon2026:FisherComplexity}.

This reflects a fundamental tradeoff when estimating an exponentially large number of parameters without quantum memory: either (i) we use distinguishable measurements (PVMs), which provide well-conditioned Fisher information on a low-dimensional subspace per measurement setting and therefore require exponentially many settings to identify the full model, or (ii) we use a large POVM to make $F_{\rm tot}$ full rank with fewer settings, in which case the information per direction becomes diluted across the exponentially large parameter space, again resulting in an exponential sample complexity~\cite{Chen2024:TightPauliLearn, Kwon2026:FisherComplexity}.

\paragraph*{Learning with quantum memory.}
To circumvent this complexity barrier, we introduce a quantum memory~\cite{Chen2022:ExpSepMemory, Chen2022:PauliChEst}. Specifically, consider an $N$-qubit noiseless quantum memory $\mathscr{H}_{\rm mem}$, which evolves trivially, and prepare Bell pairs between system and memory, $\ket{\Phi_N}\triangleq \ket{\Phi}^{\otimes N}$ with $\ket{\Phi}=(\ket{00}+\ket{11})/\sqrt{2}$. This enlarges the available Hilbert space to $\mathscr{H}=\mathscr{H}_S\otimes\mathscr{H}_{\rm mem}$ with $D_{\rm mem}=4^N$, and indeed removes the exponential sampling overhead~\cite{Chen2022:PauliChEst, Chen2024:TightPauliLearn, Caro2024:PTMlearning, Kim2025:LearnResources}, as we now show.

Using $\expval{A\otimes B}{\Phi_N} = \Tr(AB^\top)/2^N$, the Pauli basis becomes perfectly distinguishable (up to the transpose convention),
\begin{equation}
    \expval{(P_a \otimes P_b^\top)}{\Phi_N} = \delta_{ab}.
\end{equation}
In other words, the family $\{(P_a\otimes I)\ket{\Phi_N}\}_{a=1}^{R}$ is orthonormal and orthogonal to the no-jump state $\ket{\Phi_N}$, so $\ket{\Phi_N}$ admits a distinguishable jump basis spanning all $R=4^N-1$ Pauli channels. Therefore, the jump POVM $\bm{\mathcal{J}}$ of Eq.~\eqref{eq:jump-povm} can be implemented on $\mathscr{H}_S\otimes\mathscr{H}_{\rm mem}$. %, and the RPM record reduces to a multi-Poisson counting model over all parameters.

As a result, Lemma~\ref{lem:rpm-fisher} and Theorem~\ref{thm:rpm-qcrb} apply directly: RPM with quantum memory achieves the optimal eigenrate QCRB, and yields a full-rank Fisher matrix for learning the complete Pauli model with a \emph{single} experimental configuration. Hence all Pauli rates are simultaneously identifiable and achieving a target average precision $\bar{\tau}$ requires time
\begin{equation}
    T_{\rm mem}\ge \Omega\!\left(\bar{\tau}\gamma_{\min}\right),
\end{equation}
indicating an exponential complexity advantage over learning without memory~\cite{Chen2022:PauliChEst, Chen2024:TightPauliLearn, Kim2025:LearnResources}.

%==========
\subsection{Subdiffraction Quantum Imaging}
\label{sec:subdiff}

In this Section, we close with a short coda from subdiffraction incoherent imaging to highlight joint \emph{eigenrate} and \emph{eigenframe} estimation, and to illustrate how our framework applies in the seemingly disparate setting of quantum-optical imaging~\cite{Tsang2019:Starlight, Albarelli2020:ImagingPersp, Defienne2024:AdvQuImaging}. We give a brief overview and point to Appendix~\ref{app:subdiff-imaging} for details.

We consider the paradigmatic problem of imaging two weak, identical thermal emitters with centroid $\bar{x}$ and separation $d$ imaged through a diffraction-limited system with Gaussian point-spread function of width $\sigma$~\cite{Tsang2016:SPADE,Lupo2016:QuImaging}. As $d/\sigma\to 0$ (subdiffraction limit), the two intensity profiles merge on the image plane and the image becomes nearly indistinguishable from that of a single emitter, so direct imaging (i.e., local intensity detection) suffers from ``Rayleigh's curse''~\cite{Tsang2016:SPADE}. Quantum mechanically, however, the separation information remains accessible: the QFI for $d$ stays finite, but the optimal measurement basis depends on the unknown centroid $\bar{x}$~\cite{Lupo2016:QuImaging,Tsang2016:SPADE}. Estimating the separation $d$ is an eigenrate estimation problem, while estimating the centroid $\bar{x}$ is an eigenframe estimation problem (cf. Theorem~\ref{thm:qfi-matrix}).

Following the weak-thermal model of Ref.~\cite{Tsang2016:SPADE}, the weak thermal field is approximated as a single-photon state characterized by a mutual coherence $\Gamma(\bar{x},d)$, defined as the two-point correlation matrix of the thermal field on the image plane, with brightness $\Tr(\Gamma)=\varepsilon\ll 1$~\cite{Lupo2020:LinearOptLimits}. In our framework, we regard $\Gamma(\bar{x},d)$ as the dissipation matrix of a weak Markovian excitation process with jump operators $L_u=\hat{a}_u^\dagger$, indexed by image-plane coordinate $u$, acting on the optical vacuum. In the weak-signal regime (Scenario~3, Eq.~\eqref{eq:F-short-time}), the per-measurement QFI reads $F_Q(\delta t)\approx \delta t\,\dot F_Q(0^+)$ and is controlled by the initial QFI flow $\dot F_Q(0^+)$, so Theorem~\ref{thm:qfi-matrix} applies directly with vacuum correlator $\mathcal{C}(0^+)=I$. For $\nu$ identical bins and measurements, $F_Q^{\rm tot}=\nu F_Q(\delta t)$.

For identical emitters and $d/\sigma\ll 1$, translation symmetry suggests an eigen-decomposition,
\begin{equation}
  \Gamma(\bar{x},d)=V(\bar{x})\,\Lambda(d)\,V^\dagger(\bar{x}),
  \label{eq:Gamma-eigenframe-split-main}
\end{equation}
where $V(\bar{x})$ encodes the centroid as a coordinate shift, while the nonzero eigenvalues $\Lambda(d)=\mathrm{diag}(\lambda_+(d),\lambda_-(d))$ depend only on the separation~\cite{Tsang2016:SPADE,Lupo2016:QuImaging}. Thus $\bar{x}$ appears purely as an eigenframe parameter (generated by $K_{\bar{x}}=(\partial_{\bar{x}} V)V^\dagger$~\cite{Lupo2020:LinearOptLimits}), whereas $d$ appears purely through the eigenrates. Applying Theorem~\ref{thm:qfi-matrix} and assuming a Gaussian point-spread function yields (see Appendix~\ref{app:subdiff-imaging})
\begin{equation}
  \lim_{d/\sigma\to 0}F_Q^{\rm tot}
  =\frac{\nu\varepsilon}{\sigma^2}
  \begin{pmatrix}
    1 & 0\\[2pt]
    0 & \tfrac14
  \end{pmatrix},
\end{equation}
which is finite as $d/\sigma\to 0$ and reproduces the known separation QFI~\cite{Tsang2016:SPADE,Rehacek2017:MultiImaging}, thereby circumventing Rayleigh's curse, in principle, at the quantum limit.

In the eigenbasis of $\Gamma(\bar{x},d)$, the dominant modes are the symmetric and antisymmetric superpositions of the displaced point-spread functions~\cite{Lupo2016:QuImaging}. Consequently, when $\bar{x}$ is known, a spatial-mode demultiplexing (SPADE) measurement aligned to the centroid projects onto this eigenbasis and realizes an eigenrate-estimation protocol for $\lambda_-$ that saturates the separation QFI~\cite{Tsang2016:SPADE,Lupo2016:QuImaging}, precisely as in the RPM protocol. However, because $\bar{x}$ enters through the eigenframe $V(\bar{x})$, the SPADE basis itself depends on $\bar{x}$. Thus, when $\bar{x}$ is unknown \textit{a priori}, a fixed SPADE measurement is generally not optimal. One needs adaptive centroid alignment or a more general joint measurement strategy to recover the sub-Rayleigh separation advantage~\cite{Grace2020:Imaging}. %attaining the QFI requires an adaptive strategy that continuously tracks the centroid. 

Similar considerations apply to active quantum imaging, where a sample undergoing incoherent emission and/or absorption is probed with a structured quantum-optical field~\cite{Defienne2024:AdvQuImaging, Brady2025:TwinEcho}. The major difference is that the correlator matrix $\mathcal{C}(0^+)$ acquires a photon-number enhancement that depends on the probe state (e.g., Fock or two-mode squeezed), and the probe is initialized in the eigenmode basis of $\Gamma$~\cite{Brady2025:TwinEcho}.

%==========
%==========
\section{Outlook}
\label{sec:outlook}

In this work, we establish ultimate precision bounds for multiparameter Markovian noise metrology. Our results provide a unified Lindbladian framework for sensing (memoryless) stochastic signals, clarify the role of jump-space connectivity as a metrological resource, and furnish fundamental limits for noise characterization and benchmarking of quantum devices~\cite{Hashim2021:RandCompil, Burnett2019:BenchmarkDecoherence, Harper2020:PauliLearn, Lupke2020:Benchmarking, Van2023:PEC, Hashim2025:BenchmarkingRvw, Proctor2025:Benchmarking}, probing correlations in condensed-matter systems with quantum sensors~\cite{Casola2018:nvManyBody, Rovny2024:nvManyBody, Rovny2022NVcovariance, Ziffer2024:qnsCriticality, Cheng2025:mplex, Rovny2025:nvEntResource, Zhou2025:nvEntResource}, quantum imaging~\cite{Tsang2019:Starlight, Albarelli2020:ImagingPersp, Defienne2024:AdvQuImaging}, and searches for new physics with quantum measurement technology~\cite{YeZoller2024:Essay, Bass2024:NatRvw}.

Several directions remain open. A natural next step would connect the multiparameter noise metrology perspective to broader problems in Lindblad learning and tomography~\cite{Samach2022:LindbladTomography, Olsacher2025:LindbladLearning, Ivashkov2026:LindbladLearn}, where identifiability, sample complexity, and optimal precision all play a role. Extending the present theory beyond strict Markovianity to non-Markovian noise metrology~\cite{Breuer2016:NonMarkovRvw, Groszkowski2023:nonMarkovModel, Chin2012:QuMetrologyNonMarkov, White2022:nonMarkovProcessTom, Varona2025:nonMarkovLearn, Jordi2025:NonMarkovLearn} marks another important direction; although the loss of an immediate Lindblad description complicates the analysis, aspects of our approach may still carry over to effectively time-local settings with pseudo-Lindblad descriptions (cf.~\cite{Wang2026:PowerLaw}), to genuinely non-Markovian environments using effective master equations~\cite{Breuer2010} or with the help of effective memory modes~\cite{Tanimura2005,Plenio2016,Cirac2021}. Regarding measurement protocols, our RPM construction addresses eigenrate estimation, but broader classes of dissipative learning problems will likely require more intricate protocols where intrinsic measurement incompatibility may arise~\cite{Ragy2016compatibility, Albarelli2022:Incompatibility}; randomized strategies may prove useful here~\cite{Emerson2005:randNoiseEstim, Elben2023:RandToolbox, Li2023:ScrambleSensing, Gong2026:ScrambleSensing, Zhou2026:RandMeasurements}. The shot-limited regime [Scenario~2, Eq.~\eqref{eq:F-shot-lim}] also merits further study as the optimal interrogation time can depend on the unknown dissipative parameters, making adaptive estimation essential while also unlocking potentially exponential quantum advantage~\cite{Wang2024:ExpSensing, Prabhu2026:ExpSensingQSA}. Finally, although realizing quantum-sensing advantage in many-body noise metrology remains a daunting challenge, technical developments in preparing entangled probe states, collective control, and multiplexed quantum readout for large systems mark an exciting frontier~\cite{Zhou2020:QuMetroStrongSpin, Colombo2022:QSNEcho, Gao2025:NVEcho, Wu2025:NVSpinSqz, Montenegro2024:ManyBodyMetrology}.

%==========
\acknowledgments
A.J.B.\@ thanks Victor V. Albert for insightful discussions. 
A.J.B.\@ acknowledges support from the NRC Research Associateship Program at NIST.
Y.-X.W.~acknowledges support from a QuICS Hartree Postdoctoral Fellowship. 
L.P.G.P.~acknowledges support from the Beyond Moore’s Law project of the Advanced Simulation and Computing Program at LANL managed by Triad National Security, LLC, and the U.S. DoE, Office of Science, Accelerated Research in Quantum Computing, Fundamental Algorithmic Research toward Quantum Utility (FAR-Qu).
A.V.G.~was supported in part by ONR MURI, AFOSR MURI, DARPA SAVaNT ADVENT, NSF STAQ program, ARL (W911NF-24-2-0107), NSF QLCI (award No.~OMA-2120757),  DoE ASCR Quantum Testbed Pathfinder program (award No.~DE-SC0024220),   and NQVL:QSTD:Pilot:FTL. A.V.G.~also acknowledges support from the U.S.~Department of Energy, Office of Science, National Quantum Information Science Research Centers, Quantum Systems Accelerator (award No.~DE-SCL0000121).

%============ Bibliography & Appendix ==================

\bibliography{main}

\begin{thebibliography}{140}%
\makeatletter
\providecommand \@ifxundefined [1]{%
 \@ifx{#1\undefined}
}%
\providecommand \@ifnum [1]{%
 \ifnum #1\expandafter \@firstoftwo
 \else \expandafter \@secondoftwo
 \fi
}%
\providecommand \@ifx [1]{%
 \ifx #1\expandafter \@firstoftwo
 \else \expandafter \@secondoftwo
 \fi
}%
\providecommand \natexlab [1]{#1}%
\providecommand \enquote  [1]{``#1''}%
\providecommand \bibnamefont  [1]{#1}%
\providecommand \bibfnamefont [1]{#1}%
\providecommand \citenamefont [1]{#1}%
\providecommand \href@noop [0]{\@secondoftwo}%
\providecommand \href [0]{\begingroup \@sanitize@url \@href}%
\providecommand \@href[1]{\@@startlink{#1}\@@href}%
\providecommand \@@href[1]{\endgroup#1\@@endlink}%
\providecommand \@sanitize@url [0]{\catcode `\\12\catcode `\$12\catcode `\&12\catcode `\#12\catcode `\^12\catcode `\_12\catcode `\%12\relax}%
\providecommand \@@startlink[1]{}%
\providecommand \@@endlink[0]{}%
\providecommand \url  [0]{\begingroup\@sanitize@url \@url }%
\providecommand \@url [1]{\endgroup\@href {#1}{\urlprefix }}%
\providecommand \urlprefix  [0]{URL }%
\providecommand \Eprint [0]{\href }%
\providecommand \doibase [0]{https://doi.org/}%
\providecommand \selectlanguage [0]{\@gobble}%
\providecommand \bibinfo  [0]{\@secondoftwo}%
\providecommand \bibfield  [0]{\@secondoftwo}%
\providecommand \translation [1]{[#1]}%
\providecommand \BibitemOpen [0]{}%
\providecommand \bibitemStop [0]{}%
\providecommand \bibitemNoStop [0]{.\EOS\space}%
\providecommand \EOS [0]{\spacefactor3000\relax}%
\providecommand \BibitemShut  [1]{\csname bibitem#1\endcsname}%
\let\auto@bib@innerbib\@empty
%</preamble>
\bibitem [{\citenamefont {Degen}\ \emph {et~al.}(2017)\citenamefont {Degen}, \citenamefont {Reinhard},\ and\ \citenamefont {Cappellaro}}]{Degen2017:QuSensing}%
  \BibitemOpen
  \bibfield  {author} {\bibinfo {author} {\bibfnamefont {C.~L.}\ \bibnamefont {Degen}}, \bibinfo {author} {\bibfnamefont {F.}~\bibnamefont {Reinhard}},\ and\ \bibinfo {author} {\bibfnamefont {P.}~\bibnamefont {Cappellaro}},\ }\bibfield  {title} {\bibinfo {title} {{Quantum sensing}},\ }\href {https://doi.org/10.1103/RevModPhys.89.035002} {\bibfield  {journal} {\bibinfo  {journal} {Rev. Mod. Phys.}\ }\textbf {\bibinfo {volume} {89}},\ \bibinfo {pages} {035002} (\bibinfo {year} {2017})}\BibitemShut {NoStop}%
\bibitem [{\citenamefont {Pezz\`e}\ \emph {et~al.}(2018)\citenamefont {Pezz\`e}, \citenamefont {Smerzi}, \citenamefont {Oberthaler}, \citenamefont {Schmied},\ and\ \citenamefont {Treutlein}}]{Pezze2018:QuMetrologyAtoms}%
  \BibitemOpen
  \bibfield  {author} {\bibinfo {author} {\bibfnamefont {L.}~\bibnamefont {Pezz\`e}}, \bibinfo {author} {\bibfnamefont {A.}~\bibnamefont {Smerzi}}, \bibinfo {author} {\bibfnamefont {M.~K.}\ \bibnamefont {Oberthaler}}, \bibinfo {author} {\bibfnamefont {R.}~\bibnamefont {Schmied}},\ and\ \bibinfo {author} {\bibfnamefont {P.}~\bibnamefont {Treutlein}},\ }\bibfield  {title} {\bibinfo {title} {{Quantum metrology with nonclassical states of atomic ensembles}},\ }\href {https://doi.org/10.1103/RevModPhys.90.035005} {\bibfield  {journal} {\bibinfo  {journal} {Rev. Mod. Phys.}\ }\textbf {\bibinfo {volume} {90}},\ \bibinfo {pages} {035005} (\bibinfo {year} {2018})}\BibitemShut {NoStop}%
\bibitem [{\citenamefont {Pirandola}\ \emph {et~al.}(2018)\citenamefont {Pirandola}, \citenamefont {Bardhan}, \citenamefont {Gehring}, \citenamefont {Weedbrook},\ and\ \citenamefont {Lloyd}}]{Lloyd2018:PhotonQuSensing}%
  \BibitemOpen
  \bibfield  {author} {\bibinfo {author} {\bibfnamefont {S.}~\bibnamefont {Pirandola}}, \bibinfo {author} {\bibfnamefont {B.~R.}\ \bibnamefont {Bardhan}}, \bibinfo {author} {\bibfnamefont {T.}~\bibnamefont {Gehring}}, \bibinfo {author} {\bibfnamefont {C.}~\bibnamefont {Weedbrook}},\ and\ \bibinfo {author} {\bibfnamefont {S.}~\bibnamefont {Lloyd}},\ }\bibfield  {title} {\bibinfo {title} {{Advances in photonic quantum sensing}},\ }\href {https://doi.org/10.1038/s41566-018-0301-6} {\bibfield  {journal} {\bibinfo  {journal} {Nat. Photonics}\ }\textbf {\bibinfo {volume} {12}},\ \bibinfo {pages} {724–733} (\bibinfo {year} {2018})}\BibitemShut {NoStop}%
\bibitem [{\citenamefont {Huang}\ \emph {et~al.}(2024)\citenamefont {Huang}, \citenamefont {Zhuang},\ and\ \citenamefont {Lee}}]{Huang2024:EntangledMetrology}%
  \BibitemOpen
  \bibfield  {author} {\bibinfo {author} {\bibfnamefont {J.}~\bibnamefont {Huang}}, \bibinfo {author} {\bibfnamefont {M.}~\bibnamefont {Zhuang}},\ and\ \bibinfo {author} {\bibfnamefont {C.}~\bibnamefont {Lee}},\ }\bibfield  {title} {\bibinfo {title} {{Entanglement-enhanced quantum metrology: From standard quantum limit to Heisenberg limit}},\ }\bibfield  {journal} {\bibinfo  {journal} {Appl. Phys. Rev.}\ }\textbf {\bibinfo {volume} {11}},\ \href {https://doi.org/10.1063/5.0204102} {10.1063/5.0204102} (\bibinfo {year} {2024})\BibitemShut {NoStop}%
\bibitem [{\citenamefont {Zaiser}\ \emph {et~al.}(2016)\citenamefont {Zaiser}, \citenamefont {Rendler}, \citenamefont {Jakobi}, \citenamefont {Wolf}, \citenamefont {Lee}, \citenamefont {Wagner}, \citenamefont {Bergholm}, \citenamefont {Schulte-Herbr{\"u}ggen}, \citenamefont {Neumann} \emph {et~al.}}]{Zaiser2016:QuMemorySensing}%
  \BibitemOpen
  \bibfield  {author} {\bibinfo {author} {\bibfnamefont {S.}~\bibnamefont {Zaiser}}, \bibinfo {author} {\bibfnamefont {T.}~\bibnamefont {Rendler}}, \bibinfo {author} {\bibfnamefont {I.}~\bibnamefont {Jakobi}}, \bibinfo {author} {\bibfnamefont {T.}~\bibnamefont {Wolf}}, \bibinfo {author} {\bibfnamefont {S.-Y.}\ \bibnamefont {Lee}}, \bibinfo {author} {\bibfnamefont {S.}~\bibnamefont {Wagner}}, \bibinfo {author} {\bibfnamefont {V.}~\bibnamefont {Bergholm}}, \bibinfo {author} {\bibfnamefont {T.}~\bibnamefont {Schulte-Herbr{\"u}ggen}}, \bibinfo {author} {\bibfnamefont {P.}~\bibnamefont {Neumann}}, \emph {et~al.},\ }\bibfield  {title} {\bibinfo {title} {{Enhancing quantum sensing sensitivity by a quantum memory}},\ }\href {https://doi.org/10.1038/ncomms12279} {\bibfield  {journal} {\bibinfo  {journal} {Nat. Commun.}\ }\textbf {\bibinfo {volume} {7}},\ \bibinfo {pages} {12279} (\bibinfo {year} {2016})}\BibitemShut {NoStop}%
\bibitem [{\citenamefont {Chen}\ \emph {et~al.}(2022{\natexlab{a}})\citenamefont {Chen}, \citenamefont {Cotler}, \citenamefont {Huang},\ and\ \citenamefont {Li}}]{Chen2022:ExpSepMemory}%
  \BibitemOpen
  \bibfield  {author} {\bibinfo {author} {\bibfnamefont {S.}~\bibnamefont {Chen}}, \bibinfo {author} {\bibfnamefont {J.}~\bibnamefont {Cotler}}, \bibinfo {author} {\bibfnamefont {H.-Y.}\ \bibnamefont {Huang}},\ and\ \bibinfo {author} {\bibfnamefont {J.}~\bibnamefont {Li}},\ }\bibfield  {title} {\bibinfo {title} {{Exponential Separations Between Learning With and Without Quantum Memory}},\ }in\ \href {https://doi.org/10.1109/FOCS52979.2021.00063} {\emph {\bibinfo {booktitle} {2021 IEEE 62nd Annual Symposium on Foundations of Computer Science (FOCS)}}}\ (\bibinfo {year} {2022})\ pp.\ \bibinfo {pages} {574--585}\BibitemShut {NoStop}%
\bibitem [{\citenamefont {Huang}\ \emph {et~al.}(2022)\citenamefont {Huang}, \citenamefont {Broughton}, \citenamefont {Cotler}, \citenamefont {Chen}, \citenamefont {Li}, \citenamefont {Mohseni}, \citenamefont {Neven}, \citenamefont {Babbush}, \citenamefont {Kueng} \emph {et~al.}}]{Huang2022:QuAdvLearn}%
  \BibitemOpen
  \bibfield  {author} {\bibinfo {author} {\bibfnamefont {H.-Y.}\ \bibnamefont {Huang}}, \bibinfo {author} {\bibfnamefont {M.}~\bibnamefont {Broughton}}, \bibinfo {author} {\bibfnamefont {J.}~\bibnamefont {Cotler}}, \bibinfo {author} {\bibfnamefont {S.}~\bibnamefont {Chen}}, \bibinfo {author} {\bibfnamefont {J.}~\bibnamefont {Li}}, \bibinfo {author} {\bibfnamefont {M.}~\bibnamefont {Mohseni}}, \bibinfo {author} {\bibfnamefont {H.}~\bibnamefont {Neven}}, \bibinfo {author} {\bibfnamefont {R.}~\bibnamefont {Babbush}}, \bibinfo {author} {\bibfnamefont {R.}~\bibnamefont {Kueng}}, \emph {et~al.},\ }\bibfield  {title} {\bibinfo {title} {{Quantum advantage in learning from experiments}},\ }\href {https://doi.org/10.1126/science.abn7293} {\bibfield  {journal} {\bibinfo  {journal} {Science}\ }\textbf {\bibinfo {volume} {376}},\ \bibinfo {pages} {1182–1186} (\bibinfo {year} {2022})}\BibitemShut {NoStop}%
\bibitem [{\citenamefont {Massar}\ and\ \citenamefont {Popescu}(1995)}]{Massar1995:EntanglMeas}%
  \BibitemOpen
  \bibfield  {author} {\bibinfo {author} {\bibfnamefont {S.}~\bibnamefont {Massar}}\ and\ \bibinfo {author} {\bibfnamefont {S.}~\bibnamefont {Popescu}},\ }\bibfield  {title} {\bibinfo {title} {{Optimal Extraction of Information from Finite Quantum Ensembles}},\ }\href {https://doi.org/10.1103/PhysRevLett.74.1259} {\bibfield  {journal} {\bibinfo  {journal} {Phys. Rev. Lett.}\ }\textbf {\bibinfo {volume} {74}},\ \bibinfo {pages} {1259} (\bibinfo {year} {1995})}\BibitemShut {NoStop}%
\bibitem [{\citenamefont {Conlon}\ \emph {et~al.}(2023)\citenamefont {Conlon}, \citenamefont {Vogl}, \citenamefont {Marciniak}, \citenamefont {Pogorelov}, \citenamefont {Yung}, \citenamefont {Eilenberger}, \citenamefont {Berry}, \citenamefont {Santana}, \citenamefont {Blatt} \emph {et~al.}}]{Conlon2023:EntanglMeas}%
  \BibitemOpen
  \bibfield  {author} {\bibinfo {author} {\bibfnamefont {L.~O.}\ \bibnamefont {Conlon}}, \bibinfo {author} {\bibfnamefont {T.}~\bibnamefont {Vogl}}, \bibinfo {author} {\bibfnamefont {C.~D.}\ \bibnamefont {Marciniak}}, \bibinfo {author} {\bibfnamefont {I.}~\bibnamefont {Pogorelov}}, \bibinfo {author} {\bibfnamefont {S.~K.}\ \bibnamefont {Yung}}, \bibinfo {author} {\bibfnamefont {F.}~\bibnamefont {Eilenberger}}, \bibinfo {author} {\bibfnamefont {D.~W.}\ \bibnamefont {Berry}}, \bibinfo {author} {\bibfnamefont {F.~S.}\ \bibnamefont {Santana}}, \bibinfo {author} {\bibfnamefont {R.}~\bibnamefont {Blatt}}, \emph {et~al.},\ }\bibfield  {title} {\bibinfo {title} {{Approaching optimal entangling collective measurements on quantum computing platforms}},\ }\href {https://doi.org/10.1038/s41567-022-01875-7} {\bibfield  {journal} {\bibinfo  {journal} {Nat. Phys.}\ }\textbf {\bibinfo {volume} {19}},\ \bibinfo {pages} {351–357} (\bibinfo {year} {2023})}\BibitemShut {NoStop}%
\bibitem [{\citenamefont {Helstrom}(1976)}]{Helstrom1976}%
  \BibitemOpen
  \bibfield  {author} {\bibinfo {author} {\bibfnamefont {C.~W.}\ \bibnamefont {Helstrom}},\ }\href@noop {} {\emph {\bibinfo {title} {{Quantum Detection and Estimation Theory}}}}\ (\bibinfo  {publisher} {Academic Press},\ \bibinfo {address} {New York},\ \bibinfo {year} {1976})\BibitemShut {NoStop}%
\bibitem [{\citenamefont {Holevo}(1982)}]{Holevo1982}%
  \BibitemOpen
  \bibfield  {author} {\bibinfo {author} {\bibfnamefont {A.~S.}\ \bibnamefont {Holevo}},\ }\href@noop {} {\emph {\bibinfo {title} {{Probabilistic and Statistical Aspects of Quantum Theory}}}}\ (\bibinfo  {publisher} {North-Holland},\ \bibinfo {address} {Amsterdam},\ \bibinfo {year} {1982})\BibitemShut {NoStop}%
\bibitem [{\citenamefont {Paris}(2009)}]{Paris2009:QFI}%
  \BibitemOpen
  \bibfield  {author} {\bibinfo {author} {\bibfnamefont {M.~G.~A.}\ \bibnamefont {Paris}},\ }\bibfield  {title} {\bibinfo {title} {{Quantum estimation for quantum technology}},\ }\href {https://doi.org/10.1142/S0219749909004839} {\bibfield  {journal} {\bibinfo  {journal} {Int. J. Quantum Inf.}\ }\textbf {\bibinfo {volume} {7}},\ \bibinfo {pages} {125} (\bibinfo {year} {2009})}\BibitemShut {NoStop}%
\bibitem [{\citenamefont {Braunstein}\ and\ \citenamefont {Caves}(1994)}]{Braunstein1994:Bures}%
  \BibitemOpen
  \bibfield  {author} {\bibinfo {author} {\bibfnamefont {S.~L.}\ \bibnamefont {Braunstein}}\ and\ \bibinfo {author} {\bibfnamefont {C.~M.}\ \bibnamefont {Caves}},\ }\bibfield  {title} {\bibinfo {title} {{Statistical distance and the geometry of quantum states}},\ }\href {https://doi.org/10.1103/PhysRevLett.72.3439} {\bibfield  {journal} {\bibinfo  {journal} {Phys. Rev. Lett.}\ }\textbf {\bibinfo {volume} {72}},\ \bibinfo {pages} {3439} (\bibinfo {year} {1994})}\BibitemShut {NoStop}%
\bibitem [{\citenamefont {Giovannetti}\ \emph {et~al.}(2004)\citenamefont {Giovannetti}, \citenamefont {Lloyd},\ and\ \citenamefont {Maccone}}]{Giovannetti2004:BeatSQL}%
  \BibitemOpen
  \bibfield  {author} {\bibinfo {author} {\bibfnamefont {V.}~\bibnamefont {Giovannetti}}, \bibinfo {author} {\bibfnamefont {S.}~\bibnamefont {Lloyd}},\ and\ \bibinfo {author} {\bibfnamefont {L.}~\bibnamefont {Maccone}},\ }\bibfield  {title} {\bibinfo {title} {{Quantum-Enhanced Measurements: Beating the Standard Quantum Limit}},\ }\href {https://doi.org/10.1126/science.1104149} {\bibfield  {journal} {\bibinfo  {journal} {Science}\ }\textbf {\bibinfo {volume} {306}},\ \bibinfo {pages} {1330} (\bibinfo {year} {2004})}\BibitemShut {NoStop}%
\bibitem [{\citenamefont {Giovannetti}\ \emph {et~al.}(2006)\citenamefont {Giovannetti}, \citenamefont {Lloyd},\ and\ \citenamefont {Maccone}}]{Giovannetti2006:QuMetrology}%
  \BibitemOpen
  \bibfield  {author} {\bibinfo {author} {\bibfnamefont {V.}~\bibnamefont {Giovannetti}}, \bibinfo {author} {\bibfnamefont {S.}~\bibnamefont {Lloyd}},\ and\ \bibinfo {author} {\bibfnamefont {L.}~\bibnamefont {Maccone}},\ }\bibfield  {title} {\bibinfo {title} {{Quantum Metrology}},\ }\href {https://doi.org/10.1103/PhysRevLett.96.010401} {\bibfield  {journal} {\bibinfo  {journal} {Phys. Rev. Lett.}\ }\textbf {\bibinfo {volume} {96}},\ \bibinfo {pages} {010401} (\bibinfo {year} {2006})}\BibitemShut {NoStop}%
\bibitem [{\citenamefont {Brady}\ \emph {et~al.}(2026)\citenamefont {Brady}, \citenamefont {Wang}, \citenamefont {Albert}, \citenamefont {Gorshkov},\ and\ \citenamefont {Zhuang}}]{Brady2024:NoiseQSN}%
  \BibitemOpen
  \bibfield  {author} {\bibinfo {author} {\bibfnamefont {A.~J.}\ \bibnamefont {Brady}}, \bibinfo {author} {\bibfnamefont {Y.-X.}\ \bibnamefont {Wang}}, \bibinfo {author} {\bibfnamefont {V.~V.}\ \bibnamefont {Albert}}, \bibinfo {author} {\bibfnamefont {A.~V.}\ \bibnamefont {Gorshkov}},\ and\ \bibinfo {author} {\bibfnamefont {Q.}~\bibnamefont {Zhuang}},\ }\bibfield  {title} {\bibinfo {title} {{Correlated Noise Estimation with Quantum Sensor Networks}},\ }\href {https://doi.org/10.1103/sl32-jn82} {\bibfield  {journal} {\bibinfo  {journal} {Phys. Rev. Lett.}\ }\textbf {\bibinfo {volume} {136}},\ \bibinfo {pages} {080803} (\bibinfo {year} {2026})}\BibitemShut {NoStop}%
\bibitem [{\citenamefont {Wang}\ \emph {et~al.}(2024)\citenamefont {Wang}, \citenamefont {Bringewatt}, \citenamefont {Seif}, \citenamefont {Brady}, \citenamefont {Oh},\ and\ \citenamefont {Gorshkov}}]{Wang2024:ExpSensing}%
  \BibitemOpen
  \bibfield  {author} {\bibinfo {author} {\bibfnamefont {Y.-X.}\ \bibnamefont {Wang}}, \bibinfo {author} {\bibfnamefont {J.}~\bibnamefont {Bringewatt}}, \bibinfo {author} {\bibfnamefont {A.}~\bibnamefont {Seif}}, \bibinfo {author} {\bibfnamefont {A.~J.}\ \bibnamefont {Brady}}, \bibinfo {author} {\bibfnamefont {C.}~\bibnamefont {Oh}},\ and\ \bibinfo {author} {\bibfnamefont {A.~V.}\ \bibnamefont {Gorshkov}},\ }\href@noop {} {\bibinfo {title} {{Exponential entanglement advantage in sensing correlated noise}}} (\bibinfo {year} {2024}),\ \Eprint {https://arxiv.org/abs/2410.05878} {arXiv:2410.05878 [quant-ph]} \BibitemShut {NoStop}%
\bibitem [{\citenamefont {Prabhu}\ \emph {et~al.}(2026)\citenamefont {Prabhu}, \citenamefont {Kremenetski}, \citenamefont {Khan}, \citenamefont {Yanagimoto},\ and\ \citenamefont {McMahon}}]{Prabhu2026:ExpSensingQSA}%
  \BibitemOpen
  \bibfield  {author} {\bibinfo {author} {\bibfnamefont {S.}~\bibnamefont {Prabhu}}, \bibinfo {author} {\bibfnamefont {V.}~\bibnamefont {Kremenetski}}, \bibinfo {author} {\bibfnamefont {S.~A.}\ \bibnamefont {Khan}}, \bibinfo {author} {\bibfnamefont {R.}~\bibnamefont {Yanagimoto}},\ and\ \bibinfo {author} {\bibfnamefont {P.~L.}\ \bibnamefont {McMahon}},\ }\bibfield  {title} {\bibinfo {title} {{Exponential advantage in quantum sensing of correlated parameters}},\ }\href {https://doi.org/10.22331/q-2026-01-14-1963} {\bibfield  {journal} {\bibinfo  {journal} {Quantum}\ }\textbf {\bibinfo {volume} {10}},\ \bibinfo {pages} {1963} (\bibinfo {year} {2026})}\BibitemShut {NoStop}%
\bibitem [{\citenamefont {Wang}\ \emph {et~al.}(2026)\citenamefont {Wang}, \citenamefont {Brady}, \citenamefont {Belliardo},\ and\ \citenamefont {Gorshkov}}]{Wang2026:PowerLaw}%
  \BibitemOpen
  \bibfield  {author} {\bibinfo {author} {\bibfnamefont {Y.-X.}\ \bibnamefont {Wang}}, \bibinfo {author} {\bibfnamefont {A.~J.}\ \bibnamefont {Brady}}, \bibinfo {author} {\bibfnamefont {F.}~\bibnamefont {Belliardo}},\ and\ \bibinfo {author} {\bibfnamefont {A.~V.}\ \bibnamefont {Gorshkov}},\ }\href@noop {} {\bibinfo {title} {{Entanglement advantage in sensing power-law spatiotemporal noise correlations}}} (\bibinfo {year} {2026}),\ \Eprint {https://arxiv.org/abs/2603.15742} {arXiv:2603.15742 [quant-ph]} \BibitemShut {NoStop}%
\bibitem [{\citenamefont {Chen}\ \emph {et~al.}(2022{\natexlab{b}})\citenamefont {Chen}, \citenamefont {Zhou}, \citenamefont {Seif},\ and\ \citenamefont {Jiang}}]{Chen2022:PauliChEst}%
  \BibitemOpen
  \bibfield  {author} {\bibinfo {author} {\bibfnamefont {S.}~\bibnamefont {Chen}}, \bibinfo {author} {\bibfnamefont {S.}~\bibnamefont {Zhou}}, \bibinfo {author} {\bibfnamefont {A.}~\bibnamefont {Seif}},\ and\ \bibinfo {author} {\bibfnamefont {L.}~\bibnamefont {Jiang}},\ }\bibfield  {title} {\bibinfo {title} {{Quantum advantages for Pauli channel estimation}},\ }\href {https://doi.org/10.1103/PhysRevA.105.032435} {\bibfield  {journal} {\bibinfo  {journal} {Phys. Rev. A}\ }\textbf {\bibinfo {volume} {105}},\ \bibinfo {pages} {032435} (\bibinfo {year} {2022}{\natexlab{b}})}\BibitemShut {NoStop}%
\bibitem [{\citenamefont {Chen}\ \emph {et~al.}(2024)\citenamefont {Chen}, \citenamefont {Oh}, \citenamefont {Zhou}, \citenamefont {Huang},\ and\ \citenamefont {Jiang}}]{Chen2024:TightPauliLearn}%
  \BibitemOpen
  \bibfield  {author} {\bibinfo {author} {\bibfnamefont {S.}~\bibnamefont {Chen}}, \bibinfo {author} {\bibfnamefont {C.}~\bibnamefont {Oh}}, \bibinfo {author} {\bibfnamefont {S.}~\bibnamefont {Zhou}}, \bibinfo {author} {\bibfnamefont {H.-Y.}\ \bibnamefont {Huang}},\ and\ \bibinfo {author} {\bibfnamefont {L.}~\bibnamefont {Jiang}},\ }\bibfield  {title} {\bibinfo {title} {{Tight Bounds on Pauli Channel Learning without Entanglement}},\ }\href {https://doi.org/10.1103/PhysRevLett.132.180805} {\bibfield  {journal} {\bibinfo  {journal} {Phys. Rev. Lett.}\ }\textbf {\bibinfo {volume} {132}},\ \bibinfo {pages} {180805} (\bibinfo {year} {2024})}\BibitemShut {NoStop}%
\bibitem [{\citenamefont {Caro}(2024)}]{Caro2024:PTMlearning}%
  \BibitemOpen
  \bibfield  {author} {\bibinfo {author} {\bibfnamefont {M.~C.}\ \bibnamefont {Caro}},\ }\bibfield  {title} {\bibinfo {title} {{Learning Quantum Processes and Hamiltonians via the Pauli Transfer Matrix}},\ }\href {https://doi.org/10.1145/3670418} {\bibfield  {journal} {\bibinfo  {journal} {ACM Trans. Quantum Comput.}\ }\textbf {\bibinfo {volume} {5}},\ \bibinfo {pages} {1–53} (\bibinfo {year} {2024})}\BibitemShut {NoStop}%
\bibitem [{\citenamefont {Kim}\ and\ \citenamefont {Oh}(2026)}]{Kim2025:LearnResources}%
  \BibitemOpen
  \bibfield  {author} {\bibinfo {author} {\bibfnamefont {M.}~\bibnamefont {Kim}}\ and\ \bibinfo {author} {\bibfnamefont {C.}~\bibnamefont {Oh}},\ }\bibfield  {title} {\bibinfo {title} {{On the fundamental resource for exponential advantage in quantum channel learning}},\ }\href {https://doi.org/10.1038/s41467-026-68532-y} {\bibfield  {journal} {\bibinfo  {journal} {Nat. Commun.}\ }\textbf {\bibinfo {volume} {17}},\ \bibinfo {pages} {1822} (\bibinfo {year} {2026})}\BibitemShut {NoStop}%
\bibitem [{\citenamefont {Tsang}(2019)}]{Tsang2019:Starlight}%
  \BibitemOpen
  \bibfield  {author} {\bibinfo {author} {\bibfnamefont {M.}~\bibnamefont {Tsang}},\ }\bibfield  {title} {\bibinfo {title} {{Resolving starlight: a quantum perspective}},\ }\href {https://doi.org/10.1080/00107514.2020.1736375} {\bibfield  {journal} {\bibinfo  {journal} {Contemp. Phys.}\ }\textbf {\bibinfo {volume} {60}},\ \bibinfo {pages} {279} (\bibinfo {year} {2019})}\BibitemShut {NoStop}%
\bibitem [{\citenamefont {Albarelli}\ \emph {et~al.}(2020)\citenamefont {Albarelli}, \citenamefont {Barbieri}, \citenamefont {Genoni},\ and\ \citenamefont {Gianani}}]{Albarelli2020:ImagingPersp}%
  \BibitemOpen
  \bibfield  {author} {\bibinfo {author} {\bibfnamefont {F.}~\bibnamefont {Albarelli}}, \bibinfo {author} {\bibfnamefont {M.}~\bibnamefont {Barbieri}}, \bibinfo {author} {\bibfnamefont {M.}~\bibnamefont {Genoni}},\ and\ \bibinfo {author} {\bibfnamefont {I.}~\bibnamefont {Gianani}},\ }\bibfield  {title} {\bibinfo {title} {{A perspective on multiparameter quantum metrology: From theoretical tools to applications in quantum imaging}},\ }\href {https://doi.org/10.1016/j.physleta.2020.126311} {\bibfield  {journal} {\bibinfo  {journal} {Phys. Lett. A}\ }\textbf {\bibinfo {volume} {384}},\ \bibinfo {pages} {126311} (\bibinfo {year} {2020})}\BibitemShut {NoStop}%
\bibitem [{\citenamefont {Defienne}\ \emph {et~al.}(2024)\citenamefont {Defienne}, \citenamefont {Bowen}, \citenamefont {Chekhova}, \citenamefont {Lemos}, \citenamefont {Oron}, \citenamefont {Ramelow}, \citenamefont {Treps},\ and\ \citenamefont {Faccio}}]{Defienne2024:AdvQuImaging}%
  \BibitemOpen
  \bibfield  {author} {\bibinfo {author} {\bibfnamefont {H.}~\bibnamefont {Defienne}}, \bibinfo {author} {\bibfnamefont {W.~P.}\ \bibnamefont {Bowen}}, \bibinfo {author} {\bibfnamefont {M.}~\bibnamefont {Chekhova}}, \bibinfo {author} {\bibfnamefont {G.~B.}\ \bibnamefont {Lemos}}, \bibinfo {author} {\bibfnamefont {D.}~\bibnamefont {Oron}}, \bibinfo {author} {\bibfnamefont {S.}~\bibnamefont {Ramelow}}, \bibinfo {author} {\bibfnamefont {N.}~\bibnamefont {Treps}},\ and\ \bibinfo {author} {\bibfnamefont {D.}~\bibnamefont {Faccio}},\ }\bibfield  {title} {\bibinfo {title} {{Advances in quantum imaging}},\ }\href {https://doi.org/10.1038/s41566-024-01516-w} {\bibfield  {journal} {\bibinfo  {journal} {Nat. Photon.}\ }\textbf {\bibinfo {volume} {18}},\ \bibinfo {pages} {1024} (\bibinfo {year} {2024})}\BibitemShut {NoStop}%
\bibitem [{\citenamefont {Hashim}\ \emph {et~al.}(2021)\citenamefont {Hashim}, \citenamefont {Naik}, \citenamefont {Morvan}, \citenamefont {Ville}, \citenamefont {Mitchell}, \citenamefont {Kreikebaum}, \citenamefont {Davis}, \citenamefont {Smith}, \citenamefont {Iancu} \emph {et~al.}}]{Hashim2021:RandCompil}%
  \BibitemOpen
  \bibfield  {author} {\bibinfo {author} {\bibfnamefont {A.}~\bibnamefont {Hashim}}, \bibinfo {author} {\bibfnamefont {R.~K.}\ \bibnamefont {Naik}}, \bibinfo {author} {\bibfnamefont {A.}~\bibnamefont {Morvan}}, \bibinfo {author} {\bibfnamefont {J.-L.}\ \bibnamefont {Ville}}, \bibinfo {author} {\bibfnamefont {B.}~\bibnamefont {Mitchell}}, \bibinfo {author} {\bibfnamefont {J.~M.}\ \bibnamefont {Kreikebaum}}, \bibinfo {author} {\bibfnamefont {M.}~\bibnamefont {Davis}}, \bibinfo {author} {\bibfnamefont {E.}~\bibnamefont {Smith}}, \bibinfo {author} {\bibfnamefont {C.}~\bibnamefont {Iancu}}, \emph {et~al.},\ }\bibfield  {title} {\bibinfo {title} {{Randomized Compiling for Scalable Quantum Computing on a Noisy Superconducting Quantum Processor}},\ }\href {https://doi.org/10.1103/PhysRevX.11.041039} {\bibfield  {journal} {\bibinfo  {journal} {Phys. Rev. X}\ }\textbf {\bibinfo {volume} {11}},\ \bibinfo {pages} {041039} (\bibinfo {year} {2021})}\BibitemShut {NoStop}%
\bibitem [{\citenamefont {Burnett}\ \emph {et~al.}(2019)\citenamefont {Burnett}, \citenamefont {Bengtsson}, \citenamefont {Scigliuzzo}, \citenamefont {Niepce}, \citenamefont {Kudra}, \citenamefont {Delsing},\ and\ \citenamefont {Bylander}}]{Burnett2019:BenchmarkDecoherence}%
  \BibitemOpen
  \bibfield  {author} {\bibinfo {author} {\bibfnamefont {J.~J.}\ \bibnamefont {Burnett}}, \bibinfo {author} {\bibfnamefont {A.}~\bibnamefont {Bengtsson}}, \bibinfo {author} {\bibfnamefont {M.}~\bibnamefont {Scigliuzzo}}, \bibinfo {author} {\bibfnamefont {D.}~\bibnamefont {Niepce}}, \bibinfo {author} {\bibfnamefont {M.}~\bibnamefont {Kudra}}, \bibinfo {author} {\bibfnamefont {P.}~\bibnamefont {Delsing}},\ and\ \bibinfo {author} {\bibfnamefont {J.}~\bibnamefont {Bylander}},\ }\bibfield  {title} {\bibinfo {title} {{Decoherence benchmarking of superconducting qubits}},\ }\href {https://doi.org/10.1038/s41534-019-0168-5} {\bibfield  {journal} {\bibinfo  {journal} {npj Quantum Inf.}\ }\textbf {\bibinfo {volume} {5}},\ \bibinfo {pages} {54} (\bibinfo {year} {2019})}\BibitemShut {NoStop}%
\bibitem [{\citenamefont {Harper}\ \emph {et~al.}(2020)\citenamefont {Harper}, \citenamefont {Flammia},\ and\ \citenamefont {Wallman}}]{Harper2020:PauliLearn}%
  \BibitemOpen
  \bibfield  {author} {\bibinfo {author} {\bibfnamefont {R.}~\bibnamefont {Harper}}, \bibinfo {author} {\bibfnamefont {S.~T.}\ \bibnamefont {Flammia}},\ and\ \bibinfo {author} {\bibfnamefont {J.~J.}\ \bibnamefont {Wallman}},\ }\bibfield  {title} {\bibinfo {title} {{Efficient learning of quantum noise}},\ }\href {https://doi.org/10.1038/s41567-020-0992-8} {\bibfield  {journal} {\bibinfo  {journal} {Nat. Phys.}\ }\textbf {\bibinfo {volume} {16}},\ \bibinfo {pages} {1184} (\bibinfo {year} {2020})}\BibitemShut {NoStop}%
\bibitem [{\citenamefont {von L\"upke}\ \emph {et~al.}(2020)\citenamefont {von L\"upke}, \citenamefont {Beaudoin}, \citenamefont {Norris}, \citenamefont {Sung}, \citenamefont {Winik}, \citenamefont {Qiu}, \citenamefont {Kjaergaard}, \citenamefont {Kim}, \citenamefont {Yoder} \emph {et~al.}}]{Lupke2020:Benchmarking}%
  \BibitemOpen
  \bibfield  {author} {\bibinfo {author} {\bibfnamefont {U.}~\bibnamefont {von L\"upke}}, \bibinfo {author} {\bibfnamefont {F.}~\bibnamefont {Beaudoin}}, \bibinfo {author} {\bibfnamefont {L.~M.}\ \bibnamefont {Norris}}, \bibinfo {author} {\bibfnamefont {Y.}~\bibnamefont {Sung}}, \bibinfo {author} {\bibfnamefont {R.}~\bibnamefont {Winik}}, \bibinfo {author} {\bibfnamefont {J.~Y.}\ \bibnamefont {Qiu}}, \bibinfo {author} {\bibfnamefont {M.}~\bibnamefont {Kjaergaard}}, \bibinfo {author} {\bibfnamefont {D.}~\bibnamefont {Kim}}, \bibinfo {author} {\bibfnamefont {J.}~\bibnamefont {Yoder}}, \emph {et~al.},\ }\bibfield  {title} {\bibinfo {title} {{Two-Qubit Spectroscopy of Spatiotemporally Correlated Quantum Noise in Superconducting Qubits}},\ }\href {https://doi.org/10.1103/PRXQuantum.1.010305} {\bibfield  {journal} {\bibinfo  {journal} {PRX Quantum}\ }\textbf {\bibinfo {volume} {1}},\ \bibinfo {pages} {010305} (\bibinfo {year} {2020})}\BibitemShut {NoStop}%
\bibitem [{\citenamefont {Van Den~Berg}\ \emph {et~al.}(2023)\citenamefont {Van Den~Berg}, \citenamefont {Minev}, \citenamefont {Kandala},\ and\ \citenamefont {Temme}}]{Van2023:PEC}%
  \BibitemOpen
  \bibfield  {author} {\bibinfo {author} {\bibfnamefont {E.}~\bibnamefont {Van Den~Berg}}, \bibinfo {author} {\bibfnamefont {Z.~K.}\ \bibnamefont {Minev}}, \bibinfo {author} {\bibfnamefont {A.}~\bibnamefont {Kandala}},\ and\ \bibinfo {author} {\bibfnamefont {K.}~\bibnamefont {Temme}},\ }\bibfield  {title} {\bibinfo {title} {{Probabilistic error cancellation with sparse Pauli--Lindblad models on noisy quantum processors}},\ }\href {https://doi.org/10.1038/s41567-023-02042-2} {\bibfield  {journal} {\bibinfo  {journal} {Nat. Phys.}\ }\textbf {\bibinfo {volume} {19}},\ \bibinfo {pages} {1116} (\bibinfo {year} {2023})}\BibitemShut {NoStop}%
\bibitem [{\citenamefont {Hashim}\ \emph {et~al.}(2025)\citenamefont {Hashim}, \citenamefont {Nguyen}, \citenamefont {Goss}, \citenamefont {Marinelli}, \citenamefont {Naik}, \citenamefont {Chistolini}, \citenamefont {Hines}, \citenamefont {Marceaux}, \citenamefont {Kim} \emph {et~al.}}]{Hashim2025:BenchmarkingRvw}%
  \BibitemOpen
  \bibfield  {author} {\bibinfo {author} {\bibfnamefont {A.}~\bibnamefont {Hashim}}, \bibinfo {author} {\bibfnamefont {L.~B.}\ \bibnamefont {Nguyen}}, \bibinfo {author} {\bibfnamefont {N.}~\bibnamefont {Goss}}, \bibinfo {author} {\bibfnamefont {B.}~\bibnamefont {Marinelli}}, \bibinfo {author} {\bibfnamefont {R.~K.}\ \bibnamefont {Naik}}, \bibinfo {author} {\bibfnamefont {T.}~\bibnamefont {Chistolini}}, \bibinfo {author} {\bibfnamefont {J.}~\bibnamefont {Hines}}, \bibinfo {author} {\bibfnamefont {J.}~\bibnamefont {Marceaux}}, \bibinfo {author} {\bibfnamefont {Y.}~\bibnamefont {Kim}}, \emph {et~al.},\ }\bibfield  {title} {\bibinfo {title} {{Practical Introduction to Benchmarking and Characterization of Quantum Computers}},\ }\href {https://doi.org/10.1103/PRXQuantum.6.030202} {\bibfield  {journal} {\bibinfo  {journal} {PRX Quantum}\ }\textbf {\bibinfo {volume} {6}},\ \bibinfo {pages} {030202} (\bibinfo {year} {2025})}\BibitemShut {NoStop}%
\bibitem [{\citenamefont {Proctor}\ \emph {et~al.}(2025)\citenamefont {Proctor}, \citenamefont {Young}, \citenamefont {Baczewski},\ and\ \citenamefont {Blume-Kohout}}]{Proctor2025:Benchmarking}%
  \BibitemOpen
  \bibfield  {author} {\bibinfo {author} {\bibfnamefont {T.}~\bibnamefont {Proctor}}, \bibinfo {author} {\bibfnamefont {K.}~\bibnamefont {Young}}, \bibinfo {author} {\bibfnamefont {A.~D.}\ \bibnamefont {Baczewski}},\ and\ \bibinfo {author} {\bibfnamefont {R.}~\bibnamefont {Blume-Kohout}},\ }\bibfield  {title} {\bibinfo {title} {{Benchmarking quantum computers}},\ }\href {https://doi.org/10.1038/s42254-024-00796-z} {\bibfield  {journal} {\bibinfo  {journal} {Nat. Rev. Phys.}\ }\textbf {\bibinfo {volume} {7}},\ \bibinfo {pages} {105} (\bibinfo {year} {2025})}\BibitemShut {NoStop}%
\bibitem [{\citenamefont {Casola}\ \emph {et~al.}(2018)\citenamefont {Casola}, \citenamefont {van~der Sar},\ and\ \citenamefont {Yacoby}}]{Casola2018:nvManyBody}%
  \BibitemOpen
  \bibfield  {author} {\bibinfo {author} {\bibfnamefont {F.}~\bibnamefont {Casola}}, \bibinfo {author} {\bibfnamefont {T.}~\bibnamefont {van~der Sar}},\ and\ \bibinfo {author} {\bibfnamefont {A.}~\bibnamefont {Yacoby}},\ }\bibfield  {title} {\bibinfo {title} {{Probing condensed matter physics with magnetometry based on nitrogen-vacancy centres in diamond}},\ }\href {https://doi.org/10.1038/natrevmats.2017.88} {\bibfield  {journal} {\bibinfo  {journal} {Nat. Rev. Mat.}\ }\textbf {\bibinfo {volume} {3}},\ \bibinfo {pages} {17088} (\bibinfo {year} {2018})}\BibitemShut {NoStop}%
\bibitem [{\citenamefont {Rovny}\ \emph {et~al.}(2024)\citenamefont {Rovny}, \citenamefont {Gopalakrishnan}, \citenamefont {Jayich}, \citenamefont {Maletinsky}, \citenamefont {Demler},\ and\ \citenamefont {de~Leon}}]{Rovny2024:nvManyBody}%
  \BibitemOpen
  \bibfield  {author} {\bibinfo {author} {\bibfnamefont {J.}~\bibnamefont {Rovny}}, \bibinfo {author} {\bibfnamefont {S.}~\bibnamefont {Gopalakrishnan}}, \bibinfo {author} {\bibfnamefont {A.~C.~B.}\ \bibnamefont {Jayich}}, \bibinfo {author} {\bibfnamefont {P.}~\bibnamefont {Maletinsky}}, \bibinfo {author} {\bibfnamefont {E.}~\bibnamefont {Demler}},\ and\ \bibinfo {author} {\bibfnamefont {N.~P.}\ \bibnamefont {de~Leon}},\ }\bibfield  {title} {\bibinfo {title} {{Nanoscale diamond quantum sensors for many-body physics}},\ }\href {https://doi.org/10.1038/s42254-024-00775-4} {\bibfield  {journal} {\bibinfo  {journal} {Nat. Rev. Phys.}\ }\textbf {\bibinfo {volume} {6}},\ \bibinfo {pages} {753} (\bibinfo {year} {2024})}\BibitemShut {NoStop}%
\bibitem [{\citenamefont {Rovny}\ \emph {et~al.}(2022)\citenamefont {Rovny}, \citenamefont {Yuan}, \citenamefont {Fitzpatrick}, \citenamefont {Abdalla}, \citenamefont {Futamura}, \citenamefont {Fox}, \citenamefont {Cambria}, \citenamefont {Kolkowitz},\ and\ \citenamefont {de~Leon}}]{Rovny2022NVcovariance}%
  \BibitemOpen
  \bibfield  {author} {\bibinfo {author} {\bibfnamefont {J.}~\bibnamefont {Rovny}}, \bibinfo {author} {\bibfnamefont {Z.}~\bibnamefont {Yuan}}, \bibinfo {author} {\bibfnamefont {M.}~\bibnamefont {Fitzpatrick}}, \bibinfo {author} {\bibfnamefont {A.~I.}\ \bibnamefont {Abdalla}}, \bibinfo {author} {\bibfnamefont {L.}~\bibnamefont {Futamura}}, \bibinfo {author} {\bibfnamefont {C.}~\bibnamefont {Fox}}, \bibinfo {author} {\bibfnamefont {M.~C.}\ \bibnamefont {Cambria}}, \bibinfo {author} {\bibfnamefont {S.}~\bibnamefont {Kolkowitz}},\ and\ \bibinfo {author} {\bibfnamefont {N.~P.}\ \bibnamefont {de~Leon}},\ }\bibfield  {title} {\bibinfo {title} {{Nanoscale covariance magnetometry with diamond quantum sensors}},\ }\href {https://doi.org/10.1126/science.ade9858} {\bibfield  {journal} {\bibinfo  {journal} {Science}\ }\textbf {\bibinfo {volume} {378}},\ \bibinfo {pages} {1301} (\bibinfo {year} {2022})}\BibitemShut {NoStop}%
\bibitem [{\citenamefont {Ziffer}\ \emph {et~al.}(2024)\citenamefont {Ziffer}, \citenamefont {Machado}, \citenamefont {Ursprung}, \citenamefont {Lozovoi}, \citenamefont {Tazi}, \citenamefont {Yuan}, \citenamefont {Ziebel}, \citenamefont {Delord}, \citenamefont {Zeng} \emph {et~al.}}]{Ziffer2024:qnsCriticality}%
  \BibitemOpen
  \bibfield  {author} {\bibinfo {author} {\bibfnamefont {M.~E.}\ \bibnamefont {Ziffer}}, \bibinfo {author} {\bibfnamefont {F.}~\bibnamefont {Machado}}, \bibinfo {author} {\bibfnamefont {B.}~\bibnamefont {Ursprung}}, \bibinfo {author} {\bibfnamefont {A.}~\bibnamefont {Lozovoi}}, \bibinfo {author} {\bibfnamefont {A.~B.}\ \bibnamefont {Tazi}}, \bibinfo {author} {\bibfnamefont {Z.}~\bibnamefont {Yuan}}, \bibinfo {author} {\bibfnamefont {M.~E.}\ \bibnamefont {Ziebel}}, \bibinfo {author} {\bibfnamefont {T.}~\bibnamefont {Delord}}, \bibinfo {author} {\bibfnamefont {N.}~\bibnamefont {Zeng}}, \emph {et~al.},\ }\href@noop {} {\bibinfo {title} {{Quantum Noise Spectroscopy of Criticality in an Atomically Thin Magnet}}} (\bibinfo {year} {2024}),\ \Eprint {https://arxiv.org/abs/2407.05614} {arXiv:2407.05614 [cond-mat.mes-hall]} \BibitemShut {NoStop}%
\bibitem [{\citenamefont {Cheng}\ \emph {et~al.}(2025)\citenamefont {Cheng}, \citenamefont {Kazi}, \citenamefont {Rovny}, \citenamefont {Zhang}, \citenamefont {Nassar}, \citenamefont {Thompson},\ and\ \citenamefont {de~Leon}}]{Cheng2025:mplex}%
  \BibitemOpen
  \bibfield  {author} {\bibinfo {author} {\bibfnamefont {K.-H.}\ \bibnamefont {Cheng}}, \bibinfo {author} {\bibfnamefont {Z.}~\bibnamefont {Kazi}}, \bibinfo {author} {\bibfnamefont {J.}~\bibnamefont {Rovny}}, \bibinfo {author} {\bibfnamefont {B.}~\bibnamefont {Zhang}}, \bibinfo {author} {\bibfnamefont {L.~S.}\ \bibnamefont {Nassar}}, \bibinfo {author} {\bibfnamefont {J.~D.}\ \bibnamefont {Thompson}},\ and\ \bibinfo {author} {\bibfnamefont {N.~P.}\ \bibnamefont {de~Leon}},\ }\bibfield  {title} {\bibinfo {title} {{Massively Multiplexed Nanoscale Magnetometry with Diamond Quantum Sensors}},\ }\href {https://doi.org/10.1103/t8fz-3tzs} {\bibfield  {journal} {\bibinfo  {journal} {Phys. Rev. X}\ }\textbf {\bibinfo {volume} {15}},\ \bibinfo {pages} {031014} (\bibinfo {year} {2025})}\BibitemShut {NoStop}%
\bibitem [{\citenamefont {Rovny}\ \emph {et~al.}(2025)\citenamefont {Rovny}, \citenamefont {Kolkowitz},\ and\ \citenamefont {de~Leon}}]{Rovny2025:nvEntResource}%
  \BibitemOpen
  \bibfield  {author} {\bibinfo {author} {\bibfnamefont {J.}~\bibnamefont {Rovny}}, \bibinfo {author} {\bibfnamefont {S.}~\bibnamefont {Kolkowitz}},\ and\ \bibinfo {author} {\bibfnamefont {N.~P.}\ \bibnamefont {de~Leon}},\ }\bibfield  {title} {\bibinfo {title} {{Multi-qubit nanoscale sensing with entanglement as a resource}},\ }\href {https://doi.org/10.1038/s41586-025-09760-y} {\bibfield  {journal} {\bibinfo  {journal} {Nature}\ }\textbf {\bibinfo {volume} {647}},\ \bibinfo {pages} {876–882} (\bibinfo {year} {2025})}\BibitemShut {NoStop}%
\bibitem [{\citenamefont {Zhou}\ \emph {et~al.}(2025)\citenamefont {Zhou}, \citenamefont {Wang}, \citenamefont {Ye}, \citenamefont {Sun}, \citenamefont {Guo}, \citenamefont {Han}, \citenamefont {Chai}, \citenamefont {Ji}, \citenamefont {Xia} \emph {et~al.}}]{Zhou2025:nvEntResource}%
  \BibitemOpen
  \bibfield  {author} {\bibinfo {author} {\bibfnamefont {X.}~\bibnamefont {Zhou}}, \bibinfo {author} {\bibfnamefont {M.}~\bibnamefont {Wang}}, \bibinfo {author} {\bibfnamefont {X.}~\bibnamefont {Ye}}, \bibinfo {author} {\bibfnamefont {H.}~\bibnamefont {Sun}}, \bibinfo {author} {\bibfnamefont {Y.}~\bibnamefont {Guo}}, \bibinfo {author} {\bibfnamefont {S.}~\bibnamefont {Han}}, \bibinfo {author} {\bibfnamefont {Z.}~\bibnamefont {Chai}}, \bibinfo {author} {\bibfnamefont {W.}~\bibnamefont {Ji}}, \bibinfo {author} {\bibfnamefont {K.}~\bibnamefont {Xia}}, \emph {et~al.},\ }\bibfield  {title} {\bibinfo {title} {{Entanglement-enhanced nanoscale single-spin sensing}},\ }\href {https://doi.org/10.1038/s41586-025-09790-6} {\bibfield  {journal} {\bibinfo  {journal} {Nature}\ }\textbf {\bibinfo {volume} {647}},\ \bibinfo {pages} {883} (\bibinfo {year} {2025})}\BibitemShut {NoStop}%
\bibitem [{\citenamefont {Ye}\ and\ \citenamefont {Zoller}(2024)}]{YeZoller2024:Essay}%
  \BibitemOpen
  \bibfield  {author} {\bibinfo {author} {\bibfnamefont {J.}~\bibnamefont {Ye}}\ and\ \bibinfo {author} {\bibfnamefont {P.}~\bibnamefont {Zoller}},\ }\bibfield  {title} {\bibinfo {title} {{Essay: Quantum Sensing with Atomic, Molecular, and Optical Platforms for Fundamental Physics}},\ }\href {https://doi.org/10.1103/PhysRevLett.132.190001} {\bibfield  {journal} {\bibinfo  {journal} {Phys. Rev. Lett.}\ }\textbf {\bibinfo {volume} {132}},\ \bibinfo {pages} {190001} (\bibinfo {year} {2024})}\BibitemShut {NoStop}%
\bibitem [{\citenamefont {Bass}\ and\ \citenamefont {Doser}(2024)}]{Bass2024:NatRvw}%
  \BibitemOpen
  \bibfield  {author} {\bibinfo {author} {\bibfnamefont {S.~D.}\ \bibnamefont {Bass}}\ and\ \bibinfo {author} {\bibfnamefont {M.}~\bibnamefont {Doser}},\ }\bibfield  {title} {\bibinfo {title} {{Quantum sensing for particle physics}},\ }\href {https://doi.org/10.1038/s42254-024-00714-3} {\bibfield  {journal} {\bibinfo  {journal} {Nat. Rev. Phys.}\ }\textbf {\bibinfo {volume} {6}},\ \bibinfo {pages} {329} (\bibinfo {year} {2024})}\BibitemShut {NoStop}%
\bibitem [{\citenamefont {Roy}\ and\ \citenamefont {Braunstein}(2008)}]{Roy2008:ExpMetrology}%
  \BibitemOpen
  \bibfield  {author} {\bibinfo {author} {\bibfnamefont {S.~M.}\ \bibnamefont {Roy}}\ and\ \bibinfo {author} {\bibfnamefont {S.~L.}\ \bibnamefont {Braunstein}},\ }\bibfield  {title} {\bibinfo {title} {{Exponentially Enhanced Quantum Metrology}},\ }\href {https://doi.org/10.1103/PhysRevLett.100.220501} {\bibfield  {journal} {\bibinfo  {journal} {Phys. Rev. Lett.}\ }\textbf {\bibinfo {volume} {100}},\ \bibinfo {pages} {220501} (\bibinfo {year} {2008})}\BibitemShut {NoStop}%
\bibitem [{\citenamefont {Boixo}\ \emph {et~al.}(2007)\citenamefont {Boixo}, \citenamefont {Flammia}, \citenamefont {Caves},\ and\ \citenamefont {Geremia}}]{Boixo2007:QuEstim}%
  \BibitemOpen
  \bibfield  {author} {\bibinfo {author} {\bibfnamefont {S.}~\bibnamefont {Boixo}}, \bibinfo {author} {\bibfnamefont {S.~T.}\ \bibnamefont {Flammia}}, \bibinfo {author} {\bibfnamefont {C.~M.}\ \bibnamefont {Caves}},\ and\ \bibinfo {author} {\bibfnamefont {J.}~\bibnamefont {Geremia}},\ }\bibfield  {title} {\bibinfo {title} {{Generalized Limits for Single-Parameter Quantum Estimation}},\ }\href {https://doi.org/10.1103/PhysRevLett.98.090401} {\bibfield  {journal} {\bibinfo  {journal} {Phys. Rev. Lett.}\ }\textbf {\bibinfo {volume} {98}},\ \bibinfo {pages} {090401} (\bibinfo {year} {2007})}\BibitemShut {NoStop}%
\bibitem [{\citenamefont {Correa}\ \emph {et~al.}(2015)\citenamefont {Correa}, \citenamefont {Mehboudi}, \citenamefont {Adesso},\ and\ \citenamefont {Sanpera}}]{Correa2015:OptThermometry}%
  \BibitemOpen
  \bibfield  {author} {\bibinfo {author} {\bibfnamefont {L.~A.}\ \bibnamefont {Correa}}, \bibinfo {author} {\bibfnamefont {M.}~\bibnamefont {Mehboudi}}, \bibinfo {author} {\bibfnamefont {G.}~\bibnamefont {Adesso}},\ and\ \bibinfo {author} {\bibfnamefont {A.}~\bibnamefont {Sanpera}},\ }\bibfield  {title} {\bibinfo {title} {{Individual Quantum Probes for Optimal Thermometry}},\ }\href {https://doi.org/10.1103/PhysRevLett.114.220405} {\bibfield  {journal} {\bibinfo  {journal} {Phys. Rev. Lett.}\ }\textbf {\bibinfo {volume} {114}},\ \bibinfo {pages} {220405} (\bibinfo {year} {2015})}\BibitemShut {NoStop}%
\bibitem [{\citenamefont {Sekatski}\ and\ \citenamefont {Perarnau-Llobet}(2022)}]{Sekatski2022:Thermometry}%
  \BibitemOpen
  \bibfield  {author} {\bibinfo {author} {\bibfnamefont {P.}~\bibnamefont {Sekatski}}\ and\ \bibinfo {author} {\bibfnamefont {M.}~\bibnamefont {Perarnau-Llobet}},\ }\bibfield  {title} {\bibinfo {title} {{Optimal nonequilibrium thermometry in Markovian environments}},\ }\href {https://doi.org/10.22331/q-2022-12-07-869} {\bibfield  {journal} {\bibinfo  {journal} {Quantum}\ }\textbf {\bibinfo {volume} {6}},\ \bibinfo {pages} {869} (\bibinfo {year} {2022})}\BibitemShut {NoStop}%
\bibitem [{\citenamefont {Gardner}\ \emph {et~al.}(2025{\natexlab{a}})\citenamefont {Gardner}, \citenamefont {Haine}, \citenamefont {Hope}, \citenamefont {Chen},\ and\ \citenamefont {Gefen}}]{Gardner2025:FastReset}%
  \BibitemOpen
  \bibfield  {author} {\bibinfo {author} {\bibfnamefont {J.~W.}\ \bibnamefont {Gardner}}, \bibinfo {author} {\bibfnamefont {S.~A.}\ \bibnamefont {Haine}}, \bibinfo {author} {\bibfnamefont {J.~J.}\ \bibnamefont {Hope}}, \bibinfo {author} {\bibfnamefont {Y.}~\bibnamefont {Chen}},\ and\ \bibinfo {author} {\bibfnamefont {T.}~\bibnamefont {Gefen}},\ }\bibfield  {title} {\bibinfo {title} {{Lindblad estimation with fast and precise quantum control}},\ }\href {https://doi.org/10.1103/6yzb-43rs} {\bibfield  {journal} {\bibinfo  {journal} {Phys. Rev. Appl.}\ }\textbf {\bibinfo {volume} {24}},\ \bibinfo {pages} {044055} (\bibinfo {year} {2025}{\natexlab{a}})}\BibitemShut {NoStop}%
\bibitem [{\citenamefont {Sidhu}\ and\ \citenamefont {Kok}(2020)}]{Sidhu2020:QFIrvw}%
  \BibitemOpen
  \bibfield  {author} {\bibinfo {author} {\bibfnamefont {J.~S.}\ \bibnamefont {Sidhu}}\ and\ \bibinfo {author} {\bibfnamefont {P.}~\bibnamefont {Kok}},\ }\bibfield  {title} {\bibinfo {title} {{Geometric perspective on quantum parameter estimation}},\ }\href {https://doi.org/10.1116/1.5119961} {\bibfield  {journal} {\bibinfo  {journal} {AVS Quantum Sci.}\ }\textbf {\bibinfo {volume} {2}},\ \bibinfo {pages} {014701} (\bibinfo {year} {2020})}\BibitemShut {NoStop}%
\bibitem [{\citenamefont {Liu}\ \emph {et~al.}(2020)\citenamefont {Liu}, \citenamefont {Yuan}, \citenamefont {Lu},\ and\ \citenamefont {Wang}}]{LiuYuanLuWang2020}%
  \BibitemOpen
  \bibfield  {author} {\bibinfo {author} {\bibfnamefont {J.}~\bibnamefont {Liu}}, \bibinfo {author} {\bibfnamefont {H.}~\bibnamefont {Yuan}}, \bibinfo {author} {\bibfnamefont {X.}~\bibnamefont {Lu}},\ and\ \bibinfo {author} {\bibfnamefont {X.}~\bibnamefont {Wang}},\ }\bibfield  {title} {\bibinfo {title} {{Quantum Fisher Information Matrix and Multi-Parameter Estimation}},\ }\href {https://doi.org/10.1088/1751-8121/ab5d4d} {\bibfield  {journal} {\bibinfo  {journal} {J. Phys. A: Math. Theor.}\ }\textbf {\bibinfo {volume} {53}},\ \bibinfo {pages} {023001} (\bibinfo {year} {2020})}\BibitemShut {NoStop}%
\bibitem [{\citenamefont {Meyer}(2021)}]{Meyer2021}%
  \BibitemOpen
  \bibfield  {author} {\bibinfo {author} {\bibfnamefont {J.~J.}\ \bibnamefont {Meyer}},\ }\bibfield  {title} {\bibinfo {title} {{Fisher Information in Noisy Intermediate-Scale Quantum Applications}},\ }\href {https://doi.org/10.22331/q-2021-09-09-539} {\bibfield  {journal} {\bibinfo  {journal} {Quantum}\ }\textbf {\bibinfo {volume} {5}},\ \bibinfo {pages} {539} (\bibinfo {year} {2021})}\BibitemShut {NoStop}%
\bibitem [{\citenamefont {Pezzè}\ and\ \citenamefont {Smerzi}(2025)}]{Pezze2025:AvMultiQuEst}%
  \BibitemOpen
  \bibfield  {author} {\bibinfo {author} {\bibfnamefont {L.}~\bibnamefont {Pezzè}}\ and\ \bibinfo {author} {\bibfnamefont {A.}~\bibnamefont {Smerzi}},\ }\href@noop {} {\bibinfo {title} {{Advances in multiparameter quantum sensing and metrology}}} (\bibinfo {year} {2025}),\ \Eprint {https://arxiv.org/abs/2502.17396} {arXiv:2502.17396 [quant-ph]} \BibitemShut {NoStop}%
\bibitem [{\citenamefont {Fujiwara}\ and\ \citenamefont {Nagaoka}(1995)}]{Fujiwara1995:QFImetric}%
  \BibitemOpen
  \bibfield  {author} {\bibinfo {author} {\bibfnamefont {A.}~\bibnamefont {Fujiwara}}\ and\ \bibinfo {author} {\bibfnamefont {H.}~\bibnamefont {Nagaoka}},\ }\bibfield  {title} {\bibinfo {title} {{Quantum Fisher metric and estimation for pure state models}},\ }\href {https://doi.org/10.1016/0375-9601(95)00269-9} {\bibfield  {journal} {\bibinfo  {journal} {Phys. Lett. A}\ }\textbf {\bibinfo {volume} {201}},\ \bibinfo {pages} {119} (\bibinfo {year} {1995})}\BibitemShut {NoStop}%
\bibitem [{\citenamefont {Ragy}\ \emph {et~al.}(2016)\citenamefont {Ragy}, \citenamefont {Jarzyna},\ and\ \citenamefont {Demkowicz-Dobrza\ifmmode~\acute{n}\else \'{n}\fi{}ski}}]{Ragy2016compatibility}%
  \BibitemOpen
  \bibfield  {author} {\bibinfo {author} {\bibfnamefont {S.}~\bibnamefont {Ragy}}, \bibinfo {author} {\bibfnamefont {M.}~\bibnamefont {Jarzyna}},\ and\ \bibinfo {author} {\bibfnamefont {R.}~\bibnamefont {Demkowicz-Dobrza\ifmmode~\acute{n}\else \'{n}\fi{}ski}},\ }\bibfield  {title} {\bibinfo {title} {{Compatibility in multiparameter quantum metrology}},\ }\href {https://doi.org/10.1103/PhysRevA.94.052108} {\bibfield  {journal} {\bibinfo  {journal} {Phys. Rev. A}\ }\textbf {\bibinfo {volume} {94}},\ \bibinfo {pages} {052108} (\bibinfo {year} {2016})}\BibitemShut {NoStop}%
\bibitem [{\citenamefont {Pezz\`e}\ \emph {et~al.}(2017)\citenamefont {Pezz\`e}, \citenamefont {Ciampini}, \citenamefont {Spagnolo}, \citenamefont {Humphreys}, \citenamefont {Datta}, \citenamefont {Walmsley}, \citenamefont {Barbieri}, \citenamefont {Sciarrino},\ and\ \citenamefont {Smerzi}}]{Pezze2017:Multiparam}%
  \BibitemOpen
  \bibfield  {author} {\bibinfo {author} {\bibfnamefont {L.}~\bibnamefont {Pezz\`e}}, \bibinfo {author} {\bibfnamefont {M.~A.}\ \bibnamefont {Ciampini}}, \bibinfo {author} {\bibfnamefont {N.}~\bibnamefont {Spagnolo}}, \bibinfo {author} {\bibfnamefont {P.~C.}\ \bibnamefont {Humphreys}}, \bibinfo {author} {\bibfnamefont {A.}~\bibnamefont {Datta}}, \bibinfo {author} {\bibfnamefont {I.~A.}\ \bibnamefont {Walmsley}}, \bibinfo {author} {\bibfnamefont {M.}~\bibnamefont {Barbieri}}, \bibinfo {author} {\bibfnamefont {F.}~\bibnamefont {Sciarrino}},\ and\ \bibinfo {author} {\bibfnamefont {A.}~\bibnamefont {Smerzi}},\ }\bibfield  {title} {\bibinfo {title} {{Optimal Measurements for Simultaneous Quantum Estimation of Multiple Phases}},\ }\href {https://doi.org/10.1103/PhysRevLett.119.130504} {\bibfield  {journal} {\bibinfo  {journal} {Phys. Rev. Lett.}\ }\textbf {\bibinfo {volume} {119}},\ \bibinfo {pages} {130504} (\bibinfo {year} {2017})}\BibitemShut {NoStop}%
\bibitem [{\citenamefont {Albarelli}\ and\ \citenamefont {Demkowicz-Dobrza\ifmmode~\acute{n}\else \'{n}\fi{}ski}(2022)}]{Albarelli2022:Incompatibility}%
  \BibitemOpen
  \bibfield  {author} {\bibinfo {author} {\bibfnamefont {F.}~\bibnamefont {Albarelli}}\ and\ \bibinfo {author} {\bibfnamefont {R.}~\bibnamefont {Demkowicz-Dobrza\ifmmode~\acute{n}\else \'{n}\fi{}ski}},\ }\bibfield  {title} {\bibinfo {title} {{Probe Incompatibility in Multiparameter Noisy Quantum Metrology}},\ }\href {https://doi.org/10.1103/PhysRevX.12.011039} {\bibfield  {journal} {\bibinfo  {journal} {Phys. Rev. X}\ }\textbf {\bibinfo {volume} {12}},\ \bibinfo {pages} {011039} (\bibinfo {year} {2022})}\BibitemShut {NoStop}%
\bibitem [{\citenamefont {Sidhu}\ \emph {et~al.}(2021)\citenamefont {Sidhu}, \citenamefont {Ouyang}, \citenamefont {Campbell},\ and\ \citenamefont {Kok}}]{Sidhu2021:HolevoBound}%
  \BibitemOpen
  \bibfield  {author} {\bibinfo {author} {\bibfnamefont {J.~S.}\ \bibnamefont {Sidhu}}, \bibinfo {author} {\bibfnamefont {Y.}~\bibnamefont {Ouyang}}, \bibinfo {author} {\bibfnamefont {E.~T.}\ \bibnamefont {Campbell}},\ and\ \bibinfo {author} {\bibfnamefont {P.}~\bibnamefont {Kok}},\ }\bibfield  {title} {\bibinfo {title} {{Tight Bounds on the Simultaneous Estimation of Incompatible Parameters}},\ }\href {https://doi.org/10.1103/PhysRevX.11.011028} {\bibfield  {journal} {\bibinfo  {journal} {Phys. Rev. X}\ }\textbf {\bibinfo {volume} {11}},\ \bibinfo {pages} {011028} (\bibinfo {year} {2021})}\BibitemShut {NoStop}%
\bibitem [{\citenamefont {Stinespring}(1955)}]{Stinespring1955}%
  \BibitemOpen
  \bibfield  {author} {\bibinfo {author} {\bibfnamefont {W.~F.}\ \bibnamefont {Stinespring}},\ }\bibfield  {title} {\bibinfo {title} {{Positive Functions on C$^{*}$-Algebras}},\ }\href {https://doi.org/10.2307/2032342} {\bibfield  {journal} {\bibinfo  {journal} {Proc. Amer. Math. Soc.}\ }\textbf {\bibinfo {volume} {6}},\ \bibinfo {pages} {211} (\bibinfo {year} {1955})}\BibitemShut {NoStop}%
\bibitem [{\citenamefont {Schumacher}\ and\ \citenamefont {Nielsen}(1996)}]{Nielsen1996}%
  \BibitemOpen
  \bibfield  {author} {\bibinfo {author} {\bibfnamefont {B.}~\bibnamefont {Schumacher}}\ and\ \bibinfo {author} {\bibfnamefont {M.~A.}\ \bibnamefont {Nielsen}},\ }\bibfield  {title} {\bibinfo {title} {{Quantum data processing and error correction}},\ }\href {https://doi.org/10.1103/PhysRevA.54.2629} {\bibfield  {journal} {\bibinfo  {journal} {Phys. Rev. A}\ }\textbf {\bibinfo {volume} {54}},\ \bibinfo {pages} {2629} (\bibinfo {year} {1996})}\BibitemShut {NoStop}%
\bibitem [{\citenamefont {Uhlmann}(1976)}]{Uhlmann1976}%
  \BibitemOpen
  \bibfield  {author} {\bibinfo {author} {\bibfnamefont {A.}~\bibnamefont {Uhlmann}},\ }\bibfield  {title} {\bibinfo {title} {{The “transition probability” in the state space of a *-algebra}},\ }\href {https://doi.org/10.1016/0034-4877(76)90060-4} {\bibfield  {journal} {\bibinfo  {journal} {Rep. Math. Phys.}\ }\textbf {\bibinfo {volume} {9}},\ \bibinfo {pages} {273} (\bibinfo {year} {1976})}\BibitemShut {NoStop}%
\bibitem [{\citenamefont {Lindblad}(1976)}]{Lindblad1976:Generators}%
  \BibitemOpen
  \bibfield  {author} {\bibinfo {author} {\bibfnamefont {G.}~\bibnamefont {Lindblad}},\ }\bibfield  {title} {\bibinfo {title} {{On the generators of quantum dynamical semigroups}},\ }\href {https://doi.org/10.1007/BF01608499} {\bibfield  {journal} {\bibinfo  {journal} {Commun. Math. Phys.}\ }\textbf {\bibinfo {volume} {48}},\ \bibinfo {pages} {119} (\bibinfo {year} {1976})}\BibitemShut {NoStop}%
\bibitem [{\citenamefont {Gorini}\ \emph {et~al.}(1976)\citenamefont {Gorini}, \citenamefont {Kossakowski},\ and\ \citenamefont {Sudarshan}}]{Gorini1976}%
  \BibitemOpen
  \bibfield  {author} {\bibinfo {author} {\bibfnamefont {V.}~\bibnamefont {Gorini}}, \bibinfo {author} {\bibfnamefont {A.}~\bibnamefont {Kossakowski}},\ and\ \bibinfo {author} {\bibfnamefont {E.~C.~G.}\ \bibnamefont {Sudarshan}},\ }\bibfield  {title} {\bibinfo {title} {{Completely positive dynamical semigroups of $N$-level systems}},\ }\href {https://doi.org/10.1063/1.522979} {\bibfield  {journal} {\bibinfo  {journal} {J. Math. Phys.}\ }\textbf {\bibinfo {volume} {17}},\ \bibinfo {pages} {821} (\bibinfo {year} {1976})}\BibitemShut {NoStop}%
\bibitem [{\citenamefont {Breuer}\ and\ \citenamefont {Petruccione}(2007)}]{Breuer2007}%
  \BibitemOpen
  \bibfield  {author} {\bibinfo {author} {\bibfnamefont {H.-P.}\ \bibnamefont {Breuer}}\ and\ \bibinfo {author} {\bibfnamefont {F.}~\bibnamefont {Petruccione}},\ }\href@noop {} {\emph {\bibinfo {title} {{The Theory of Open Quantum Systems}}}}\ (\bibinfo  {publisher} {Oxford University Press},\ \bibinfo {year} {2007})\BibitemShut {NoStop}%
\bibitem [{\citenamefont {Albert}\ \emph {et~al.}(2016)\citenamefont {Albert}, \citenamefont {Bradlyn}, \citenamefont {Fraas},\ and\ \citenamefont {Jiang}}]{Albert2016:GeometryLindbladians}%
  \BibitemOpen
  \bibfield  {author} {\bibinfo {author} {\bibfnamefont {V.~V.}\ \bibnamefont {Albert}}, \bibinfo {author} {\bibfnamefont {B.}~\bibnamefont {Bradlyn}}, \bibinfo {author} {\bibfnamefont {M.}~\bibnamefont {Fraas}},\ and\ \bibinfo {author} {\bibfnamefont {L.}~\bibnamefont {Jiang}},\ }\bibfield  {title} {\bibinfo {title} {{Geometry and Response of Lindbladians}},\ }\href {https://doi.org/10.1103/PhysRevX.6.041031} {\bibfield  {journal} {\bibinfo  {journal} {Phys. Rev. X}\ }\textbf {\bibinfo {volume} {6}},\ \bibinfo {pages} {041031} (\bibinfo {year} {2016})}\BibitemShut {NoStop}%
\bibitem [{\citenamefont {Fujiwara}\ and\ \citenamefont {Imai}(2008)}]{Fujiwara2008:fibre}%
  \BibitemOpen
  \bibfield  {author} {\bibinfo {author} {\bibfnamefont {A.}~\bibnamefont {Fujiwara}}\ and\ \bibinfo {author} {\bibfnamefont {H.}~\bibnamefont {Imai}},\ }\bibfield  {title} {\bibinfo {title} {{A fibre bundle over manifolds of quantum channels and its application to quantum statistics}},\ }\href {https://doi.org/10.1088/1751-8113/41/25/255304} {\bibfield  {journal} {\bibinfo  {journal} {J. Phys. A: Math. Theor.}\ }\textbf {\bibinfo {volume} {41}},\ \bibinfo {pages} {255304} (\bibinfo {year} {2008})}\BibitemShut {NoStop}%
\bibitem [{\citenamefont {Escher}\ \emph {et~al.}(2011)\citenamefont {Escher}, \citenamefont {de~Matos~Filho},\ and\ \citenamefont {Davidovich}}]{Escher2011:Framework}%
  \BibitemOpen
  \bibfield  {author} {\bibinfo {author} {\bibfnamefont {B.~M.}\ \bibnamefont {Escher}}, \bibinfo {author} {\bibfnamefont {R.~L.}\ \bibnamefont {de~Matos~Filho}},\ and\ \bibinfo {author} {\bibfnamefont {L.}~\bibnamefont {Davidovich}},\ }\bibfield  {title} {\bibinfo {title} {{General framework for estimating the ultimate precision limit in noisy quantum-enhanced metrology}},\ }\href {https://doi.org/10.1038/nphys1958} {\bibfield  {journal} {\bibinfo  {journal} {Nat. Phys.}\ }\textbf {\bibinfo {volume} {7}},\ \bibinfo {pages} {406} (\bibinfo {year} {2011})}\BibitemShut {NoStop}%
\bibitem [{\citenamefont {Demkowicz-Dobrzański}\ \emph {et~al.}(2012)\citenamefont {Demkowicz-Dobrzański}, \citenamefont {Kołodyński},\ and\ \citenamefont {Guţă}}]{Dobrza2012ElusiveHeisenberg}%
  \BibitemOpen
  \bibfield  {author} {\bibinfo {author} {\bibfnamefont {R.}~\bibnamefont {Demkowicz-Dobrzański}}, \bibinfo {author} {\bibfnamefont {J.}~\bibnamefont {Kołodyński}},\ and\ \bibinfo {author} {\bibfnamefont {M.}~\bibnamefont {Guţă}},\ }\bibfield  {title} {\bibinfo {title} {{The elusive Heisenberg limit in quantum-enhanced metrology}},\ }\href {https://doi.org/10.1038/ncomms2067} {\bibfield  {journal} {\bibinfo  {journal} {Nat. Commun.}\ }\textbf {\bibinfo {volume} {3}},\ \bibinfo {pages} {1063} (\bibinfo {year} {2012})}\BibitemShut {NoStop}%
\bibitem [{\citenamefont {Ko{\l}ody{\'n}ski}\ and\ \citenamefont {Demkowicz-Dobrza{\'n}ski}(2013)}]{Kolodynski2013:EfficientTools}%
  \BibitemOpen
  \bibfield  {author} {\bibinfo {author} {\bibfnamefont {J.}~\bibnamefont {Ko{\l}ody{\'n}ski}}\ and\ \bibinfo {author} {\bibfnamefont {R.}~\bibnamefont {Demkowicz-Dobrza{\'n}ski}},\ }\bibfield  {title} {\bibinfo {title} {{Efficient tools for quantum metrology with uncorrelated noise}},\ }\href {https://doi.org/10.1088/1367-2630/15/7/073043} {\bibfield  {journal} {\bibinfo  {journal} {New J. Phys.}\ }\textbf {\bibinfo {volume} {15}},\ \bibinfo {pages} {073043} (\bibinfo {year} {2013})}\BibitemShut {NoStop}%
\bibitem [{\citenamefont {Catana}\ \emph {et~al.}(2015)\citenamefont {Catana}, \citenamefont {Bouten},\ and\ \citenamefont {Guţă}}]{Catana2015}%
  \BibitemOpen
  \bibfield  {author} {\bibinfo {author} {\bibfnamefont {C.}~\bibnamefont {Catana}}, \bibinfo {author} {\bibfnamefont {L.}~\bibnamefont {Bouten}},\ and\ \bibinfo {author} {\bibfnamefont {M.}~\bibnamefont {Guţă}},\ }\bibfield  {title} {\bibinfo {title} {{Fisher informations and local asymptotic normality for continuous-time quantum Markov processes}},\ }\href {https://doi.org/10.1088/1751-8113/48/36/365301} {\bibfield  {journal} {\bibinfo  {journal} {J. Phys. A: Math. Theor.}\ }\textbf {\bibinfo {volume} {48}},\ \bibinfo {pages} {365301} (\bibinfo {year} {2015})}\BibitemShut {NoStop}%
\bibitem [{\citenamefont {Guta}\ and\ \citenamefont {Kiukas}(2017)}]{Guta2017}%
  \BibitemOpen
  \bibfield  {author} {\bibinfo {author} {\bibfnamefont {M.}~\bibnamefont {Guta}}\ and\ \bibinfo {author} {\bibfnamefont {J.}~\bibnamefont {Kiukas}},\ }\bibfield  {title} {\bibinfo {title} {{Information geometry and local asymptotic normality for multi-parameter estimation of quantum Markov dynamics}},\ }\bibfield  {journal} {\bibinfo  {journal} {J. Math. Phys.}\ }\textbf {\bibinfo {volume} {58}},\ \href {https://doi.org/10.1063/1.4982958} {10.1063/1.4982958} (\bibinfo {year} {2017})\BibitemShut {NoStop}%
\bibitem [{\citenamefont {Kurdzia\l{}ek}\ \emph {et~al.}(2023)\citenamefont {Kurdzia\l{}ek}, \citenamefont {G\'orecki}, \citenamefont {Albarelli},\ and\ \citenamefont {Demkowicz-Dobrza\ifmmode~\acute{n}\else \'{n}\fi{}ski}}]{Gorecki2023:CausalAdaptive}%
  \BibitemOpen
  \bibfield  {author} {\bibinfo {author} {\bibfnamefont {S.}~\bibnamefont {Kurdzia\l{}ek}}, \bibinfo {author} {\bibfnamefont {W.}~\bibnamefont {G\'orecki}}, \bibinfo {author} {\bibfnamefont {F.}~\bibnamefont {Albarelli}},\ and\ \bibinfo {author} {\bibfnamefont {R.}~\bibnamefont {Demkowicz-Dobrza\ifmmode~\acute{n}\else \'{n}\fi{}ski}},\ }\bibfield  {title} {\bibinfo {title} {{Using Adaptiveness and Causal Superpositions Against Noise in Quantum Metrology}},\ }\href {https://doi.org/10.1103/PhysRevLett.131.090801} {\bibfield  {journal} {\bibinfo  {journal} {Phys. Rev. Lett.}\ }\textbf {\bibinfo {volume} {131}},\ \bibinfo {pages} {090801} (\bibinfo {year} {2023})}\BibitemShut {NoStop}%
\bibitem [{\citenamefont {Wan}\ and\ \citenamefont {Lasenby}(2022)}]{Wan2022:AdaptiveMarkov}%
  \BibitemOpen
  \bibfield  {author} {\bibinfo {author} {\bibfnamefont {K.}~\bibnamefont {Wan}}\ and\ \bibinfo {author} {\bibfnamefont {R.}~\bibnamefont {Lasenby}},\ }\bibfield  {title} {\bibinfo {title} {{Bounds on adaptive quantum metrology under Markovian noise}},\ }\href {https://doi.org/10.1103/PhysRevResearch.4.033092} {\bibfield  {journal} {\bibinfo  {journal} {Phys. Rev. Res.}\ }\textbf {\bibinfo {volume} {4}},\ \bibinfo {pages} {033092} (\bibinfo {year} {2022})}\BibitemShut {NoStop}%
\bibitem [{\citenamefont {Das}\ \emph {et~al.}(2025)\citenamefont {Das}, \citenamefont {G\'orecki},\ and\ \citenamefont {Demkowicz-Dobrza\ifmmode~\acute{n}\else \'{n}\fi{}ski}}]{Das2025:QFItimeScalings}%
  \BibitemOpen
  \bibfield  {author} {\bibinfo {author} {\bibfnamefont {A.}~\bibnamefont {Das}}, \bibinfo {author} {\bibfnamefont {W.}~\bibnamefont {G\'orecki}},\ and\ \bibinfo {author} {\bibfnamefont {R.}~\bibnamefont {Demkowicz-Dobrza\ifmmode~\acute{n}\else \'{n}\fi{}ski}},\ }\bibfield  {title} {\bibinfo {title} {{Universal time scalings of sensitivity in Markovian quantum metrology}},\ }\href {https://doi.org/10.1103/PhysRevA.111.L020403} {\bibfield  {journal} {\bibinfo  {journal} {Phys. Rev. A}\ }\textbf {\bibinfo {volume} {111}},\ \bibinfo {pages} {L020403} (\bibinfo {year} {2025})}\BibitemShut {NoStop}%
\bibitem [{\citenamefont {Lu}\ \emph {et~al.}(2010)\citenamefont {Lu}, \citenamefont {Wang},\ and\ \citenamefont {Sun}}]{Lu2010:QFIflow}%
  \BibitemOpen
  \bibfield  {author} {\bibinfo {author} {\bibfnamefont {X.-M.}\ \bibnamefont {Lu}}, \bibinfo {author} {\bibfnamefont {X.}~\bibnamefont {Wang}},\ and\ \bibinfo {author} {\bibfnamefont {C.~P.}\ \bibnamefont {Sun}},\ }\bibfield  {title} {\bibinfo {title} {{Quantum Fisher information flow and non-Markovian processes of open systems}},\ }\href {https://doi.org/10.1103/PhysRevA.82.042103} {\bibfield  {journal} {\bibinfo  {journal} {Phys. Rev. A}\ }\textbf {\bibinfo {volume} {82}},\ \bibinfo {pages} {042103} (\bibinfo {year} {2010})}\BibitemShut {NoStop}%
\bibitem [{\citenamefont {Pires}\ \emph {et~al.}(2016)\citenamefont {Pires}, \citenamefont {Cianciaruso}, \citenamefont {C\'eleri}, \citenamefont {Adesso},\ and\ \citenamefont {Soares-Pinto}}]{Pires2016:QSL}%
  \BibitemOpen
  \bibfield  {author} {\bibinfo {author} {\bibfnamefont {D.~P.}\ \bibnamefont {Pires}}, \bibinfo {author} {\bibfnamefont {M.}~\bibnamefont {Cianciaruso}}, \bibinfo {author} {\bibfnamefont {L.~C.}\ \bibnamefont {C\'eleri}}, \bibinfo {author} {\bibfnamefont {G.}~\bibnamefont {Adesso}},\ and\ \bibinfo {author} {\bibfnamefont {D.~O.}\ \bibnamefont {Soares-Pinto}},\ }\bibfield  {title} {\bibinfo {title} {{Generalized Geometric Quantum Speed Limits}},\ }\href {https://doi.org/10.1103/PhysRevX.6.021031} {\bibfield  {journal} {\bibinfo  {journal} {Phys. Rev. X}\ }\textbf {\bibinfo {volume} {6}},\ \bibinfo {pages} {021031} (\bibinfo {year} {2016})}\BibitemShut {NoStop}%
\bibitem [{\citenamefont {Ito}\ and\ \citenamefont {Dechant}(2020)}]{Ito2020:StochSpeedLimit}%
  \BibitemOpen
  \bibfield  {author} {\bibinfo {author} {\bibfnamefont {S.}~\bibnamefont {Ito}}\ and\ \bibinfo {author} {\bibfnamefont {A.}~\bibnamefont {Dechant}},\ }\bibfield  {title} {\bibinfo {title} {{Stochastic Time Evolution, Information Geometry, and the Cram\'er-Rao Bound}},\ }\href {https://doi.org/10.1103/PhysRevX.10.021056} {\bibfield  {journal} {\bibinfo  {journal} {Phys. Rev. X}\ }\textbf {\bibinfo {volume} {10}},\ \bibinfo {pages} {021056} (\bibinfo {year} {2020})}\BibitemShut {NoStop}%
\bibitem [{\citenamefont {Garc\'{\i}a-Pintos}\ \emph {et~al.}(2022)\citenamefont {Garc\'{\i}a-Pintos}, \citenamefont {Nicholson}, \citenamefont {Green}, \citenamefont {del Campo},\ and\ \citenamefont {Gorshkov}}]{Luispe2022SpeedLimits}%
  \BibitemOpen
  \bibfield  {author} {\bibinfo {author} {\bibfnamefont {L.~P.}\ \bibnamefont {Garc\'{\i}a-Pintos}}, \bibinfo {author} {\bibfnamefont {S.~B.}\ \bibnamefont {Nicholson}}, \bibinfo {author} {\bibfnamefont {J.~R.}\ \bibnamefont {Green}}, \bibinfo {author} {\bibfnamefont {A.}~\bibnamefont {del Campo}},\ and\ \bibinfo {author} {\bibfnamefont {A.~V.}\ \bibnamefont {Gorshkov}},\ }\bibfield  {title} {\bibinfo {title} {{Unifying Quantum and Classical Speed Limits on Observables}},\ }\href {https://doi.org/10.1103/PhysRevX.12.011038} {\bibfield  {journal} {\bibinfo  {journal} {Phys. Rev. X}\ }\textbf {\bibinfo {volume} {12}},\ \bibinfo {pages} {011038} (\bibinfo {year} {2022})}\BibitemShut {NoStop}%
\bibitem [{\citenamefont {Matsumoto}(2002)}]{Matsumoto2002}%
  \BibitemOpen
  \bibfield  {author} {\bibinfo {author} {\bibfnamefont {K.}~\bibnamefont {Matsumoto}},\ }\bibfield  {title} {\bibinfo {title} {{A new approach to the Cram{\'e}r-Rao-type bound of the pure-state model}},\ }\href {https://doi.org/10.1088/0305-4470/35/13/307} {\bibfield  {journal} {\bibinfo  {journal} {J. Phys. A: Math. Gen.}\ }\textbf {\bibinfo {volume} {35}},\ \bibinfo {pages} {3111} (\bibinfo {year} {2002})}\BibitemShut {NoStop}%
\bibitem [{\citenamefont {Radaelli}\ \emph {et~al.}(2026)\citenamefont {Radaelli}, \citenamefont {Smiga}, \citenamefont {Landi},\ and\ \citenamefont {Binder}}]{Radaelli2026:QuJumpUnravel}%
  \BibitemOpen
  \bibfield  {author} {\bibinfo {author} {\bibfnamefont {M.}~\bibnamefont {Radaelli}}, \bibinfo {author} {\bibfnamefont {J.~A.}\ \bibnamefont {Smiga}}, \bibinfo {author} {\bibfnamefont {G.~T.}\ \bibnamefont {Landi}},\ and\ \bibinfo {author} {\bibfnamefont {F.~C.}\ \bibnamefont {Binder}},\ }\bibfield  {title} {\bibinfo {title} {{Parameter estimation for quantum jump unraveling}},\ }\href {https://doi.org/10.22331/q-2026-02-02-1993} {\bibfield  {journal} {\bibinfo  {journal} {Quantum}\ }\textbf {\bibinfo {volume} {10}},\ \bibinfo {pages} {1993} (\bibinfo {year} {2026})}\BibitemShut {NoStop}%
\bibitem [{\citenamefont {Kay}(1993)}]{Kay1993}%
  \BibitemOpen
  \bibfield  {author} {\bibinfo {author} {\bibfnamefont {S.~M.}\ \bibnamefont {Kay}},\ }\href@noop {} {\emph {\bibinfo {title} {{Fundamentals of Statistical Signal Processing: Estimation Theory}}}},\ Vol.~\bibinfo {volume} {1}\ (\bibinfo  {publisher} {Prentice-Hall PTR},\ \bibinfo {year} {1993})\BibitemShut {NoStop}%
\bibitem [{\citenamefont {Proctor}\ \emph {et~al.}(2018)\citenamefont {Proctor}, \citenamefont {Knott},\ and\ \citenamefont {Dunningham}}]{Proctor2018:qsn}%
  \BibitemOpen
  \bibfield  {author} {\bibinfo {author} {\bibfnamefont {T.~J.}\ \bibnamefont {Proctor}}, \bibinfo {author} {\bibfnamefont {P.~A.}\ \bibnamefont {Knott}},\ and\ \bibinfo {author} {\bibfnamefont {J.~A.}\ \bibnamefont {Dunningham}},\ }\bibfield  {title} {\bibinfo {title} {{Multiparameter Estimation in Networked Quantum Sensors}},\ }\href {https://doi.org/10.1103/PhysRevLett.120.080501} {\bibfield  {journal} {\bibinfo  {journal} {Phys. Rev. Lett.}\ }\textbf {\bibinfo {volume} {120}},\ \bibinfo {pages} {080501} (\bibinfo {year} {2018})}\BibitemShut {NoStop}%
\bibitem [{\citenamefont {Zhuang}\ \emph {et~al.}(2018)\citenamefont {Zhuang}, \citenamefont {Zhang},\ and\ \citenamefont {Shapiro}}]{Zhuang2018:dqs}%
  \BibitemOpen
  \bibfield  {author} {\bibinfo {author} {\bibfnamefont {Q.}~\bibnamefont {Zhuang}}, \bibinfo {author} {\bibfnamefont {Z.}~\bibnamefont {Zhang}},\ and\ \bibinfo {author} {\bibfnamefont {J.~H.}\ \bibnamefont {Shapiro}},\ }\bibfield  {title} {\bibinfo {title} {{Distributed quantum sensing using continuous-variable multipartite entanglement}},\ }\href {https://doi.org/10.1103/PhysRevA.97.032329} {\bibfield  {journal} {\bibinfo  {journal} {Phys. Rev. A}\ }\textbf {\bibinfo {volume} {97}},\ \bibinfo {pages} {032329} (\bibinfo {year} {2018})}\BibitemShut {NoStop}%
\bibitem [{\citenamefont {Eldredge}\ \emph {et~al.}(2018)\citenamefont {Eldredge}, \citenamefont {Foss-Feig}, \citenamefont {Gross}, \citenamefont {Rolston},\ and\ \citenamefont {Gorshkov}}]{Eldredge2018:SecureQSNs}%
  \BibitemOpen
  \bibfield  {author} {\bibinfo {author} {\bibfnamefont {Z.}~\bibnamefont {Eldredge}}, \bibinfo {author} {\bibfnamefont {M.}~\bibnamefont {Foss-Feig}}, \bibinfo {author} {\bibfnamefont {J.~A.}\ \bibnamefont {Gross}}, \bibinfo {author} {\bibfnamefont {S.~L.}\ \bibnamefont {Rolston}},\ and\ \bibinfo {author} {\bibfnamefont {A.~V.}\ \bibnamefont {Gorshkov}},\ }\bibfield  {title} {\bibinfo {title} {{Optimal and secure measurement protocols for quantum sensor networks}},\ }\href {https://doi.org/10.1103/PhysRevA.97.042337} {\bibfield  {journal} {\bibinfo  {journal} {Phys. Rev. A}\ }\textbf {\bibinfo {volume} {97}},\ \bibinfo {pages} {042337} (\bibinfo {year} {2018})}\BibitemShut {NoStop}%
\bibitem [{\citenamefont {Ge}\ \emph {et~al.}(2018)\citenamefont {Ge}, \citenamefont {Jacobs}, \citenamefont {Eldredge}, \citenamefont {Gorshkov},\ and\ \citenamefont {Foss-Feig}}]{Ge2018:qsnLinNetworks}%
  \BibitemOpen
  \bibfield  {author} {\bibinfo {author} {\bibfnamefont {W.}~\bibnamefont {Ge}}, \bibinfo {author} {\bibfnamefont {K.}~\bibnamefont {Jacobs}}, \bibinfo {author} {\bibfnamefont {Z.}~\bibnamefont {Eldredge}}, \bibinfo {author} {\bibfnamefont {A.~V.}\ \bibnamefont {Gorshkov}},\ and\ \bibinfo {author} {\bibfnamefont {M.}~\bibnamefont {Foss-Feig}},\ }\bibfield  {title} {\bibinfo {title} {{Distributed Quantum Metrology with Linear Networks and Separable Inputs}},\ }\href {https://doi.org/10.1103/PhysRevLett.121.043604} {\bibfield  {journal} {\bibinfo  {journal} {Phys. Rev. Lett.}\ }\textbf {\bibinfo {volume} {121}},\ \bibinfo {pages} {043604} (\bibinfo {year} {2018})}\BibitemShut {NoStop}%
\bibitem [{\citenamefont {Zhang}\ and\ \citenamefont {Zhuang}(2021)}]{Zhang2021:DQSrvw}%
  \BibitemOpen
  \bibfield  {author} {\bibinfo {author} {\bibfnamefont {Z.}~\bibnamefont {Zhang}}\ and\ \bibinfo {author} {\bibfnamefont {Q.}~\bibnamefont {Zhuang}},\ }\bibfield  {title} {\bibinfo {title} {{Distributed quantum sensing}},\ }\href {https://doi.org/10.1088/2058-9565/abd4c3} {\bibfield  {journal} {\bibinfo  {journal} {Quantum Sci. Technol.}\ }\textbf {\bibinfo {volume} {6}},\ \bibinfo {pages} {043001} (\bibinfo {year} {2021})}\BibitemShut {NoStop}%
\bibitem [{\citenamefont {Monz}\ \emph {et~al.}(2011)\citenamefont {Monz}, \citenamefont {Schindler}, \citenamefont {Barreiro}, \citenamefont {Chwalla}, \citenamefont {Nigg}, \citenamefont {Coish}, \citenamefont {Harlander}, \citenamefont {H\"ansel}, \citenamefont {Hennrich} \emph {et~al.}}]{Monz2011:Superdecoherence}%
  \BibitemOpen
  \bibfield  {author} {\bibinfo {author} {\bibfnamefont {T.}~\bibnamefont {Monz}}, \bibinfo {author} {\bibfnamefont {P.}~\bibnamefont {Schindler}}, \bibinfo {author} {\bibfnamefont {J.~T.}\ \bibnamefont {Barreiro}}, \bibinfo {author} {\bibfnamefont {M.}~\bibnamefont {Chwalla}}, \bibinfo {author} {\bibfnamefont {D.}~\bibnamefont {Nigg}}, \bibinfo {author} {\bibfnamefont {W.~A.}\ \bibnamefont {Coish}}, \bibinfo {author} {\bibfnamefont {M.}~\bibnamefont {Harlander}}, \bibinfo {author} {\bibfnamefont {W.}~\bibnamefont {H\"ansel}}, \bibinfo {author} {\bibfnamefont {M.}~\bibnamefont {Hennrich}}, \emph {et~al.},\ }\bibfield  {title} {\bibinfo {title} {{14-Qubit Entanglement: Creation and Coherence}},\ }\href {https://doi.org/10.1103/PhysRevLett.106.130506} {\bibfield  {journal} {\bibinfo  {journal} {Phys. Rev. Lett.}\ }\textbf {\bibinfo {volume} {106}},\ \bibinfo {pages} {130506} (\bibinfo {year} {2011})}\BibitemShut {NoStop}%
\bibitem [{\citenamefont {Shammah}\ \emph {et~al.}(2018)\citenamefont {Shammah}, \citenamefont {Ahmed}, \citenamefont {Lambert}, \citenamefont {De~Liberato},\ and\ \citenamefont {Nori}}]{Shammah2018:CollectiveDecoherence}%
  \BibitemOpen
  \bibfield  {author} {\bibinfo {author} {\bibfnamefont {N.}~\bibnamefont {Shammah}}, \bibinfo {author} {\bibfnamefont {S.}~\bibnamefont {Ahmed}}, \bibinfo {author} {\bibfnamefont {N.}~\bibnamefont {Lambert}}, \bibinfo {author} {\bibfnamefont {S.}~\bibnamefont {De~Liberato}},\ and\ \bibinfo {author} {\bibfnamefont {F.}~\bibnamefont {Nori}},\ }\bibfield  {title} {\bibinfo {title} {Open quantum systems with local and collective incoherent processes: Efficient numerical simulations using permutational invariance},\ }\href {https://doi.org/10.1103/PhysRevA.98.063815} {\bibfield  {journal} {\bibinfo  {journal} {Phys. Rev. A}\ }\textbf {\bibinfo {volume} {98}},\ \bibinfo {pages} {063815} (\bibinfo {year} {2018})}\BibitemShut {NoStop}%
\bibitem [{\citenamefont {Tsang}\ \emph {et~al.}(2011)\citenamefont {Tsang}, \citenamefont {Wiseman},\ and\ \citenamefont {Caves}}]{Tsang2011:Waveform}%
  \BibitemOpen
  \bibfield  {author} {\bibinfo {author} {\bibfnamefont {M.}~\bibnamefont {Tsang}}, \bibinfo {author} {\bibfnamefont {H.~M.}\ \bibnamefont {Wiseman}},\ and\ \bibinfo {author} {\bibfnamefont {C.~M.}\ \bibnamefont {Caves}},\ }\bibfield  {title} {\bibinfo {title} {{Fundamental Quantum Limit to Waveform Estimation}},\ }\href {https://doi.org/10.1103/PhysRevLett.106.090401} {\bibfield  {journal} {\bibinfo  {journal} {Phys. Rev. Lett.}\ }\textbf {\bibinfo {volume} {106}},\ \bibinfo {pages} {090401} (\bibinfo {year} {2011})}\BibitemShut {NoStop}%
\bibitem [{\citenamefont {Ng}\ \emph {et~al.}(2016)\citenamefont {Ng}, \citenamefont {Ang}, \citenamefont {Wheatley}, \citenamefont {Yonezawa}, \citenamefont {Furusawa}, \citenamefont {Huntington},\ and\ \citenamefont {Tsang}}]{Ng2016:WaveformSpectrum}%
  \BibitemOpen
  \bibfield  {author} {\bibinfo {author} {\bibfnamefont {S.}~\bibnamefont {Ng}}, \bibinfo {author} {\bibfnamefont {S.~Z.}\ \bibnamefont {Ang}}, \bibinfo {author} {\bibfnamefont {T.~A.}\ \bibnamefont {Wheatley}}, \bibinfo {author} {\bibfnamefont {H.}~\bibnamefont {Yonezawa}}, \bibinfo {author} {\bibfnamefont {A.}~\bibnamefont {Furusawa}}, \bibinfo {author} {\bibfnamefont {E.~H.}\ \bibnamefont {Huntington}},\ and\ \bibinfo {author} {\bibfnamefont {M.}~\bibnamefont {Tsang}},\ }\bibfield  {title} {\bibinfo {title} {{Spectrum analysis with quantum dynamical systems}},\ }\href {https://doi.org/10.1103/PhysRevA.93.042121} {\bibfield  {journal} {\bibinfo  {journal} {Phys. Rev. A}\ }\textbf {\bibinfo {volume} {93}},\ \bibinfo {pages} {042121} (\bibinfo {year} {2016})}\BibitemShut {NoStop}%
\bibitem [{\citenamefont {Gardner}\ \emph {et~al.}(2025{\natexlab{b}})\citenamefont {Gardner}, \citenamefont {Gefen}, \citenamefont {Haine}, \citenamefont {Hope}, \citenamefont {Preskill}, \citenamefont {Chen},\ and\ \citenamefont {McCuller}}]{Gardner2025:Waveform}%
  \BibitemOpen
  \bibfield  {author} {\bibinfo {author} {\bibfnamefont {J.~W.}\ \bibnamefont {Gardner}}, \bibinfo {author} {\bibfnamefont {T.}~\bibnamefont {Gefen}}, \bibinfo {author} {\bibfnamefont {S.~A.}\ \bibnamefont {Haine}}, \bibinfo {author} {\bibfnamefont {J.~J.}\ \bibnamefont {Hope}}, \bibinfo {author} {\bibfnamefont {J.}~\bibnamefont {Preskill}}, \bibinfo {author} {\bibfnamefont {Y.}~\bibnamefont {Chen}},\ and\ \bibinfo {author} {\bibfnamefont {L.}~\bibnamefont {McCuller}},\ }\bibfield  {title} {\bibinfo {title} {{Stochastic Waveform Estimation at the Fundamental Quantum Limit}},\ }\href {https://doi.org/10.1103/h91r-4ws9} {\bibfield  {journal} {\bibinfo  {journal} {PRX Quantum}\ }\textbf {\bibinfo {volume} {6}},\ \bibinfo {pages} {030311} (\bibinfo {year} {2025}{\natexlab{b}})}\BibitemShut {NoStop}%
\bibitem [{\citenamefont {Brady}\ \emph {et~al.}(2022)\citenamefont {Brady}, \citenamefont {Gao}, \citenamefont {Harnik}, \citenamefont {Liu}, \citenamefont {Zhang},\ and\ \citenamefont {Zhuang}}]{Brady2022:cavityDMSearch}%
  \BibitemOpen
  \bibfield  {author} {\bibinfo {author} {\bibfnamefont {A.~J.}\ \bibnamefont {Brady}}, \bibinfo {author} {\bibfnamefont {C.}~\bibnamefont {Gao}}, \bibinfo {author} {\bibfnamefont {R.}~\bibnamefont {Harnik}}, \bibinfo {author} {\bibfnamefont {Z.}~\bibnamefont {Liu}}, \bibinfo {author} {\bibfnamefont {Z.}~\bibnamefont {Zhang}},\ and\ \bibinfo {author} {\bibfnamefont {Q.}~\bibnamefont {Zhuang}},\ }\bibfield  {title} {\bibinfo {title} {{Entangled Sensor-Networks for Dark-Matter Searches}},\ }\href {https://doi.org/10.1103/PRXQuantum.3.030333} {\bibfield  {journal} {\bibinfo  {journal} {PRX Quantum}\ }\textbf {\bibinfo {volume} {3}},\ \bibinfo {pages} {030333} (\bibinfo {year} {2022})}\BibitemShut {NoStop}%
\bibitem [{\citenamefont {Brady}\ \emph {et~al.}(2023)\citenamefont {Brady}, \citenamefont {Chen}, \citenamefont {Xia}, \citenamefont {Manley}, \citenamefont {Dey~Chowdhury}, \citenamefont {Xiao}, \citenamefont {Liu}, \citenamefont {Harnik}, \citenamefont {Wilson} \emph {et~al.}}]{Brady2023:OmechArrayDMsearch}%
  \BibitemOpen
  \bibfield  {author} {\bibinfo {author} {\bibfnamefont {A.~J.}\ \bibnamefont {Brady}}, \bibinfo {author} {\bibfnamefont {X.}~\bibnamefont {Chen}}, \bibinfo {author} {\bibfnamefont {Y.}~\bibnamefont {Xia}}, \bibinfo {author} {\bibfnamefont {J.}~\bibnamefont {Manley}}, \bibinfo {author} {\bibfnamefont {M.}~\bibnamefont {Dey~Chowdhury}}, \bibinfo {author} {\bibfnamefont {K.}~\bibnamefont {Xiao}}, \bibinfo {author} {\bibfnamefont {Z.}~\bibnamefont {Liu}}, \bibinfo {author} {\bibfnamefont {R.}~\bibnamefont {Harnik}}, \bibinfo {author} {\bibfnamefont {D.~J.}\ \bibnamefont {Wilson}}, \emph {et~al.},\ }\bibfield  {title} {\bibinfo {title} {{Entanglement-enhanced optomechanical sensor array with application to dark matter searches}},\ }\href {https://doi.org/10.1038/s42005-023-01357-z} {\bibfield  {journal} {\bibinfo  {journal} {Commun. Phys.}\ }\textbf {\bibinfo {volume} {6}},\ \bibinfo {pages} {237} (\bibinfo {year} {2023})}\BibitemShut {NoStop}%
\bibitem [{\citenamefont {Xia}\ \emph {et~al.}(2023)\citenamefont {Xia}, \citenamefont {Agrawal}, \citenamefont {Pluchar}, \citenamefont {Brady}, \citenamefont {Liu}, \citenamefont {Zhuang}, \citenamefont {Wilson},\ and\ \citenamefont {Zhang}}]{Xia2023:OmechDQS}%
  \BibitemOpen
  \bibfield  {author} {\bibinfo {author} {\bibfnamefont {Y.}~\bibnamefont {Xia}}, \bibinfo {author} {\bibfnamefont {A.~R.}\ \bibnamefont {Agrawal}}, \bibinfo {author} {\bibfnamefont {C.~M.}\ \bibnamefont {Pluchar}}, \bibinfo {author} {\bibfnamefont {A.~J.}\ \bibnamefont {Brady}}, \bibinfo {author} {\bibfnamefont {Z.}~\bibnamefont {Liu}}, \bibinfo {author} {\bibfnamefont {Q.}~\bibnamefont {Zhuang}}, \bibinfo {author} {\bibfnamefont {D.~J.}\ \bibnamefont {Wilson}},\ and\ \bibinfo {author} {\bibfnamefont {Z.}~\bibnamefont {Zhang}},\ }\bibfield  {title} {\bibinfo {title} {{Entanglement-enhanced optomechanical sensing}},\ }\href {https://doi.org/10.1038/s41566-023-01178-0} {\bibfield  {journal} {\bibinfo  {journal} {Nat. Photon.}\ }\textbf {\bibinfo {volume} {17}},\ \bibinfo {pages} {470} (\bibinfo {year} {2023})}\BibitemShut {NoStop}%
\bibitem [{\citenamefont {Guo}\ \emph {et~al.}(2020)\citenamefont {Guo}, \citenamefont {Breum}, \citenamefont {Borregaard}, \citenamefont {Izumi}, \citenamefont {Larsen}, \citenamefont {Gehring}, \citenamefont {Christandl}, \citenamefont {Neergaard-Nielsen},\ and\ \citenamefont {Andersen}}]{Guo2020:dqsCVnetwork}%
  \BibitemOpen
  \bibfield  {author} {\bibinfo {author} {\bibfnamefont {X.}~\bibnamefont {Guo}}, \bibinfo {author} {\bibfnamefont {C.~R.}\ \bibnamefont {Breum}}, \bibinfo {author} {\bibfnamefont {J.}~\bibnamefont {Borregaard}}, \bibinfo {author} {\bibfnamefont {S.}~\bibnamefont {Izumi}}, \bibinfo {author} {\bibfnamefont {M.~V.}\ \bibnamefont {Larsen}}, \bibinfo {author} {\bibfnamefont {T.}~\bibnamefont {Gehring}}, \bibinfo {author} {\bibfnamefont {M.}~\bibnamefont {Christandl}}, \bibinfo {author} {\bibfnamefont {J.~S.}\ \bibnamefont {Neergaard-Nielsen}},\ and\ \bibinfo {author} {\bibfnamefont {U.~L.}\ \bibnamefont {Andersen}},\ }\bibfield  {title} {\bibinfo {title} {{Distributed quantum sensing in a continuous-variable entangled network}},\ }\href {https://doi.org/10.1038/s41567-019-0743-x} {\bibfield  {journal} {\bibinfo  {journal} {Nat. Phys.}\ }\textbf {\bibinfo {volume} {16}},\ \bibinfo {pages} {281} (\bibinfo {year} {2020})}\BibitemShut {NoStop}%
\bibitem [{\citenamefont {Gilmore}\ \emph {et~al.}(2021)\citenamefont {Gilmore}, \citenamefont {Affolter}, \citenamefont {Lewis-Swan}, \citenamefont {Barberena}, \citenamefont {Jordan}, \citenamefont {Rey},\ and\ \citenamefont {Bollinger}}]{Gilmore2021:IonEFieldQSN}%
  \BibitemOpen
  \bibfield  {author} {\bibinfo {author} {\bibfnamefont {K.~A.}\ \bibnamefont {Gilmore}}, \bibinfo {author} {\bibfnamefont {M.}~\bibnamefont {Affolter}}, \bibinfo {author} {\bibfnamefont {R.~J.}\ \bibnamefont {Lewis-Swan}}, \bibinfo {author} {\bibfnamefont {D.}~\bibnamefont {Barberena}}, \bibinfo {author} {\bibfnamefont {E.}~\bibnamefont {Jordan}}, \bibinfo {author} {\bibfnamefont {A.~M.}\ \bibnamefont {Rey}},\ and\ \bibinfo {author} {\bibfnamefont {J.~J.}\ \bibnamefont {Bollinger}},\ }\bibfield  {title} {\bibinfo {title} {{Quantum-enhanced sensing of displacements and electric fields with two-dimensional trapped-ion crystals}},\ }\href {https://doi.org/10.1126/science.abi5226} {\bibfield  {journal} {\bibinfo  {journal} {Science}\ }\textbf {\bibinfo {volume} {373}},\ \bibinfo {pages} {673} (\bibinfo {year} {2021})}\BibitemShut {NoStop}%
\bibitem [{\citenamefont {Vasilyev}\ \emph {et~al.}(2024)\citenamefont {Vasilyev}, \citenamefont {Shankar}, \citenamefont {Kaubruegger},\ and\ \citenamefont {Zoller}}]{Vasilyev2024:QuCompass}%
  \BibitemOpen
  \bibfield  {author} {\bibinfo {author} {\bibfnamefont {D.~V.}\ \bibnamefont {Vasilyev}}, \bibinfo {author} {\bibfnamefont {A.}~\bibnamefont {Shankar}}, \bibinfo {author} {\bibfnamefont {R.}~\bibnamefont {Kaubruegger}},\ and\ \bibinfo {author} {\bibfnamefont {P.}~\bibnamefont {Zoller}},\ }\href@noop {} {\bibinfo {title} {{Optimal Multiparameter Metrology: The Quantum Compass Solution}}} (\bibinfo {year} {2024}),\ \Eprint {https://arxiv.org/abs/2404.14194} {arXiv:2404.14194 [quant-ph]} \BibitemShut {NoStop}%
\bibitem [{\citenamefont {Shankar}(2012)}]{Shankar2012QM}%
  \BibitemOpen
  \bibfield  {author} {\bibinfo {author} {\bibfnamefont {R.}~\bibnamefont {Shankar}},\ }\href {https://doi.org/10.1007/978-1-4757-0576-8} {\emph {\bibinfo {title} {Principles of Quantum Mechanics}}},\ \bibinfo {edition} {2nd}\ ed.\ (\bibinfo  {publisher} {Springer},\ \bibinfo {year} {2012})\BibitemShut {NoStop}%
\bibitem [{\citenamefont {Omanakuttan}\ \emph {et~al.}(2024)\citenamefont {Omanakuttan}, \citenamefont {Buchemmavari}, \citenamefont {Gross}, \citenamefont {Deutsch},\ and\ \citenamefont {Marvian}}]{Omanakuttan2024:SpinCat}%
  \BibitemOpen
  \bibfield  {author} {\bibinfo {author} {\bibfnamefont {S.}~\bibnamefont {Omanakuttan}}, \bibinfo {author} {\bibfnamefont {V.}~\bibnamefont {Buchemmavari}}, \bibinfo {author} {\bibfnamefont {J.~A.}\ \bibnamefont {Gross}}, \bibinfo {author} {\bibfnamefont {I.~H.}\ \bibnamefont {Deutsch}},\ and\ \bibinfo {author} {\bibfnamefont {M.}~\bibnamefont {Marvian}},\ }\bibfield  {title} {\bibinfo {title} {Fault-tolerant quantum computation using large spin-cat codes},\ }\href {https://doi.org/10.1103/PRXQuantum.5.020355} {\bibfield  {journal} {\bibinfo  {journal} {PRX Quantum}\ }\textbf {\bibinfo {volume} {5}},\ \bibinfo {pages} {020355} (\bibinfo {year} {2024})}\BibitemShut {NoStop}%
\bibitem [{\citenamefont {T\'oth}(2012)}]{Toth2012:dicke}%
  \BibitemOpen
  \bibfield  {author} {\bibinfo {author} {\bibfnamefont {G.}~\bibnamefont {T\'oth}},\ }\bibfield  {title} {\bibinfo {title} {{Multipartite entanglement and high-precision metrology}},\ }\href {https://doi.org/10.1103/PhysRevA.85.022322} {\bibfield  {journal} {\bibinfo  {journal} {Phys. Rev. A}\ }\textbf {\bibinfo {volume} {85}},\ \bibinfo {pages} {022322} (\bibinfo {year} {2012})}\BibitemShut {NoStop}%
\bibitem [{\citenamefont {Fujiwara}\ and\ \citenamefont {Imai}(2003)}]{Fujiwara2003:PauliNoise}%
  \BibitemOpen
  \bibfield  {author} {\bibinfo {author} {\bibfnamefont {A.}~\bibnamefont {Fujiwara}}\ and\ \bibinfo {author} {\bibfnamefont {H.}~\bibnamefont {Imai}},\ }\bibfield  {title} {\bibinfo {title} {{Quantum parameter estimation of a generalized Pauli channel}},\ }\href {https://doi.org/10.1088/0305-4470/36/29/314} {\bibfield  {journal} {\bibinfo  {journal} {J. Phys. A: Math. Gen.}\ }\textbf {\bibinfo {volume} {36}},\ \bibinfo {pages} {8093} (\bibinfo {year} {2003})}\BibitemShut {NoStop}%
\bibitem [{\citenamefont {Flammia}\ and\ \citenamefont {Wallman}(2020)}]{Flammia2020:PauliLearn}%
  \BibitemOpen
  \bibfield  {author} {\bibinfo {author} {\bibfnamefont {S.~T.}\ \bibnamefont {Flammia}}\ and\ \bibinfo {author} {\bibfnamefont {J.~J.}\ \bibnamefont {Wallman}},\ }\bibfield  {title} {\bibinfo {title} {{Efficient Estimation of Pauli Channels}},\ }\href {https://doi.org/10.1145/3408039} {\bibfield  {journal} {\bibinfo  {journal} {ACM Trans. Quantum Comput.}\ }\textbf {\bibinfo {volume} {1}},\ \bibinfo {pages} {1} (\bibinfo {year} {2020})}\BibitemShut {NoStop}%
\bibitem [{\citenamefont {Kwon}\ \emph {et~al.}(2026)\citenamefont {Kwon}, \citenamefont {Lie},\ and\ \citenamefont {Jiang}}]{Kwon2026:FisherComplexity}%
  \BibitemOpen
  \bibfield  {author} {\bibinfo {author} {\bibfnamefont {H.}~\bibnamefont {Kwon}}, \bibinfo {author} {\bibfnamefont {S.~H.}\ \bibnamefont {Lie}},\ and\ \bibinfo {author} {\bibfnamefont {L.}~\bibnamefont {Jiang}},\ }\href@noop {} {\bibinfo {title} {{Universal Sample Complexity Bounds in Quantum Learning Theory via Fisher Information matrix}}} (\bibinfo {year} {2026}),\ \Eprint {https://arxiv.org/abs/2602.21510} {arXiv:2602.21510 [quant-ph]} \BibitemShut {NoStop}%
\bibitem [{\citenamefont {Tu}\ and\ \citenamefont {Jiang}(2025)}]{Tu2025:LearnMixedU}%
  \BibitemOpen
  \bibfield  {author} {\bibinfo {author} {\bibfnamefont {Y.}~\bibnamefont {Tu}}\ and\ \bibinfo {author} {\bibfnamefont {L.}~\bibnamefont {Jiang}},\ }\href@noop {} {\bibinfo {title} {{Quantum Advantage in Learning Mixed Unitary Channels}}} (\bibinfo {year} {2025}),\ \Eprint {https://arxiv.org/abs/2511.13683} {arXiv:2511.13683 [quant-ph]} \BibitemShut {NoStop}%
\bibitem [{\citenamefont {Tsang}\ \emph {et~al.}(2016)\citenamefont {Tsang}, \citenamefont {Nair},\ and\ \citenamefont {Lu}}]{Tsang2016:SPADE}%
  \BibitemOpen
  \bibfield  {author} {\bibinfo {author} {\bibfnamefont {M.}~\bibnamefont {Tsang}}, \bibinfo {author} {\bibfnamefont {R.}~\bibnamefont {Nair}},\ and\ \bibinfo {author} {\bibfnamefont {X.-M.}\ \bibnamefont {Lu}},\ }\bibfield  {title} {\bibinfo {title} {{Quantum Theory of Superresolution for Two Incoherent Optical Point Sources}},\ }\href {https://doi.org/10.1103/PhysRevX.6.031033} {\bibfield  {journal} {\bibinfo  {journal} {Phys. Rev. X}\ }\textbf {\bibinfo {volume} {6}},\ \bibinfo {pages} {031033} (\bibinfo {year} {2016})}\BibitemShut {NoStop}%
\bibitem [{\citenamefont {Lupo}\ and\ \citenamefont {Pirandola}(2016)}]{Lupo2016:QuImaging}%
  \BibitemOpen
  \bibfield  {author} {\bibinfo {author} {\bibfnamefont {C.}~\bibnamefont {Lupo}}\ and\ \bibinfo {author} {\bibfnamefont {S.}~\bibnamefont {Pirandola}},\ }\bibfield  {title} {\bibinfo {title} {{Ultimate Precision Bound of Quantum and Subwavelength Imaging}},\ }\href {https://doi.org/10.1103/PhysRevLett.117.190802} {\bibfield  {journal} {\bibinfo  {journal} {Phys. Rev. Lett.}\ }\textbf {\bibinfo {volume} {117}},\ \bibinfo {pages} {190802} (\bibinfo {year} {2016})}\BibitemShut {NoStop}%
\bibitem [{\citenamefont {Lupo}\ \emph {et~al.}(2020)\citenamefont {Lupo}, \citenamefont {Huang},\ and\ \citenamefont {Kok}}]{Lupo2020:LinearOptLimits}%
  \BibitemOpen
  \bibfield  {author} {\bibinfo {author} {\bibfnamefont {C.}~\bibnamefont {Lupo}}, \bibinfo {author} {\bibfnamefont {Z.}~\bibnamefont {Huang}},\ and\ \bibinfo {author} {\bibfnamefont {P.}~\bibnamefont {Kok}},\ }\bibfield  {title} {\bibinfo {title} {{Quantum Limits to Incoherent Imaging are Achieved by Linear Interferometry}},\ }\href {https://doi.org/10.1103/PhysRevLett.124.080503} {\bibfield  {journal} {\bibinfo  {journal} {Phys. Rev. Lett.}\ }\textbf {\bibinfo {volume} {124}},\ \bibinfo {pages} {080503} (\bibinfo {year} {2020})}\BibitemShut {NoStop}%
\bibitem [{\citenamefont {\ifmmode \check{R}\else \v{R}\fi{}eha\ifmmode~\check{c}\else \v{c}\fi{}ek}\ \emph {et~al.}(2017)\citenamefont {\ifmmode \check{R}\else \v{R}\fi{}eha\ifmmode~\check{c}\else \v{c}\fi{}ek}, \citenamefont {Hradil}, \citenamefont {Stoklasa}, \citenamefont {Pa\'ur}, \citenamefont {Grover}, \citenamefont {Krzic},\ and\ \citenamefont {S\'anchez-Soto}}]{Rehacek2017:MultiImaging}%
  \BibitemOpen
  \bibfield  {author} {\bibinfo {author} {\bibfnamefont {J.}~\bibnamefont {\ifmmode \check{R}\else \v{R}\fi{}eha\ifmmode~\check{c}\else \v{c}\fi{}ek}}, \bibinfo {author} {\bibfnamefont {Z.}~\bibnamefont {Hradil}}, \bibinfo {author} {\bibfnamefont {B.}~\bibnamefont {Stoklasa}}, \bibinfo {author} {\bibfnamefont {M.}~\bibnamefont {Pa\'ur}}, \bibinfo {author} {\bibfnamefont {J.}~\bibnamefont {Grover}}, \bibinfo {author} {\bibfnamefont {A.}~\bibnamefont {Krzic}},\ and\ \bibinfo {author} {\bibfnamefont {L.~L.}\ \bibnamefont {S\'anchez-Soto}},\ }\bibfield  {title} {\bibinfo {title} {{Multiparameter quantum metrology of incoherent point sources: Towards realistic superresolution}},\ }\href {https://doi.org/10.1103/PhysRevA.96.062107} {\bibfield  {journal} {\bibinfo  {journal} {Phys. Rev. A}\ }\textbf {\bibinfo {volume} {96}},\ \bibinfo {pages} {062107} (\bibinfo {year} {2017})}\BibitemShut {NoStop}%
\bibitem [{\citenamefont {Grace}\ \emph {et~al.}(2020)\citenamefont {Grace}, \citenamefont {Dutton}, \citenamefont {Ashok},\ and\ \citenamefont {Guha}}]{Grace2020:Imaging}%
  \BibitemOpen
  \bibfield  {author} {\bibinfo {author} {\bibfnamefont {M.~R.}\ \bibnamefont {Grace}}, \bibinfo {author} {\bibfnamefont {Z.}~\bibnamefont {Dutton}}, \bibinfo {author} {\bibfnamefont {A.}~\bibnamefont {Ashok}},\ and\ \bibinfo {author} {\bibfnamefont {S.}~\bibnamefont {Guha}},\ }\bibfield  {title} {\bibinfo {title} {Approaching quantum-limited imaging resolution without prior knowledge of the object location},\ }\href {https://doi.org/10.1364/josaa.392116} {\bibfield  {journal} {\bibinfo  {journal} {JOSA A}\ }\textbf {\bibinfo {volume} {37}},\ \bibinfo {pages} {1288} (\bibinfo {year} {2020})}\BibitemShut {NoStop}%
\bibitem [{\citenamefont {Brady}\ \emph {et~al.}(2025)\citenamefont {Brady}, \citenamefont {Gong}, \citenamefont {Gorshkov},\ and\ \citenamefont {Guha}}]{Brady2025:TwinEcho}%
  \BibitemOpen
  \bibfield  {author} {\bibinfo {author} {\bibfnamefont {A.~J.}\ \bibnamefont {Brady}}, \bibinfo {author} {\bibfnamefont {Z.}~\bibnamefont {Gong}}, \bibinfo {author} {\bibfnamefont {A.~V.}\ \bibnamefont {Gorshkov}},\ and\ \bibinfo {author} {\bibfnamefont {S.}~\bibnamefont {Guha}},\ }\href@noop {} {\bibinfo {title} {{Incoherent Imaging with Spatially Structured Quantum Probes}}} (\bibinfo {year} {2025}),\ \Eprint {https://arxiv.org/abs/2510.09521} {arXiv:2510.09521 [quant-ph]} \BibitemShut {NoStop}%
\bibitem [{\citenamefont {Samach}\ \emph {et~al.}(2022)\citenamefont {Samach}, \citenamefont {Greene}, \citenamefont {Borregaard}, \citenamefont {Christandl}, \citenamefont {Barreto}, \citenamefont {Kim}, \citenamefont {McNally}, \citenamefont {Melville}, \citenamefont {Niedzielski} \emph {et~al.}}]{Samach2022:LindbladTomography}%
  \BibitemOpen
  \bibfield  {author} {\bibinfo {author} {\bibfnamefont {G.~O.}\ \bibnamefont {Samach}}, \bibinfo {author} {\bibfnamefont {A.}~\bibnamefont {Greene}}, \bibinfo {author} {\bibfnamefont {J.}~\bibnamefont {Borregaard}}, \bibinfo {author} {\bibfnamefont {M.}~\bibnamefont {Christandl}}, \bibinfo {author} {\bibfnamefont {J.}~\bibnamefont {Barreto}}, \bibinfo {author} {\bibfnamefont {D.~K.}\ \bibnamefont {Kim}}, \bibinfo {author} {\bibfnamefont {C.~M.}\ \bibnamefont {McNally}}, \bibinfo {author} {\bibfnamefont {A.}~\bibnamefont {Melville}}, \bibinfo {author} {\bibfnamefont {B.~M.}\ \bibnamefont {Niedzielski}}, \emph {et~al.},\ }\bibfield  {title} {\bibinfo {title} {{Lindblad Tomography of a Superconducting Quantum Processor}},\ }\href {https://doi.org/10.1103/PhysRevApplied.18.064056} {\bibfield  {journal} {\bibinfo  {journal} {Phys. Rev. Appl.}\ }\textbf {\bibinfo {volume} {18}},\ \bibinfo {pages} {064056} (\bibinfo {year} {2022})}\BibitemShut {NoStop}%
\bibitem [{\citenamefont {Olsacher}\ \emph {et~al.}(2025)\citenamefont {Olsacher}, \citenamefont {Kraft}, \citenamefont {Kokail}, \citenamefont {Kraus},\ and\ \citenamefont {Zoller}}]{Olsacher2025:LindbladLearning}%
  \BibitemOpen
  \bibfield  {author} {\bibinfo {author} {\bibfnamefont {T.}~\bibnamefont {Olsacher}}, \bibinfo {author} {\bibfnamefont {T.}~\bibnamefont {Kraft}}, \bibinfo {author} {\bibfnamefont {C.}~\bibnamefont {Kokail}}, \bibinfo {author} {\bibfnamefont {B.}~\bibnamefont {Kraus}},\ and\ \bibinfo {author} {\bibfnamefont {P.}~\bibnamefont {Zoller}},\ }\bibfield  {title} {\bibinfo {title} {{Hamiltonian and Liouvillian learning in weakly-dissipative quantum many-body systems}},\ }\href {https://doi.org/10.1088/2058-9565/ad9ed5} {\bibfield  {journal} {\bibinfo  {journal} {Quantum Sci. Technol.}\ }\textbf {\bibinfo {volume} {10}},\ \bibinfo {pages} {015065} (\bibinfo {year} {2025})}\BibitemShut {NoStop}%
\bibitem [{\citenamefont {Ivashkov}\ \emph {et~al.}(2026)\citenamefont {Ivashkov}, \citenamefont {Romanov}, \citenamefont {Gong}, \citenamefont {Gu}, \citenamefont {Hu},\ and\ \citenamefont {Yelin}}]{Ivashkov2026:LindbladLearn}%
  \BibitemOpen
  \bibfield  {author} {\bibinfo {author} {\bibfnamefont {P.}~\bibnamefont {Ivashkov}}, \bibinfo {author} {\bibfnamefont {N.}~\bibnamefont {Romanov}}, \bibinfo {author} {\bibfnamefont {W.}~\bibnamefont {Gong}}, \bibinfo {author} {\bibfnamefont {A.}~\bibnamefont {Gu}}, \bibinfo {author} {\bibfnamefont {H.-Y.}\ \bibnamefont {Hu}},\ and\ \bibinfo {author} {\bibfnamefont {S.~F.}\ \bibnamefont {Yelin}},\ }\href@noop {} {\bibinfo {title} {{Ansatz-Free Learning of Lindbladian Dynamics In Situ}}} (\bibinfo {year} {2026}),\ \Eprint {https://arxiv.org/abs/2603.05492} {arXiv:2603.05492 [quant-ph]} \BibitemShut {NoStop}%
\bibitem [{\citenamefont {Breuer}\ \emph {et~al.}(2016)\citenamefont {Breuer}, \citenamefont {Laine}, \citenamefont {Piilo},\ and\ \citenamefont {Vacchini}}]{Breuer2016:NonMarkovRvw}%
  \BibitemOpen
  \bibfield  {author} {\bibinfo {author} {\bibfnamefont {H.-P.}\ \bibnamefont {Breuer}}, \bibinfo {author} {\bibfnamefont {E.-M.}\ \bibnamefont {Laine}}, \bibinfo {author} {\bibfnamefont {J.}~\bibnamefont {Piilo}},\ and\ \bibinfo {author} {\bibfnamefont {B.}~\bibnamefont {Vacchini}},\ }\bibfield  {title} {\bibinfo {title} {{Colloquium: Non-Markovian dynamics in open quantum systems}},\ }\href {https://doi.org/10.1103/RevModPhys.88.021002} {\bibfield  {journal} {\bibinfo  {journal} {Rev. Mod. Phys.}\ }\textbf {\bibinfo {volume} {88}},\ \bibinfo {pages} {021002} (\bibinfo {year} {2016})}\BibitemShut {NoStop}%
\bibitem [{\citenamefont {Groszkowski}\ \emph {et~al.}(2023)\citenamefont {Groszkowski}, \citenamefont {Seif}, \citenamefont {Koch},\ and\ \citenamefont {Clerk}}]{Groszkowski2023:nonMarkovModel}%
  \BibitemOpen
  \bibfield  {author} {\bibinfo {author} {\bibfnamefont {P.}~\bibnamefont {Groszkowski}}, \bibinfo {author} {\bibfnamefont {A.}~\bibnamefont {Seif}}, \bibinfo {author} {\bibfnamefont {J.}~\bibnamefont {Koch}},\ and\ \bibinfo {author} {\bibfnamefont {A.~A.}\ \bibnamefont {Clerk}},\ }\bibfield  {title} {\bibinfo {title} {{Simple master equations for describing driven systems subject to classical non-Markovian noise}},\ }\href {https://doi.org/10.22331/q-2023-04-06-972} {\bibfield  {journal} {\bibinfo  {journal} {Quantum}\ }\textbf {\bibinfo {volume} {7}},\ \bibinfo {pages} {972} (\bibinfo {year} {2023})}\BibitemShut {NoStop}%
\bibitem [{\citenamefont {Chin}\ \emph {et~al.}(2012)\citenamefont {Chin}, \citenamefont {Huelga},\ and\ \citenamefont {Plenio}}]{Chin2012:QuMetrologyNonMarkov}%
  \BibitemOpen
  \bibfield  {author} {\bibinfo {author} {\bibfnamefont {A.~W.}\ \bibnamefont {Chin}}, \bibinfo {author} {\bibfnamefont {S.~F.}\ \bibnamefont {Huelga}},\ and\ \bibinfo {author} {\bibfnamefont {M.~B.}\ \bibnamefont {Plenio}},\ }\bibfield  {title} {\bibinfo {title} {{Quantum Metrology in Non-Markovian Environments}},\ }\href {https://doi.org/10.1103/PhysRevLett.109.233601} {\bibfield  {journal} {\bibinfo  {journal} {Phys. Rev. Lett.}\ }\textbf {\bibinfo {volume} {109}},\ \bibinfo {pages} {233601} (\bibinfo {year} {2012})}\BibitemShut {NoStop}%
\bibitem [{\citenamefont {White}\ \emph {et~al.}(2022)\citenamefont {White}, \citenamefont {Pollock}, \citenamefont {Hollenberg}, \citenamefont {Modi},\ and\ \citenamefont {Hill}}]{White2022:nonMarkovProcessTom}%
  \BibitemOpen
  \bibfield  {author} {\bibinfo {author} {\bibfnamefont {G.}~\bibnamefont {White}}, \bibinfo {author} {\bibfnamefont {F.}~\bibnamefont {Pollock}}, \bibinfo {author} {\bibfnamefont {L.}~\bibnamefont {Hollenberg}}, \bibinfo {author} {\bibfnamefont {K.}~\bibnamefont {Modi}},\ and\ \bibinfo {author} {\bibfnamefont {C.}~\bibnamefont {Hill}},\ }\bibfield  {title} {\bibinfo {title} {{Non-Markovian Quantum Process Tomography}},\ }\href {https://doi.org/10.1103/PRXQuantum.3.020344} {\bibfield  {journal} {\bibinfo  {journal} {PRX Quantum}\ }\textbf {\bibinfo {volume} {3}},\ \bibinfo {pages} {020344} (\bibinfo {year} {2022})}\BibitemShut {NoStop}%
\bibitem [{\citenamefont {Varona}\ \emph {et~al.}(2025)\citenamefont {Varona}, \citenamefont {M{\"u}ller},\ and\ \citenamefont {Bermudez}}]{Varona2025:nonMarkovLearn}%
  \BibitemOpen
  \bibfield  {author} {\bibinfo {author} {\bibfnamefont {S.}~\bibnamefont {Varona}}, \bibinfo {author} {\bibfnamefont {M.}~\bibnamefont {M{\"u}ller}},\ and\ \bibinfo {author} {\bibfnamefont {A.}~\bibnamefont {Bermudez}},\ }\bibfield  {title} {\bibinfo {title} {{Lindblad-like quantum tomography for non-Markovian quantum dynamical maps}},\ }\href {https://doi.org/10.1038/s41534-025-01044-7} {\bibfield  {journal} {\bibinfo  {journal} {npj Quantum Inf.}\ }\textbf {\bibinfo {volume} {11}},\ \bibinfo {pages} {96} (\bibinfo {year} {2025})}\BibitemShut {NoStop}%
\bibitem [{\citenamefont {Montañà-López}\ \emph {et~al.}(2025)\citenamefont {Montañà-López}, \citenamefont {Elben}, \citenamefont {Choi},\ and\ \citenamefont {Trivedi}}]{Jordi2025:NonMarkovLearn}%
  \BibitemOpen
  \bibfield  {author} {\bibinfo {author} {\bibfnamefont {J.~A.}\ \bibnamefont {Montañà-López}}, \bibinfo {author} {\bibfnamefont {A.}~\bibnamefont {Elben}}, \bibinfo {author} {\bibfnamefont {J.}~\bibnamefont {Choi}},\ and\ \bibinfo {author} {\bibfnamefont {R.}~\bibnamefont {Trivedi}},\ }\href@noop {} {\bibinfo {title} {{Efficiently learning non-Markovian noise in many-body quantum simulators}}} (\bibinfo {year} {2025}),\ \Eprint {https://arxiv.org/abs/2511.16772} {arXiv:2511.16772 [quant-ph]} \BibitemShut {NoStop}%
\bibitem [{\citenamefont {Vacchini}\ and\ \citenamefont {Breuer}(2010)}]{Breuer2010}%
  \BibitemOpen
  \bibfield  {author} {\bibinfo {author} {\bibfnamefont {B.}~\bibnamefont {Vacchini}}\ and\ \bibinfo {author} {\bibfnamefont {H.-P.}\ \bibnamefont {Breuer}},\ }\bibfield  {title} {\bibinfo {title} {{Exact master equations for the non-Markovian decay of a qubit}},\ }\href {https://doi.org/10.1103/PhysRevA.81.042103} {\bibfield  {journal} {\bibinfo  {journal} {Phys. Rev. A}\ }\textbf {\bibinfo {volume} {81}},\ \bibinfo {pages} {042103} (\bibinfo {year} {2010})}\BibitemShut {NoStop}%
\bibitem [{\citenamefont {Ishizaki}\ and\ \citenamefont {Tanimura}(2005)}]{Tanimura2005}%
  \BibitemOpen
  \bibfield  {author} {\bibinfo {author} {\bibfnamefont {A.}~\bibnamefont {Ishizaki}}\ and\ \bibinfo {author} {\bibfnamefont {Y.}~\bibnamefont {Tanimura}},\ }\bibfield  {title} {\bibinfo {title} {{Quantum Dynamics of System Strongly Coupled to Low-Temperature Colored Noise Bath: Reduced Hierarchy Equations Approach}},\ }\href {https://doi.org/10.1143/jpsj.74.3131} {\bibfield  {journal} {\bibinfo  {journal} {J. Phys. Soc. Jpn.}\ }\textbf {\bibinfo {volume} {74}},\ \bibinfo {pages} {3131–3134} (\bibinfo {year} {2005})}\BibitemShut {NoStop}%
\bibitem [{\citenamefont {Rosenbach}\ \emph {et~al.}(2016)\citenamefont {Rosenbach}, \citenamefont {Cerrillo}, \citenamefont {Huelga}, \citenamefont {Cao},\ and\ \citenamefont {Plenio}}]{Plenio2016}%
  \BibitemOpen
  \bibfield  {author} {\bibinfo {author} {\bibfnamefont {R.}~\bibnamefont {Rosenbach}}, \bibinfo {author} {\bibfnamefont {J.}~\bibnamefont {Cerrillo}}, \bibinfo {author} {\bibfnamefont {S.~F.}\ \bibnamefont {Huelga}}, \bibinfo {author} {\bibfnamefont {J.}~\bibnamefont {Cao}},\ and\ \bibinfo {author} {\bibfnamefont {M.~B.}\ \bibnamefont {Plenio}},\ }\bibfield  {title} {\bibinfo {title} {{Efficient simulation of non-Markovian system-environment interaction}},\ }\href {https://doi.org/10.1088/1367-2630/18/2/023035} {\bibfield  {journal} {\bibinfo  {journal} {New J. Phys.}\ }\textbf {\bibinfo {volume} {18}},\ \bibinfo {pages} {023035} (\bibinfo {year} {2016})}\BibitemShut {NoStop}%
\bibitem [{\citenamefont {Trivedi}\ \emph {et~al.}(2021)\citenamefont {Trivedi}, \citenamefont {Malz},\ and\ \citenamefont {Cirac}}]{Cirac2021}%
  \BibitemOpen
  \bibfield  {author} {\bibinfo {author} {\bibfnamefont {R.}~\bibnamefont {Trivedi}}, \bibinfo {author} {\bibfnamefont {D.}~\bibnamefont {Malz}},\ and\ \bibinfo {author} {\bibfnamefont {J.~I.}\ \bibnamefont {Cirac}},\ }\bibfield  {title} {\bibinfo {title} {{Convergence Guarantees for Discrete Mode Approximations to Non-Markovian Quantum Baths}},\ }\href {https://doi.org/10.1103/PhysRevLett.127.250404} {\bibfield  {journal} {\bibinfo  {journal} {Phys. Rev. Lett.}\ }\textbf {\bibinfo {volume} {127}},\ \bibinfo {pages} {250404} (\bibinfo {year} {2021})}\BibitemShut {NoStop}%
\bibitem [{\citenamefont {Emerson}\ \emph {et~al.}(2005)\citenamefont {Emerson}, \citenamefont {Alicki},\ and\ \citenamefont {Życzkowski}}]{Emerson2005:randNoiseEstim}%
  \BibitemOpen
  \bibfield  {author} {\bibinfo {author} {\bibfnamefont {J.}~\bibnamefont {Emerson}}, \bibinfo {author} {\bibfnamefont {R.}~\bibnamefont {Alicki}},\ and\ \bibinfo {author} {\bibfnamefont {K.}~\bibnamefont {Życzkowski}},\ }\bibfield  {title} {\bibinfo {title} {{Scalable noise estimation with random unitary operators}},\ }\href {https://doi.org/10.1088/1464-4266/7/10/021} {\bibfield  {journal} {\bibinfo  {journal} {J. Opt. B: Quantum Semiclass. Opt.}\ }\textbf {\bibinfo {volume} {7}},\ \bibinfo {pages} {S347–S352} (\bibinfo {year} {2005})}\BibitemShut {NoStop}%
\bibitem [{\citenamefont {Elben}\ \emph {et~al.}(2023)\citenamefont {Elben}, \citenamefont {Flammia}, \citenamefont {Huang}, \citenamefont {Kueng}, \citenamefont {Preskill}, \citenamefont {Vermersch},\ and\ \citenamefont {Zoller}}]{Elben2023:RandToolbox}%
  \BibitemOpen
  \bibfield  {author} {\bibinfo {author} {\bibfnamefont {A.}~\bibnamefont {Elben}}, \bibinfo {author} {\bibfnamefont {S.~T.}\ \bibnamefont {Flammia}}, \bibinfo {author} {\bibfnamefont {H.-Y.}\ \bibnamefont {Huang}}, \bibinfo {author} {\bibfnamefont {R.}~\bibnamefont {Kueng}}, \bibinfo {author} {\bibfnamefont {J.}~\bibnamefont {Preskill}}, \bibinfo {author} {\bibfnamefont {B.}~\bibnamefont {Vermersch}},\ and\ \bibinfo {author} {\bibfnamefont {P.}~\bibnamefont {Zoller}},\ }\bibfield  {title} {\bibinfo {title} {{The randomized measurement toolbox}},\ }\href {https://doi.org/10.1038/s42254-022-00535-2} {\bibfield  {journal} {\bibinfo  {journal} {Nature Reviews Physics}\ }\textbf {\bibinfo {volume} {5}},\ \bibinfo {pages} {9} (\bibinfo {year} {2023})}\BibitemShut {NoStop}%
\bibitem [{\citenamefont {Li}\ \emph {et~al.}(2023)\citenamefont {Li}, \citenamefont {Colombo}, \citenamefont {Shu}, \citenamefont {Velez}, \citenamefont {Pilatowsky-Cameo}, \citenamefont {Schmied}, \citenamefont {Choi}, \citenamefont {Lukin}, \citenamefont {Pedrozo-Pe{\~n}afiel} \emph {et~al.}}]{Li2023:ScrambleSensing}%
  \BibitemOpen
  \bibfield  {author} {\bibinfo {author} {\bibfnamefont {Z.}~\bibnamefont {Li}}, \bibinfo {author} {\bibfnamefont {S.}~\bibnamefont {Colombo}}, \bibinfo {author} {\bibfnamefont {C.}~\bibnamefont {Shu}}, \bibinfo {author} {\bibfnamefont {G.}~\bibnamefont {Velez}}, \bibinfo {author} {\bibfnamefont {S.}~\bibnamefont {Pilatowsky-Cameo}}, \bibinfo {author} {\bibfnamefont {R.}~\bibnamefont {Schmied}}, \bibinfo {author} {\bibfnamefont {S.}~\bibnamefont {Choi}}, \bibinfo {author} {\bibfnamefont {M.}~\bibnamefont {Lukin}}, \bibinfo {author} {\bibfnamefont {E.}~\bibnamefont {Pedrozo-Pe{\~n}afiel}}, \emph {et~al.},\ }\bibfield  {title} {\bibinfo {title} {{Improving metrology with quantum scrambling}},\ }\href {https://doi.org/10.1126/science.adg9500} {\bibfield  {journal} {\bibinfo  {journal} {Science}\ }\textbf {\bibinfo {volume} {380}},\ \bibinfo {pages} {1381} (\bibinfo {year} {2023})}\BibitemShut {NoStop}%
\bibitem [{\citenamefont {Gong}\ \emph {et~al.}(2026)\citenamefont {Gong}, \citenamefont {Ye}, \citenamefont {Mark},\ and\ \citenamefont {Choi}}]{Gong2026:ScrambleSensing}%
  \BibitemOpen
  \bibfield  {author} {\bibinfo {author} {\bibfnamefont {W.}~\bibnamefont {Gong}}, \bibinfo {author} {\bibfnamefont {B.}~\bibnamefont {Ye}}, \bibinfo {author} {\bibfnamefont {D.}~\bibnamefont {Mark}},\ and\ \bibinfo {author} {\bibfnamefont {S.}~\bibnamefont {Choi}},\ }\href@noop {} {\bibinfo {title} {{Robust multiparameter estimation using quantum scrambling}}} (\bibinfo {year} {2026}),\ \Eprint {https://arxiv.org/abs/2601.23283} {arXiv:2601.23283 [quant-ph]} \BibitemShut {NoStop}%
\bibitem [{\citenamefont {Zhou}\ and\ \citenamefont {Chen}(2026)}]{Zhou2026:RandMeasurements}%
  \BibitemOpen
  \bibfield  {author} {\bibinfo {author} {\bibfnamefont {S.}~\bibnamefont {Zhou}}\ and\ \bibinfo {author} {\bibfnamefont {S.}~\bibnamefont {Chen}},\ }\bibfield  {title} {\bibinfo {title} {{Randomized Measurements for Multiparameter Quantum Metrology}},\ }\href {https://doi.org/10.1103/s27y-gbrp} {\bibfield  {journal} {\bibinfo  {journal} {PRX Quantum}\ }\textbf {\bibinfo {volume} {7}},\ \bibinfo {pages} {010314} (\bibinfo {year} {2026})}\BibitemShut {NoStop}%
\bibitem [{\citenamefont {Zhou}\ \emph {et~al.}(2020)\citenamefont {Zhou}, \citenamefont {Choi}, \citenamefont {Choi}, \citenamefont {Landig}, \citenamefont {Douglas}, \citenamefont {Isoya}, \citenamefont {Jelezko}, \citenamefont {Onoda}, \citenamefont {Sumiya} \emph {et~al.}}]{Zhou2020:QuMetroStrongSpin}%
  \BibitemOpen
  \bibfield  {author} {\bibinfo {author} {\bibfnamefont {H.}~\bibnamefont {Zhou}}, \bibinfo {author} {\bibfnamefont {J.}~\bibnamefont {Choi}}, \bibinfo {author} {\bibfnamefont {S.}~\bibnamefont {Choi}}, \bibinfo {author} {\bibfnamefont {R.}~\bibnamefont {Landig}}, \bibinfo {author} {\bibfnamefont {A.~M.}\ \bibnamefont {Douglas}}, \bibinfo {author} {\bibfnamefont {J.}~\bibnamefont {Isoya}}, \bibinfo {author} {\bibfnamefont {F.}~\bibnamefont {Jelezko}}, \bibinfo {author} {\bibfnamefont {S.}~\bibnamefont {Onoda}}, \bibinfo {author} {\bibfnamefont {H.}~\bibnamefont {Sumiya}}, \emph {et~al.},\ }\bibfield  {title} {\bibinfo {title} {{Quantum Metrology with Strongly Interacting Spin Systems}},\ }\href {https://doi.org/10.1103/PhysRevX.10.031003} {\bibfield  {journal} {\bibinfo  {journal} {Phys. Rev. X}\ }\textbf {\bibinfo {volume} {10}},\ \bibinfo {pages} {031003} (\bibinfo {year} {2020})}\BibitemShut {NoStop}%
\bibitem [{\citenamefont {Colombo}\ \emph {et~al.}(2022)\citenamefont {Colombo}, \citenamefont {Pedrozo-Penafiel}, \citenamefont {Adiyatullin}, \citenamefont {Li}, \citenamefont {Mendez}, \citenamefont {Shu},\ and\ \citenamefont {Vuleti{\'c}}}]{Colombo2022:QSNEcho}%
  \BibitemOpen
  \bibfield  {author} {\bibinfo {author} {\bibfnamefont {S.}~\bibnamefont {Colombo}}, \bibinfo {author} {\bibfnamefont {E.}~\bibnamefont {Pedrozo-Penafiel}}, \bibinfo {author} {\bibfnamefont {A.~F.}\ \bibnamefont {Adiyatullin}}, \bibinfo {author} {\bibfnamefont {Z.}~\bibnamefont {Li}}, \bibinfo {author} {\bibfnamefont {E.}~\bibnamefont {Mendez}}, \bibinfo {author} {\bibfnamefont {C.}~\bibnamefont {Shu}},\ and\ \bibinfo {author} {\bibfnamefont {V.}~\bibnamefont {Vuleti{\'c}}},\ }\bibfield  {title} {\bibinfo {title} {{Time-reversal-based quantum metrology with many-body entangled states}},\ }\href {https://doi.org/10.1038/s41567-022-01653-5} {\bibfield  {journal} {\bibinfo  {journal} {Nat. Phys.}\ }\textbf {\bibinfo {volume} {18}},\ \bibinfo {pages} {925} (\bibinfo {year} {2022})}\BibitemShut {NoStop}%
\bibitem [{\citenamefont {Gao}\ \emph {et~al.}(2025)\citenamefont {Gao}, \citenamefont {Martin}, \citenamefont {Hughes}, \citenamefont {Leitao}, \citenamefont {Put}, \citenamefont {Zhou}, \citenamefont {Koyluoglu}, \citenamefont {Meynell}, \citenamefont {Jayich} \emph {et~al.}}]{Gao2025:NVEcho}%
  \BibitemOpen
  \bibfield  {author} {\bibinfo {author} {\bibfnamefont {H.}~\bibnamefont {Gao}}, \bibinfo {author} {\bibfnamefont {L.~S.}\ \bibnamefont {Martin}}, \bibinfo {author} {\bibfnamefont {L.~B.}\ \bibnamefont {Hughes}}, \bibinfo {author} {\bibfnamefont {N.~T.}\ \bibnamefont {Leitao}}, \bibinfo {author} {\bibfnamefont {P.}~\bibnamefont {Put}}, \bibinfo {author} {\bibfnamefont {H.}~\bibnamefont {Zhou}}, \bibinfo {author} {\bibfnamefont {N.~U.}\ \bibnamefont {Koyluoglu}}, \bibinfo {author} {\bibfnamefont {S.~A.}\ \bibnamefont {Meynell}}, \bibinfo {author} {\bibfnamefont {A.~C.~B.}\ \bibnamefont {Jayich}}, \emph {et~al.},\ }\bibfield  {title} {\bibinfo {title} {{Signal amplification in a solid-state sensor through asymmetric many-body echo}},\ }\href {https://doi.org/10.1038/s41586-025-09452-7} {\bibfield  {journal} {\bibinfo  {journal} {Nature}\ }\textbf {\bibinfo {volume} {646}},\ \bibinfo {pages} {68} (\bibinfo {year} {2025})}\BibitemShut {NoStop}%
\bibitem [{\citenamefont {Wu}\ \emph {et~al.}(2025)\citenamefont {Wu}, \citenamefont {Davis}, \citenamefont {Hughes}, \citenamefont {Ye}, \citenamefont {Wang}, \citenamefont {Kufel}, \citenamefont {Ono}, \citenamefont {Meynell}, \citenamefont {Block} \emph {et~al.}}]{Wu2025:NVSpinSqz}%
  \BibitemOpen
  \bibfield  {author} {\bibinfo {author} {\bibfnamefont {W.}~\bibnamefont {Wu}}, \bibinfo {author} {\bibfnamefont {E.~J.}\ \bibnamefont {Davis}}, \bibinfo {author} {\bibfnamefont {L.~B.}\ \bibnamefont {Hughes}}, \bibinfo {author} {\bibfnamefont {B.}~\bibnamefont {Ye}}, \bibinfo {author} {\bibfnamefont {Z.}~\bibnamefont {Wang}}, \bibinfo {author} {\bibfnamefont {D.}~\bibnamefont {Kufel}}, \bibinfo {author} {\bibfnamefont {T.}~\bibnamefont {Ono}}, \bibinfo {author} {\bibfnamefont {S.~A.}\ \bibnamefont {Meynell}}, \bibinfo {author} {\bibfnamefont {M.}~\bibnamefont {Block}}, \emph {et~al.},\ }\bibfield  {title} {\bibinfo {title} {{Spin squeezing in an ensemble of nitrogen-vacancy centres in diamond}},\ }\href {https://doi.org/10.1038/s41586-025-09524-8} {\bibfield  {journal} {\bibinfo  {journal} {Nature}\ }\textbf {\bibinfo {volume} {646}},\ \bibinfo {pages} {74} (\bibinfo {year} {2025})}\BibitemShut {NoStop}%
\bibitem [{\citenamefont {Montenegro}\ \emph {et~al.}(2025)\citenamefont {Montenegro}, \citenamefont {Mukhopadhyay}, \citenamefont {Yousefjani}, \citenamefont {Sarkar}, \citenamefont {Mishra}, \citenamefont {Paris},\ and\ \citenamefont {Bayat}}]{Montenegro2024:ManyBodyMetrology}%
  \BibitemOpen
  \bibfield  {author} {\bibinfo {author} {\bibfnamefont {V.}~\bibnamefont {Montenegro}}, \bibinfo {author} {\bibfnamefont {C.}~\bibnamefont {Mukhopadhyay}}, \bibinfo {author} {\bibfnamefont {R.}~\bibnamefont {Yousefjani}}, \bibinfo {author} {\bibfnamefont {S.}~\bibnamefont {Sarkar}}, \bibinfo {author} {\bibfnamefont {U.}~\bibnamefont {Mishra}}, \bibinfo {author} {\bibfnamefont {M.~G.}\ \bibnamefont {Paris}},\ and\ \bibinfo {author} {\bibfnamefont {A.}~\bibnamefont {Bayat}},\ }\bibfield  {title} {\bibinfo {title} {{Review: Quantum metrology and sensing with many-body systems}},\ }\href {https://doi.org/10.1016/j.physrep.2025.05.005} {\bibfield  {journal} {\bibinfo  {journal} {Phys. Rep.}\ }\textbf {\bibinfo {volume} {1134}},\ \bibinfo {pages} {1} (\bibinfo {year} {2025})}\BibitemShut {NoStop}%
\bibitem [{\citenamefont {Kubo}(1962)}]{Kubo1962}%
  \BibitemOpen
  \bibfield  {author} {\bibinfo {author} {\bibfnamefont {R.}~\bibnamefont {Kubo}},\ }\bibfield  {title} {\bibinfo {title} {{Generalized Cumulant Expansion Method}},\ }\href {https://doi.org/10.1143/JPSJ.17.1100} {\bibfield  {journal} {\bibinfo  {journal} {J. Phys. Soc. Jpn.}\ }\textbf {\bibinfo {volume} {17}},\ \bibinfo {pages} {1100} (\bibinfo {year} {1962})}\BibitemShut {NoStop}%
\bibitem [{\citenamefont {Pang}\ and\ \citenamefont {Brun}(2014)}]{Pang2014:TimeDepHam}%
  \BibitemOpen
  \bibfield  {author} {\bibinfo {author} {\bibfnamefont {S.}~\bibnamefont {Pang}}\ and\ \bibinfo {author} {\bibfnamefont {T.~A.}\ \bibnamefont {Brun}},\ }\bibfield  {title} {\bibinfo {title} {{Quantum metrology for a general Hamiltonian parameter}},\ }\href {https://doi.org/10.1103/PhysRevA.90.022117} {\bibfield  {journal} {\bibinfo  {journal} {Phys. Rev. A}\ }\textbf {\bibinfo {volume} {90}},\ \bibinfo {pages} {022117} (\bibinfo {year} {2014})}\BibitemShut {NoStop}%
\bibitem [{\citenamefont {Pang}\ and\ \citenamefont {Jordan}(2017)}]{Pang2017:TimeDepHam}%
  \BibitemOpen
  \bibfield  {author} {\bibinfo {author} {\bibfnamefont {S.}~\bibnamefont {Pang}}\ and\ \bibinfo {author} {\bibfnamefont {A.~N.}\ \bibnamefont {Jordan}},\ }\bibfield  {title} {\bibinfo {title} {{Optimal adaptive control for quantum metrology with time-dependent Hamiltonians}},\ }\href {https://doi.org/10.1038/ncomms14695} {\bibfield  {journal} {\bibinfo  {journal} {Nat. Commun.}\ }\textbf {\bibinfo {volume} {8}},\ \bibinfo {pages} {14695} (\bibinfo {year} {2017})}\BibitemShut {NoStop}%
\bibitem [{\citenamefont {Provost}\ and\ \citenamefont {Vallee}(1980)}]{Provost1980}%
  \BibitemOpen
  \bibfield  {author} {\bibinfo {author} {\bibfnamefont {J.}~\bibnamefont {Provost}}\ and\ \bibinfo {author} {\bibfnamefont {G.}~\bibnamefont {Vallee}},\ }\bibfield  {title} {\bibinfo {title} {{Riemannian structure on manifolds of quantum states}},\ }\href {https://doi.org/10.1007/BF02193559} {\bibfield  {journal} {\bibinfo  {journal} {Commun. Math. Phys.}\ }\textbf {\bibinfo {volume} {76}},\ \bibinfo {pages} {289} (\bibinfo {year} {1980})}\BibitemShut {NoStop}%
\bibitem [{\citenamefont {Carollo}\ \emph {et~al.}(2018)\citenamefont {Carollo}, \citenamefont {Spagnolo},\ and\ \citenamefont {Valenti}}]{Carollo2018:UhlmannCurvature}%
  \BibitemOpen
  \bibfield  {author} {\bibinfo {author} {\bibfnamefont {A.}~\bibnamefont {Carollo}}, \bibinfo {author} {\bibfnamefont {B.}~\bibnamefont {Spagnolo}},\ and\ \bibinfo {author} {\bibfnamefont {D.}~\bibnamefont {Valenti}},\ }\bibfield  {title} {\bibinfo {title} {{Uhlmann curvature in dissipative phase transitions}},\ }\href {https://doi.org/10.1038/s41598-018-27362-9} {\bibfield  {journal} {\bibinfo  {journal} {Sci. Rep.}\ }\textbf {\bibinfo {volume} {8}},\ \bibinfo {pages} {9852} (\bibinfo {year} {2018})}\BibitemShut {NoStop}%
\bibitem [{\citenamefont {Sekatski}\ \emph {et~al.}(2017)\citenamefont {Sekatski}, \citenamefont {Skotiniotis}, \citenamefont {Ko{\l{}}ody{\'{n}}ski},\ and\ \citenamefont {D{\"{u}}r}}]{Sekatski2017:FullFast}%
  \BibitemOpen
  \bibfield  {author} {\bibinfo {author} {\bibfnamefont {P.}~\bibnamefont {Sekatski}}, \bibinfo {author} {\bibfnamefont {M.}~\bibnamefont {Skotiniotis}}, \bibinfo {author} {\bibfnamefont {J.}~\bibnamefont {Ko{\l{}}ody{\'{n}}ski}},\ and\ \bibinfo {author} {\bibfnamefont {W.}~\bibnamefont {D{\"{u}}r}},\ }\bibfield  {title} {\bibinfo {title} {{Quantum metrology with full and fast quantum control}},\ }\href {https://doi.org/10.22331/q-2017-09-06-27} {\bibfield  {journal} {\bibinfo  {journal} {{Quantum}}\ }\textbf {\bibinfo {volume} {1}},\ \bibinfo {pages} {27} (\bibinfo {year} {2017})}\BibitemShut {NoStop}%
\bibitem [{\citenamefont {Demkowicz-Dobrza\ifmmode~\acute{n}\else \'{n}\fi{}ski}\ \emph {et~al.}(2017)\citenamefont {Demkowicz-Dobrza\ifmmode~\acute{n}\else \'{n}\fi{}ski}, \citenamefont {Czajkowski},\ and\ \citenamefont {Sekatski}}]{Demkowicz2017:Adaptive}%
  \BibitemOpen
  \bibfield  {author} {\bibinfo {author} {\bibfnamefont {R.}~\bibnamefont {Demkowicz-Dobrza\ifmmode~\acute{n}\else \'{n}\fi{}ski}}, \bibinfo {author} {\bibfnamefont {J.}~\bibnamefont {Czajkowski}},\ and\ \bibinfo {author} {\bibfnamefont {P.}~\bibnamefont {Sekatski}},\ }\bibfield  {title} {\bibinfo {title} {{Adaptive Quantum Metrology under General Markovian Noise}},\ }\href {https://doi.org/10.1103/PhysRevX.7.041009} {\bibfield  {journal} {\bibinfo  {journal} {Phys. Rev. X}\ }\textbf {\bibinfo {volume} {7}},\ \bibinfo {pages} {041009} (\bibinfo {year} {2017})}\BibitemShut {NoStop}%
\bibitem [{\citenamefont {Zhou}\ \emph {et~al.}(2018)\citenamefont {Zhou}, \citenamefont {Zhang}, \citenamefont {Preskill},\ and\ \citenamefont {Jiang}}]{Zhou2018:HNLS}%
  \BibitemOpen
  \bibfield  {author} {\bibinfo {author} {\bibfnamefont {S.}~\bibnamefont {Zhou}}, \bibinfo {author} {\bibfnamefont {M.}~\bibnamefont {Zhang}}, \bibinfo {author} {\bibfnamefont {J.}~\bibnamefont {Preskill}},\ and\ \bibinfo {author} {\bibfnamefont {L.}~\bibnamefont {Jiang}},\ }\bibfield  {title} {\bibinfo {title} {{Achieving the Heisenberg limit in quantum metrology using quantum error correction}},\ }\href {https://doi.org/10.1038/s41467-017-02510-3} {\bibfield  {journal} {\bibinfo  {journal} {Nat. Commun.}\ }\textbf {\bibinfo {volume} {9}},\ \bibinfo {pages} {78} (\bibinfo {year} {2018})}\BibitemShut {NoStop}%
\bibitem [{\citenamefont {Zhou}\ and\ \citenamefont {Jiang}(2021)}]{Zhou2021:Asymptotic}%
  \BibitemOpen
  \bibfield  {author} {\bibinfo {author} {\bibfnamefont {S.}~\bibnamefont {Zhou}}\ and\ \bibinfo {author} {\bibfnamefont {L.}~\bibnamefont {Jiang}},\ }\bibfield  {title} {\bibinfo {title} {{Asymptotic Theory of Quantum Channel Estimation}},\ }\href {https://doi.org/10.1103/PRXQuantum.2.010343} {\bibfield  {journal} {\bibinfo  {journal} {PRX Quantum}\ }\textbf {\bibinfo {volume} {2}},\ \bibinfo {pages} {010343} (\bibinfo {year} {2021})}\BibitemShut {NoStop}%
\end{thebibliography}%

%\newpage 
\onecolumngrid
\appendix

% 1. Automatically reset numbers by section and add the dot (A.1)
\counterwithin{equation}{section}
\counterwithin{thm}{section}
\counterwithin{prop}{section}
\counterwithin{lem}{section}
\counterwithin{cor}{section}

% 2. Sync the Hyperlinks so they don't jump to the Main Text
\makeatletter

% Printed equation numbers: A1, A2, ... / B1, B2, ...
\renewcommand{\theequation}{\thesection\arabic{equation}}

% Unique hyperref anchors for appendix sections/equations
\renewcommand{\theHsection}{app.\Alph{section}}
\renewcommand{\theHequation}{app.\Alph{section}.\arabic{equation}}
\renewcommand{\theHthm}{\thesection.\arabic{thm}}
\renewcommand{\theHprop}{\thesection.\arabic{prop}}
\renewcommand{\theHlem}{\thesection.\arabic{lem}}
\renewcommand{\theHcor}{\thesection.\arabic{cor}}
\makeatother

%==========
%==========
\section{Details on the Collisional Purification}
\label{app:collisional}

In this Appendix, we provide a self-contained derivation of the collisional purification (or purification) used to realize Markovian open-system dynamics [see Sec.~\ref{sec:markov-evol}], by showing (i) how a vacuum, memoryless bosonic bath produces a generic Lindblad dissipator, and (ii) how the same physics arises as the continuous-time limit of a discrete repeated-interaction (collision) model. Throughout, the dilation provides a \emph{purification} of the reduced dynamics in the sense that if the probe is initialized in a pure state $\rho(0)=\psi$, with $\psi=\dyad{\psi}$, and each bath component is initialized in the vacuum, then the joint system-environment state remains pure under the global unitary evolution, and the reduced system's time evolution exactly agrees with the dynamics predicted by the Lindblad master equation.

\paragraph*{Target dynamics.}
We take the reduced probe state $\rho(t)$ to satisfy a parametrized Lindblad master equation,
\begin{equation}
\label{eq:app_lindblad}
\partial_t \rho(t) = -i\comm{H_c(t)}{\rho(t)} + \mathscr{D}_{\bm{\theta}}(\rho(t)),
\end{equation}
where $H_c(t)$ is a (parameter-independent) control Hamiltonian, and the parametrized dissipator is written in the canonical form as
\begin{equation}
\label{eq:app_dissip_can}
\mathscr{D}_{\bm{\theta}}(\rho)
=\sum_{k=1}^R \gamma_k(\bm{\theta})
\left(J_k(\bm{\theta})\,\rho\, J_k^\dagger(\bm{\theta})
-\tfrac12\acomm{J_k^\dagger(\bm{\theta})J_k(\bm{\theta})}{\rho}\right).
\end{equation}
We realize \eqref{eq:app_lindblad} by a unitary dilation on $\mathscr{H}_S\otimes \mathscr{H}_{\rm env}$ generated by
\begin{equation}
\label{eq:app_Htot}
H_{\rm tot}(\bm{\theta}; t) = H_c(t) + H(\bm{\theta}; t),
\end{equation}
where the system-bath interaction reads
\begin{equation}
\label{eq:app_continuum_H}
H(\bm{\theta}; t)
= \sum_{k=1}^R \sqrt{\gamma_k(\bm{\theta})}\,J_k(\bm{\theta})\otimes e_k^\dagger(t) + {\rm h.c.},
\end{equation}
The bath operators $\{e_k(t)\}$ are bosonic annihilation operators satisfying
\begin{equation}
\label{eq:app_white_noise}
\comm{e_k(t)}{e_{k'}^\dagger(t')}=\delta_{kk'}\delta(t-t'),\qquad
\Tr\!\big(\varphi_{\rm env}e_k(t)\big)=0,\qquad
\Tr\!\big(\varphi_{\rm env}e_k(t)e_{k'}^\dagger(t')\big)=\delta_{kk'}\delta(t-t'),
\end{equation}
with all other second moments vanishing in the vacuum. The state of the system is obtained from the global state
$\Psi(t) \triangleq U(t)\,[\psi \otimes\varphi_{\rm env}]\,U^\dagger(t)$ by $\rho(t)=\Tr_{\rm env}\Psi(t)$ with $\rho(0)=\psi$. We define this construction as the canonical collisional purification. 

%==========
\subsection{Collisional Model Implies Markovian Evolution}
\label{app:markov-kubo}

We sketch how the memoryless bath specified in \eqref{eq:app_continuum_H}--\eqref{eq:app_white_noise} yields the Lindblad form \eqref{eq:app_lindblad}, emphasizing how Gaussianity and Markovianity reduce the Kubo cumulant expansion to second order~\cite{Kubo1962}.

\paragraph*{Interaction picture.}
First, define $U_c(t)=\mathcal{T}\exp\!\big(-i\int_0^t \dd{s}\,H_c(s)\big)$ and the reduced state
$\tilde{\rho}(t)=U_c^\dagger(t)\rho(t)U_c(t)$ in the system's interaction picture governed by quantum control.
The interaction-picture propagator on $\mathscr{H}_S\otimes\mathscr{H}_{\rm env}$ is
$U_I(t)=\mathcal{T}\exp\!\big(-i\int_0^t\dd{s}\,\tilde{H}(s)\big)$, where
\begin{equation}
\label{eq:app_HI}
\tilde{H}(t)
= \sum_{k=1}^R \sqrt{\gamma_k(\bm{\theta})}\,\tilde{J}_k(\bm{\theta};t)\otimes e_k^\dagger(t) +{\rm h.c.},\qquad
\tilde{J}_k(\bm{\theta};t)\triangleq U_c^\dagger(t)J_k(\bm{\theta})U_c(t).
\end{equation}
The reduced interaction-picture dynamics are
\begin{equation}
\label{eq:app_reduced_I}
\tilde{\rho}(t)=\Tr_{\rm env}\!\left[U_I(t)\,\big(\psi\otimes\varphi_{\rm env}\big)\,U_I^\dagger(t)\right],
\end{equation}
with $\tilde{\rho}(0)=\psi$.

\paragraph*{Cumulant expansion and Gaussian truncation.}
The Kubo cumulant expansion~\cite{Kubo1962} expresses the reduced map generated by $\tilde{H}(t)$ in terms of cumulant superoperators. Schematically,
\begin{equation}
\label{eq:app_kubo}
\tilde{\rho}(t)=\exp\!\left[\sum_{n=1}^\infty K_n(t)\right]\psi,
\end{equation}
the first two cumulants are
\begin{align}
\label{eq:app_cumulants}
K_1(t)\rho &= -i\int_0^t \dd{s}\,\Tr_{\rm env}\!\big[\tilde{H}(s),\rho\otimes\varphi_{\rm env}\big],\\
K_2(t)\rho &= -\int_0^t \dd{s}\int_0^s \dd{s}'\,\Tr_{\rm env}\!\big[\tilde{H}(s),[\tilde{H}(s'),\rho\otimes\varphi_{\rm env}]\big].
\end{align}
Since $\Tr(\varphi_{\rm env}e_k(t))=0$, one has $K_1(t)=0$. Moreover, the bath is Gaussian and Markovian, so all connected cumulants beyond second order vanish, $K_{n>2}(t)=0$. Thus the reduced dynamics are entirely captured by $K_2(t)$.

\paragraph*{Time-local evolution via white noise.}
Finally, the Markov assumption \eqref{eq:app_white_noise} collapses the double integral in $K_2(t)$ to a single time integral: the bath trace produces $\delta(s-s')$, so the contribution from $s'\approx s$ dominates and the evolution becomes time-local. Concretely,
\begin{equation}
\label{eq:app_K2_local}
K_2(t)\rho=\int_0^t \dd{s}\,\tilde{\mathscr{D}}_{\bm{\theta},s}(\rho),
\end{equation}
where the instantaneous interaction-picture dissipator is
\begin{equation}
\label{eq:app_dissip_I}
\tilde{\mathscr{D}}_{\bm{\theta},t}(\rho)
=\sum_{k=1}^R \gamma_k(\bm{\theta})
\left(\tilde{J}_k(\bm{\theta};t)\rho \tilde{J}_k^\dagger(\bm{\theta};t)
-\tfrac12\acomm{\tilde{J}_k^\dagger(\bm{\theta};t)\tilde{J}_k(\bm{\theta};t)}{\rho}\right).
\end{equation}
Therefore,
\begin{equation}
\label{eq:app_solution_I}
\tilde{\rho}(t)=\mathcal{T}\exp\!\left(\int_0^t \dd{s}\,\tilde{\mathscr{D}}_{\bm{\theta},s}\right)(\psi),
\end{equation}
and transforming back to the Schr\"odinger picture yields the Lindblad master equation \eqref{eq:app_lindblad} with dissipator \eqref{eq:app_dissip_can}.

\subparagraph*{Remark.}
It is sometimes convenient to keep a fixed system operator basis $\{L_i\}_{i=1}^R$ and encode all parameter dependence in bath correlations. In this case, we write the interaction as
\begin{equation}
\label{eq:app_phys_basis}
H(t)=\sum_{i=1}^R L_i\otimes f_i^\dagger(t)+{\rm h.c.},\qquad
\Tr(\varphi_{\rm env}f_i(t)f_j^\dagger(t'))=\Gamma_{ij}(\bm{\theta})\delta(t-t'), \qquad \Tr(\varphi_{\rm env}f_i(t))=0,
\end{equation}
where $\Gamma(\bm{\theta})$ is Hermitian and positive. Diagonalizing $\Gamma(\bm{\theta})=V(\bm{\theta})\Lambda(\bm{\theta})V^\dagger(\bm{\theta})$ with $\Lambda=\mathrm{diag}(\gamma_k)$ and defining $e_k(t)=\sum_i V_{ik}^\dagger(\bm{\theta}) f_i(t)$ recovers the independent vacuum noise channels \eqref{eq:app_white_noise} and the canonical jumps $J_k(\bm{\theta})=\sum_i V_{ik}(\bm{\theta})L_i$ appearing in \eqref{eq:app_dissip_can}. This is usually the physical picture that we have in mind when probing a correlated stochastic force, field, or many-body system.

%==========
\subsection{Discrete-Time Model}
\label{app:discrete-collisional}

In this Appendix, we construct an explicit discrete-time dilation that converges to the Lindblad evolution \eqref{eq:app_lindblad} in the limit $\Delta t\to 0$. In Appendix~\ref{app:lem-qfi-flow}, we use this formalism to derive the QFI flow lemma (Lemma~\ref{lem:qfi-flow}) of the main text.

Let the environment be a chain of independent, identical bath elements $E_1,E_2,\dots,E_n$,
\begin{equation}
\label{eq:app_chain}
\mathscr{H}_{\rm env}=\bigotimes_{j=1}^n \mathscr{H}_{E_j},\qquad
\varphi_{\rm env}=\bigotimes_{j=1}^n \varphi_{E_j},\qquad
\varphi_{E_j}=\dyad{0}.
\end{equation}
For $R$ jump channels, take each $E_j$ to consist of $R$ bosonic modes with annihilation operators $\{b_{k,j}\}_{k=1}^R$ satisfying
\begin{equation}
\label{eq:app_b_comm}
\comm{b_{k,j}}{b_{k',j'}^\dagger}=\delta_{kk'}\delta_{jj'},\qquad
\Tr(\varphi_{E_j} b_{k,j})=0,\qquad
\Tr(\varphi_{E_j} b_{k,j} b_{k',j'}^\dagger)=\delta_{kk'} \delta _{jj'}.
\end{equation}
Discretize time as $t_j=j\Delta t$. Over the $j$th time bin, we take the joint unitary
\begin{equation}
\label{eq:app_Uj}
U_j(\bm{\theta})
=\exp\!\big(-i\Delta t\,H_c(t_j)\big)\,\exp\!\Big(-i\sqrt{\Delta t}\,H_j(\bm{\theta})\Big),
\end{equation}
with system-bath interaction Hamiltonian
\begin{equation}
\label{eq:app_Hj}
H_j(\bm{\theta})
\triangleq \sum_{k=1}^R \sqrt{\gamma_k(\bm{\theta})}
\left(J_k(\bm{\theta})\otimes b_{k,j}^\dagger + J_k^\dagger(\bm{\theta})\otimes b_{k,j}\right).
\end{equation}
The $\sqrt{\Delta t}$ scaling is chosen so that a single collision produces an $\order{\Delta t}$ change in the \emph{reduced} state:
after tracing over a vacuum ancilla, the first-order term in $\sqrt{\Delta t}$ vanishes because $\Tr(\varphi_{E_j} b_{k,j})=0$,
while the second-order term survives because $\Tr(\varphi_{E_j} b_{k,j} b_{k',j}^\dagger)=\delta_{kk'}$, producing the dissipative
increment that remains finite as $\Delta t\to 0$.

Define $\rho_j\triangleq \rho(t_j)$. One collision implements the CPTP map
\begin{equation}
\label{eq:app_map}
\rho_{j+1}=\mathcal{E}_{j,\bm{\theta}}(\rho_j)
\triangleq \Tr_{E_j}\!\Big[U_j(\bm{\theta})\big(\rho_j\otimes\varphi_{E_j}\big)U_j^\dagger(\bm{\theta})\Big].
\end{equation}
Expanding $e^{-i\sqrt{\Delta t}H_j}$ to second order and tracing over the vacuum yields, up to $\order{\Delta t^2}$,
a Kraus decomposition,
\begin{equation}
\label{eq:app_kraus}
\rho_{j+1}=K_{0,j}\rho_j K_{0,j}^\dagger+\sum_{k=1}^R K_{k,j}\rho_j K_{k,j}^\dagger+\order{\Delta t^2},
\end{equation}
with
\begin{equation}
\label{eq:app_kraus_ops}
K_{0,j} = I - i\Delta t\,H_c(t_j) - \frac{\Delta t}{2}\sum_{k=1}^R \gamma_k(\bm{\theta})\,J_k^\dagger(\bm{\theta})J_k(\bm{\theta}),
\qquad
K_{k,j} = \sqrt{\Delta t\,\gamma_k(\bm{\theta})}\,J_k(\bm{\theta}).
\end{equation}
A short calculation then gives the increment
\begin{equation}
\label{eq:app_increment}
\rho_{j+1}-\rho_j=\Delta t\,\mathcal{L}_{\bm{\theta},t_j}(\rho_j)+\order{\Delta t^2},
\end{equation}
where
\begin{equation}
\label{eq:app_generator}
\mathcal{L}_{\bm{\theta},t}(\rho)
= -i\comm{H_c(t)}{\rho}
+\sum_{k=1}^R \gamma_k(\bm{\theta})
\left(J_k(\bm{\theta})\rho J_k^\dagger(\bm{\theta})
-\tfrac12\acomm{J_k^\dagger(\bm{\theta})J_k(\bm{\theta})}{\rho}\right).
\end{equation}
Thus $\mathcal{L}_{\bm{\theta},t}$ coincides with the Lindblad generator in \eqref{eq:app_lindblad}--\eqref{eq:app_dissip_can}.

Iterating $n=t/\Delta t$ steps gives the global unitary
\begin{equation}
\label{eq:app_U1n}
U_{1:n}(\bm{\theta})\triangleq U_n(\bm{\theta})\cdots U_1(\bm{\theta}),
\end{equation}
and the joint state
$\Psi_{1:n}(\bm{\theta})=U_{1:n}(\bm{\theta})\big(\psi\otimes\varphi_{\rm env}\big)U_{1:n}^\dagger(\bm{\theta})$,
which is pure whenever $\psi$ is pure.
The reduced state is $\rho(t)=\Tr_{\rm env}\Psi_{1:n}(\bm{\theta})$.
Taking $\Delta t\to 0$ with $t$ fixed yields the time-ordered exponential solution,
\begin{equation}
\label{eq:app_cont_limit}
\lim_{\Delta t\to 0}\rho(t)=\mathcal{T}\exp\!\left(\int_0^t \dd{s}\,\mathcal{L}_{\bm{\theta},s}\right)(\psi),
\end{equation}
which is the continuous-time Markovian evolution generated by \eqref{eq:app_lindblad}. 

%==========
%==========
\section{Proofs of Main Results}
\label{app:proofs}

%==========
\subsection{Proof of Lemma~\ref{lem:qfi-flow}}
\label{app:lem-qfi-flow}

In this Appendix, we prove that the QFI for estimating purely dissipative Markovian processes grows at most linear in time,  Lemma~\ref{lem:qfi-flow} of the main text, which we restate here for convenience. For simplicity of presentation, we consider one jump operator $L$ with parametrized dissipation rate $\gamma_\theta$. The extension to multiple jumps and multiple parameters is straightforward, since the argument relies only on the memoryless structure of the bath.

\QFIFlow*

Before proving the result, we state the following Lemma regarding QFI structure for unitary models, which we then apply to our collisional purification.

\begin{lem}
    \label{lem:unitary-qfi}
    Consider pure-state model of the QFI matrix~\cite{Matsumoto2002}, with parameters encoded into an initial probe state $\Psi$ by a unitary process, $\ket{\Psi(\theta)} = U(\theta; t) \ket{\Psi}$. Let $\Delta O \triangleq O - \expval{O}$, and define the parameter generator
    \begin{equation}
         G(\theta; t) \triangleq i U^\dagger (\partial_\theta U) = \int_0^t \dd{s}\, g(\theta; s),
    \end{equation}
    with generator density
    \begin{equation}\label{eq:generator-rate}
        g(\theta; s)
        = U^\dagger(\theta;s)\partial_\theta H(\theta;s)U(\theta;s).
    \end{equation}
    Then, the QFI obeys~\cite{Boixo2007:QuEstim}
    \begin{equation}
        F_Q =4\,\Var_{\Psi}(G(\theta; s))
        = 4 \expval{\Delta G(\theta;t)^2}{\Psi},
    \end{equation}
    with expectation taken with respect to the initial state.
\end{lem}

For time-\emph{independent} Hamiltonians, Lemma~\ref{lem:unitary-qfi} implies the famous Heisenberg-scaling in $t$, $(F_Q)_{aa} \leq 4 t^2 \norm{\partial_a H(\bm{\theta})}$. However, this scaling is not generic and changes when $H$ depends on time~\cite{Pang2014:TimeDepHam, Pang2017:TimeDepHam}. In particular, for the memoryless (Markovian) processes, $F_Q(t) = \order{t}$, as we now show.

\begin{proof}[Proof of Lemma~\ref{lem:qfi-flow}]
    Without loss of generality, we consider single-parameter estimation and one jump operator $L$ coupled to a discrete set of independent bosonic bath elements $b_j$ (see Appendix~\ref{app:collisional}). We work with the canonical collisional purification: the global state $\Psi(\theta;t)$ on system plus bath is pure, and the reduced (system-only) state is obtained by tracing out the bath. Since the partial trace is a $\theta$-independent CPTP map,
    \begin{equation}
        \label{eq:qfi_monotone}
        F_Q\big(\rho(\theta;t)\big)\leq  F_Q\big(\Psi(\theta;t)\big).
    \end{equation}
    In this proof, we therefore compute the QFI of the purified (collisional) model; for brevity we write $F_Q(t)\triangleq F_Q(\Psi(\theta;t))$.
    
    \paragraph*{Generator decomposition.}
    Let $t=n\Delta t$ and define the joint unitary up to time $t$ by\footnote{We ignore the control Hamiltonian $H_c$ for simplicity of presentation, as this does not alter the main result.}
    \begin{equation}
    \label{eq:U1n_def}
        U_{1:n}(\theta)\triangleq U_n(\theta)\cdots U_1(\theta),
        \qquad
        U_j(\theta)=\exp\!\big(-i\sqrt{\Delta t}\,H_j(\theta)\big),
    \end{equation}
    with interaction Hamiltonian
    \begin{equation}
        \label{eq:Hj_def}
        H_{j}(\theta)=\sqrt{\gamma_\theta}\left(L\otimes b_j^\dagger + L^\dagger\otimes b_j\right),
    \end{equation}
    where each bath element is initialized in vacuum, $\varphi_{E_j}=\dyad{0}$, so that
    \begin{equation}
        \label{eq:vac_moments}
        \Tr(\varphi_{E_j} b_j)=0,
        \qquad
        \Tr(\varphi_{E_j} b_j b_j^\dagger)=1,
        \qquad
        \Tr(\varphi_{E_j} b_j^\dagger b_j)=\Tr(\varphi_{E_j} b_j^2)=\Tr(\varphi_{E_j}b_j^{\dagger\, 2})=0.
    \end{equation}
    Differentiating $H_j(\theta)$ gives
    \begin{equation}
    \label{eq:partial-Hj}
        \partial_\theta H_{j}
        = A\otimes b_j^\dagger + A^\dagger \otimes b_j,
        \qquad
        A \triangleq \partial_\theta\sqrt{\gamma_\theta}\, L.
    \end{equation}
    
    For total interaction time $t=n\Delta t$, define the generator
    \begin{equation}
        \label{eq:Gn_def}
        G_n(\theta)\triangleq i\,U_{1:n}^\dagger(\theta)\,\partial_\theta U_{1:n}(\theta).
    \end{equation}
    A standard product-rule computation yields the decomposition
    \begin{equation}
        \label{eq:Gn_sum}
        G_n(\theta)=\sum_{j=1}^n \widetilde g_j(\theta),
        \qquad
        \widetilde g_j(\theta)\triangleq U_{1:j-1}^\dagger(\theta)\,g_j(\theta)\,U_{1:j-1}(\theta),
        \qquad
        g_j(\theta)\triangleq i\,U_j^\dagger(\theta)\,\partial_\theta U_j(\theta).
    \end{equation}
    Moreover, $g_j(\theta)$ admits the integral representation
    \begin{equation}
        \label{eq:gj_integral}
        g_j(\theta)
        =
        \sqrt{\Delta t}\int_0^1 \dd{u}\,
        e^{iu\sqrt{\Delta t}H_j(\theta)}\big(\partial_\theta H_j(\theta)\big)e^{-iu\sqrt{\Delta t}H_j(\theta)},
    \end{equation}
    and therefore $g_j(\theta) \approx \sqrt{\Delta t}\, \partial_\theta H_j + \order{\Delta t}$.
    
    \paragraph*{Memoryless additivity.}
    Because $U_{1:j-1}$ acts trivially on the fresh bath element $E_j$, the operator $\widetilde g_j(\theta)$ acts nontrivially only on the system and $E_j$. Since the bath is initially a product of vacua, the cross moments factorize for $j\neq j'$ and, using $\Tr(\varphi_{E_j}b_j)=0$, one has $\expval{\widetilde g_j(\theta)}=0$ and
    \begin{equation}
        \label{eq:cross_vanish}
        \expval{\widetilde g_j(\theta)\,\widetilde g_{j'}(\theta)}
        =
        \expval{\widetilde g_j(\theta)}\expval{ \widetilde g_{j'}(\theta)}
        =0
        \qquad (j\neq j').
    \end{equation}
    Hence the variance is additive,
    \begin{equation}
    \label{eq:var_add}
        \Var(G_n(\theta))
        =\sum_{j=1}^n \Var(\widetilde g_j(\theta))
        = \sum_{j=1}^n \expval{\widetilde g_j(\theta)^2}.
    \end{equation}
    To leading order in $\Delta t$, Eq.~\eqref{eq:gj_integral} implies $\widetilde g_j(\theta)=\sqrt{\Delta t}\,U_{1:j-1}^\dagger(\partial_\theta H_j)U_{1:j-1}+\order{\Delta t}$. Since the bath is the bosonic vacuum, whose odd moments vanish, any contribution to $\expval{\widetilde g_j(\theta)^2}$ at order $\order{\Delta t^{k/2}}$ with $k$ odd vanishes in expectation. Therefore,
    \begin{equation}
    \label{eq:local_second_moment}
        \expval{\widetilde g_j(\theta)^2}
        =
        \Delta t\,\Tr_{SE_j}\!\left[(\partial_\theta H_j)^2\;\rho_{j-1}(\theta)\otimes\varphi_{E_j}\right]
        +\order{\Delta t^2},
        \qquad
        t_{j-1}=(j-1)\Delta t,
    \end{equation}
    where $\rho_{j-1}(\theta)$ is the system state immediately before the $j$th collision~\eqref{eq:app_map}. Using \eqref{eq:partial-Hj} and the vacuum moments \eqref{eq:vac_moments},
    \begin{align}
        \Tr_{E_j}\left(\varphi_{E_j}\,(\partial_\theta H_j)^2\right)
        &=
        \Tr_{E_j}\left(\varphi_{E_j}\,(A\otimes b_j^\dagger + A^\dagger\otimes b_j)^2\right) \nonumber\\
        &=
        A^\dagger A \\
        &= (\partial_\theta\sqrt{\gamma_\theta})^2\,L^\dagger L,
    \end{align}
    and therefore
    \begin{equation}
        \label{eq:local_eval}
        \expval{\widetilde g_j(\theta)^2}
        =
        \Delta t\,(\partial_\theta\sqrt{\gamma_\theta})^2\,\Tr\!\big(L^\dagger L\,\rho_{j-1}(\theta)\big)
        +\order{\Delta t^2}.
    \end{equation}
    
    \paragraph*{QFI flow.}
    For the purified (collisional) model, the global state is pure and depends on $\theta$ only through the joint unitary, hence (Lemma~\ref{lem:unitary-qfi})
    \begin{equation}
        F_Q(t)=4\,\Var\!\big(G_n(\theta)\big).
    \end{equation}
    Combining this with \eqref{eq:var_add} and \eqref{eq:local_eval} gives
    \begin{equation}
        \label{eq:QFI_Riemann}
        F_Q(t) = 4(\partial_\theta\sqrt{\gamma_\theta})^2\sum_{j=1}^n \Delta t\,\Tr(L^\dagger L\,\rho_{j-1}(\theta))
        + O(t\,\Delta t).
    \end{equation}
    Note that the first term is $\order{t}$. Taking the limit $\Delta t\to 0$ (with $t=n\Delta t$ fixed) thus yields
    \begin{equation}
    \label{eq:QFI_integral}
        F_Q(t)
        =
        4(\partial_\theta\sqrt{\gamma_\theta})^2\int_0^t \dd{s}\,\Tr(L^\dagger L\,\rho(\theta;s)).
    \end{equation}
    By the fundamental theorem of calculus, this exhibits the QFI as an integral of an instantaneous flow,
    \begin{equation}
        F_Q(t)=\int_0^t \dd{s}\,\dot F_Q(s), \qquad \dot F_Q(s)\triangleq 4(\partial_\theta\sqrt{\gamma_\theta})^2\,\Tr\!\big(L^\dagger L\,\rho(\theta;s)\big),
    \end{equation}
    and using $\Tr(L^\dagger L\,\rho(\theta;s))\le \norm{L^\dagger L}=\norm{L}^2$ we obtain
    \begin{equation}
        F_Q(t)\le 4t\,(\partial_\theta\sqrt{\gamma_\theta})^2\,\norm{L}^2,
    \end{equation}
    which concludes the proof.
\end{proof}

%==========
\subsection{Proof of Theorem~\ref{thm:qfi-matrix}}
\label{app:thm-collision-qfi}

In this Appendix, we prove the explicit form of the collisional QFI matrix flow stated in Theorem~\ref{thm:qfi-matrix} of the main text, which we restate here for convenience.

\QFIM*

In fact, we obtain the QFI matrix result from a more general flow identity for the (pure-state) quantum geometric tensor (QGT) of the collisional purification (Theorem~\ref{thm:qgt-flow} below),
\begin{equation}\label{eq:qgt-def-pure}
  \mathcal{Q}_{\alpha\beta}
  =\ip{\partial_\alpha \Psi}{\partial_\beta \Psi}
  -\ip{\partial_\alpha\Psi}{\Psi}\!\ip{\Psi}{\partial_\beta\Psi},
\end{equation}
which decomposes as
\begin{equation}
  \mathcal{Q} = \frac{1}{4}F_Q - \frac{i}{2}\Omega,
\end{equation}
so that $F_Q=4\Re\,\mathcal{Q}$ is the QFI matrix while $\Omega=-2\Im\,\mathcal{Q}$ is the (Uhlmann) curvature governing measurement incompatibility~\cite{Ragy2016compatibility, Pezze2017:Multiparam, LiuYuanLuWang2020} (see also Refs.~\cite{Provost1980, Matsumoto2002, Carollo2018:UhlmannCurvature}).

We first prove two lemmas.
\begin{lem}[Rate Equation]
    \label{lemma:rate-eq}
    Given $\Gamma= V\Lambda V^\dagger$ and $K_\alpha = (\partial_\alpha V) V^\dagger$, the fundamental rate equation reads
    \begin{equation}
      \partial_\alpha\Gamma =  V(\partial_\alpha\Lambda) V^\dagger + \comm{K_\alpha}{\Gamma}.
    \end{equation}
\end{lem}
\begin{proof}
The following equalities hold:
\begin{align}
  \partial_\alpha\Gamma &= \partial_\alpha( V\Lambda  V^\dagger) \\
  &= (\partial_\alpha  V)\Lambda  V^\dagger +  V(\partial_\alpha\Lambda) V^\dagger +  V \Lambda (\partial_\alpha  V^\dagger). \label{eq:lem3-step1}
\end{align}
Recall the connection in the canonical basis $K'_\alpha =  V^\dagger (\partial_\alpha  V)$, which implies $K_\alpha =  VK'_\alpha V^\dagger = (\partial_\alpha  V)  V^\dagger$ in the local basis and also $K_\alpha^\dagger  =  V(\partial_\alpha  V^\dagger)=-K_\alpha$. Therefore:
\begin{align}
  (\partial_\alpha  V)\Lambda  V^\dagger +   V \Lambda (\partial_\alpha  V^\dagger) &= (\partial_\alpha  V) V^\dagger ( V\Lambda  V^\dagger) +  ( V \Lambda V^\dagger)  V(\partial_\alpha  V^\dagger) \\
  &= K_\alpha \Gamma -\Gamma K_\alpha \\
  &= \comm{K_\alpha}{\Gamma}.
\end{align}
Substituting this into Eq.~\eqref{eq:lem3-step1} and rearranging terms concludes the proof.
\end{proof}

\begin{lem}[$\Gamma$-Frame Connection]
    \label{lemma:connection}
    Let $\Gamma= V\Lambda  V^\dagger$, such that $J_r(\bm{\theta})=\sum_a  V_{ar} L_a$ defines the canonical jump operators. Then:
    \begin{equation}
      \label{eq:berry-connection}
      \partial_\alpha J_r = \sum_q (K'_\alpha)_{qr} J_q, \qquad K'_\alpha =  V^\dagger (\partial_\alpha V),
    \end{equation}
    where $K'_\alpha$ denotes the $\Gamma$-frame connection, which is a non-Abelian Berry (or Wilczek-Zee) connection such that $(K'_\alpha)_{ij}=\ip{\lambda_i}{\partial_\alpha \lambda_j}$.
\end{lem}
\begin{proof}
    The relation $J_r(\bm{\theta})=\sum_a  V_{ar} L_a$ follows from substituting the eigendecomposition $\Gamma= V\Lambda V^\dagger$ into the Lindbladian. Using this relation, we compute
    \begin{align}
      \partial_\alpha J_r &= \sum_a (\partial_\alpha V)_{ar} L_a \\
      &= \sum_a ( V K'_\alpha)_{ar} L_a \\
      &= \sum_a\sum_q  V_{aq} (K'_\alpha)_{qr} L_a \\
      &= \sum_q (K'_\alpha)_{qr} \left(\sum_a  V_{aq} L_a \right)\\
      &= \sum_q (K'_\alpha)_{qr} J_q,
    \end{align}
    where $K'_\alpha =  V^\dagger(\partial_\alpha V)$ denotes the eigenframe connection in the canonical basis.
\end{proof}

\begin{thm}[Collisional QGT Flow]
\label{thm:qgt-flow}
    
    Consider the Lindbladian
    \[
      \mathcal L_{\bm{\theta}}(\rho)= -i\comm{H_c(t)}{\rho} + \mathscr{D}_{\bm{\theta}}(\rho),\qquad  \mathscr{D}_{\bm{\theta}}(\rho)=\sum_{i,j}\Gamma_{ij}(\bm{\theta})
      \left(L_i\rho L_j^\dagger-\tfrac12\{L_j^\dagger L_i,\rho\}\right),
    \]
    with (time-dependent) control Hamiltonian $H_c$, fixed local jumps $L_i$ and a differentiable (time-independent) dissipation matrix $\Gamma(\bm{\theta}) \geq 0$. Choose the diagonalization $\Gamma=V\Lambda V^\dagger$ and define the $\Gamma$-frame connection $K_\alpha \triangleq (\partial_\alpha V)V^\dagger$. In the local jump=operator basis, the collisional Quantum Geometric Tensor (QGT) flow then reads
    \begin{equation}
      \dot{\mathcal Q}_{\alpha\beta}(\bm{\theta};t) =\frac{1}{4}\Tr(\mathcal{K}_{\alpha\beta}(\bm{\theta})\,\mathcal{C}(t)),
    \end{equation}
    where $\mathcal{C}_{ij}(t)=\Tr(\rho(t)L_i^\dagger L_j)$,
    $\rho(t)=\mathcal{T}e^{\int_{0}^t\dd{s}\mathcal L_{\bm{\theta}}}(\psi)$,
    \begin{align}
        \mathcal{K}_{\alpha\beta}
        &\triangleq 4 \big(\partial_\beta \zeta\big)\big(\partial_\alpha \zeta\big)^\dagger \\
        &= \big(\partial_\beta \Gamma + \acomm{K_\beta}{\Gamma}\big)\,
        \Gamma^{-1}\big(\partial_\alpha \Gamma + \acomm{K_\alpha}{\Gamma}\big)^\dagger,
    \end{align}
    and $\zeta= V\Lambda^{1/2}$ such that $\Gamma=\zeta \zeta^\dagger$.
\end{thm}

\begin{proof}
    Diagonalize the dissipation matrix as
    \begin{equation}
      \Gamma(\bm{\theta})
      =  V(\bm{\theta})\,\Lambda(\bm{\theta})\, V(\bm{\theta})^\dagger,
      \qquad
      \Lambda(\bm{\theta})
      = \mathrm{diag}\big(\gamma_1(\bm{\theta}),\dots,\gamma_R(\bm{\theta})\big),
    \end{equation}
    with canonical rates $\gamma_r >0$ and unitary $V$ acting on the jump-index space. In terms of the canonical jump operators,
    \begin{equation}
      J_r(\bm{\theta})
      \triangleq \sum_a  V_{ar}(\bm{\theta})\,L_a,
      \qquad r=1,\dots,R,
    \end{equation}
    the dissipative part of the Lindbladian diagonalizes:
    \begin{equation}
      \mathscr{D}_{\bm{\theta}}(\rho)
      = \sum_{r} \gamma_r(\bm{\theta})\,
      \Big(J_r\rho J_r^\dagger
      - \tfrac{1}{2}\{J_r^\dagger J_r,\rho\}\Big).
    \end{equation}
    
    We now go to the collisional purification picture. The canonical Hamiltonian reads
    \begin{equation}
      H(\bm{\theta}; t)
      = \sum_{r} \sqrt{\gamma_r(\bm{\theta})}
      \Big(J_r(\bm{\theta})\otimes e_r^\dagger(t)
      + J_r^\dagger(\bm{\theta})\otimes e_r(t)\Big),
    \end{equation}
    where the bath operators $e_r(t)$ are independent of $\bm{\theta}$ and satisfy (vacuum gauge):
    \begin{equation}
      \expval*{e_r(t)e_{r'}^\dagger (t')}_E=\delta_{rr'}\delta(t-t').
    \end{equation}
    For each parameter $\theta_\alpha$,
    define the generator density
    \begin{equation}
      \label{eq:galpha-density}
      g_{\alpha}(\bm{\theta}; t)
      \triangleq U(\bm{\theta}; t,0)^\dagger\,
      \partial_\alpha H(\bm{\theta}; t)\,
      U(\bm{\theta}; t,0),
    \end{equation}
    with
    \begin{align}
      \partial_\alpha H(\bm{\theta}; t) &= \sum_r \sqrt{\gamma_r}\left(\frac{\partial_\alpha\gamma_r}{2\gamma_r}\left(J_r\otimes e_r^\dagger + \text{h.c.}\right) + \left(\partial_\alpha J_r \otimes e_r^\dagger + \text{h.c.}\right)\right) \\
      &= \sum_r \sqrt{\gamma_r}\left(\frac{\partial_\alpha\gamma_r}{2\gamma_r}\left(J_r\otimes e_r^\dagger + \text{h.c.}\right) + \left(\sum_q (K_\alpha')_{qr} J_q \otimes e_r^\dagger + \text{h.c.}\right) \right)\\
      &= \frac{1}{2}\sum_{r,q} \sqrt{\gamma_r} (F_{\alpha})_{qr} J_q \otimes e_r^\dagger + \text{h.c.}, \label{eq:alphaH}
    \end{align}
    where we use Lemma~\ref{lemma:connection} for the second equality and let
    \begin{equation}
      \label{eq:Falpha}
      F_\alpha \triangleq \Lambda^{-1}\partial_\alpha \Lambda + 2K'_\alpha, \qquad (F_\alpha)_{qr} = \frac{\partial_\alpha\gamma_r}{\gamma_r}\delta_{qr} + 2(K_\alpha')_{qr}.
    \end{equation}

    For the collisional unitary process, the QGT of the purified state may be written in terms of the integrated generators
    $G_\alpha(t)=iU(\bm{\theta};t,0)^\dagger\partial_\alpha U(\bm{\theta};t,0)=\int_0^t\dd{s}\,g_\alpha(\bm{\theta};s)$ as
    $\mathcal{Q}_{\alpha\beta}(t)=\expval{ G_\alpha(t)\,G_\beta(t)}$ since $\expval{g_\alpha}=0$.
    Therefore its time derivative admits the time-bin representation
    \begin{equation}
        \label{eq:qgt_flow_bin}
        \dot{\mathcal{Q}}_{\alpha\beta}(\bm{\theta};t)
        = \partial _{t} \expval{ G_\alpha(t)\,G_\beta(t)}
        = \int_0^{t} \dd s 
        [ \expval{g_\alpha(\bm{\theta};t ) g_\beta(\bm{\theta};s)} + \expval{g_\alpha(\bm{\theta};s ) g_\beta(\bm{\theta};t)} ] 
    \end{equation}
    Define the canonical correlator matrix
    \begin{equation}
      \mathcal{C}'_{rq}(t)
      \triangleq \Tr\big\{J_r^\dagger J_q\,\rho(t)\big\},
      \qquad
      \rho(t) = \mathcal{T}e^{\int_{0}^t \dd{s}\mathcal{L}_{\bm{\theta}, t}}(\psi).
    \end{equation}
    Substitute Eq.~\eqref{eq:alphaH} into Eq.~\eqref{eq:galpha-density} and use the vacuum gauge $\expval*{e_r(t)e_{r'}^\dagger (t')}_E=\delta_{rr'}\delta(t-t')$, with all other correlators vanishing, to deduce
    \begin{align}
      \expval{g_\alpha(\bm{\theta};s)\,g_\beta(\bm{\theta};s') } &=
      \frac{1}{4} \delta(s-s')\sum_{q,q'} \sum_r \gamma_r (F_\beta)_{q'r} (F_\alpha^\dagger)_{rq} \mathcal{C}'_{qq'}(s)\\
      &= \frac{1}{4} \delta(s-s') \sum_{q,q'} (\mathcal{K}_{\alpha\beta}')_{q'q}\mathcal{C}'_{qq'}(s) \\
      &= \frac{1}{4}\delta(s-s') \Tr(\mathcal{K}'_{\alpha\beta}\, \mathcal{C}'(s)),
    \end{align}
    with
    \begin{equation}
      \label{eq:DK-comp-canonical}
      \mathcal{K}'_{\alpha\beta}(\bm{\theta})
      = \,F_\beta(\bm{\theta})\,
      \Lambda(\bm{\theta})\,
      F_\alpha^\dagger(\bm{\theta}).
    \end{equation}
    Since $J_r =\sum_a  V_{ar}L_a$, the canonical and local correlator matrices relate by
    \begin{equation}
      \mathcal{C}'(t)
      =  V^\dagger\,\mathcal{C}(t)\, V,
    \end{equation}
    where $\mathcal{C}_{ab}(t)
    = \Tr(L_a^\dagger L_b\,\rho(t))$ denotes the correlator matrix in the local basis. By cyclicity of the trace,
    \begin{equation}
      \Tr(\mathcal{K}'_{\alpha\beta}\, \mathcal{C}'(t))
      = \Tr(
        \mathcal{K}_{\alpha\beta}\,
      \mathcal{C}(t)), \qquad \mathcal{K}_{\alpha\beta}(\bm{\theta})
      \triangleq  V(\bm{\theta})\,
      \mathcal{K}'_{\alpha\beta}(\bm{\theta})\,
      V^\dagger(\bm{\theta}).
    \end{equation}
    We expand $\mathcal{K}_{\alpha\beta}$ [Eq.~\eqref{eq:DK-comp-canonical}] in a more convenient split:
    \begin{align}
      \mathcal{K}_{\alpha\beta}&= ( V F_\beta  V^\dagger)( V \Lambda  V^\dagger)( VF_\alpha^\dagger  V^\dagger) \\
      &= \Big(( V F_\beta  V^\dagger)\Gamma\Big) \Gamma^{-1}\Big(\Gamma( VF_\alpha^\dagger  V^\dagger)\Big).
    \end{align}
    Use Eq.~\eqref{eq:Falpha} together with $K_\beta= VK'_\beta  V^\dagger$,  $\Gamma= V\Lambda V^\dagger$, and Lemma~\ref{lemma:rate-eq} to compute
    \begin{align}
      ( VF_\beta  V^\dagger)\Gamma &=  V\left((\partial_\beta\Lambda)\Lambda^{-1} + 2K_\beta'\right) V^\dagger \Gamma\\
      &=  V(\partial_\beta\Lambda) V^\dagger + 2 K_\beta\Gamma \\
      &= \left(\partial_\beta\Gamma -\comm{K_\beta}{\Gamma}\right) + 2 K_\beta \Gamma\\
      &= \partial_\beta\Gamma + \acomm{K_\beta}{\Gamma}.
    \end{align}
    Therefore:
    \begin{equation}
      \mathcal{K}_{\alpha\beta}(\bm{\theta})= \left(\partial_\beta\Gamma + \acomm{K_\beta}{\Gamma}\right) \Gamma^{-1}\left(\partial_\alpha\Gamma + \acomm{K_\alpha}{\Gamma}\right)^\dagger.
    \end{equation}
    Gathering everything together, the QGT flow reads
    \begin{equation}
      \dot{\mathcal{Q}}_{\alpha\beta}(\bm{\theta};t) = \frac{1}{4}\, \Tr\left\{ \mathcal{K}_{\alpha\beta}(\bm{\theta})\,\mathcal{C}(t)\right\}.
    \end{equation}
    
    We now prove the dissipative amplitude form. Define the dissipative amplitude
    \begin{equation}
      \zeta(\bm{\theta})\triangleq  V (\bm{\theta})\,\Lambda(\bm{\theta})^{1/2},
      \qquad
      \Gamma=\zeta\zeta^\dagger.
    \end{equation}
    Work in the canonical gauge:
    \begin{equation}\label{eq:D-Lhalf}
      \mathcal{K}'_{\alpha\beta}
      =\big(F_\beta\Lambda^{1/2}\big)
      \big(F_\alpha\Lambda^{1/2}\big)^\dagger, \qquad F_\beta\Lambda^{1/2}
      =2\big[\partial_\beta(\Lambda^{1/2})+K'_\beta\Lambda^{1/2}\big].
    \end{equation}
    Now differentiate $\zeta=V\Lambda^{1/2}$:
    \begin{align}
      \partial_\beta\zeta
      &=(\partial_\beta V)\Lambda^{1/2}+V\partial_\beta(\Lambda^{1/2}) \\
      &=V\big[V^\dagger\partial_\beta V\,\Lambda^{1/2}
      +\partial_\beta(\Lambda^{1/2})\big]\\
      &=V\big[K'_\beta\Lambda^{1/2}+\partial_\beta(\Lambda^{1/2})\big].
    \end{align}
    Combine this with Eq.~\eqref{eq:D-Lhalf} to find
    \begin{equation}
      VF_\beta\Lambda^{1/2}
      =2\,\partial_\beta\zeta.
    \end{equation}
    Therefore,
    \begin{align}
      \mathcal{K}_{\alpha\beta}
      &=\big(VF_\beta\Lambda^{1/2}\big)
      \big(VF_\alpha\Lambda^{1/2}\big)^\dagger \\
      &=4\,(\partial_\beta\zeta)\,(\partial_\alpha\zeta)^\dagger,
    \end{align}
    which completes the proof.
\end{proof}

Taking the real part of $\dot{\mathcal{Q}}_{\alpha\beta}$ in the results above produces the QFI matrix~\eqref{eq:qfim-flow} of Theorem~\ref{thm:qfi-matrix}.

It turns out that we may tighten purification when the jump operators are Hermitian, as we now argue.

\begin{prop}[QGT Flow for Hermitian Jumps]
    \label{prop:connected-qgt}
    Suppose the local jump basis is Hermitian,
    \begin{equation}
        L_a = L_a^\dagger, \qquad a=1,\dots,R.
    \end{equation}
    Define the connected correlator
    \begin{equation}
        \mathcal C^{c}_{ab}(t)
        \triangleq
        \Tr\!\big(\Delta L_a(t)\,\Delta L_b(t)\,\rho(t)\big),
        \qquad
        \Delta L_a(t)\triangleq L_a-\Tr(L_a\rho(t))\,I.
    \end{equation}
    Then the QGT flow reads
    \begin{equation}
         \dot{\mathcal Q}^c_{\alpha\beta}(\bm{\theta};t) =\frac{1}{4}\Tr(\mathcal{K}_{\alpha\beta}(\bm{\theta})\,\mathcal{C}^c(t)),
    \end{equation}
    where $\mathcal K_{\alpha\beta}$
    is the kernel from Theorem~\ref{thm:qfi-matrix}.
\end{prop}

\begin{proof}[Proof sketch]
    The proof follows by revisiting the discrete collisional derivations of Lemma~\ref{lem:qfi-flow} (Appendix~\ref{app:lem-qfi-flow}) and Theorem~\ref{thm:qgt-flow} (Appendix~\ref{app:thm-collision-qfi}). In each time bin, the dressed generator, $\tilde{g}_{\alpha, j}$, splits into a fluctuating system part plus an system-identity term that is proportional to the instantaneous means $\Tr(L_a\rho(t_j))\in\mathbb{R}$ when $L_a$ are Hermitian. This system-identity term is exactly a Hermitian displacement generator on the bath element. Thus, we may choose a product of bath displacements in a new environment gauge (viz., $U_E$) that cancels this term bin by bin.\footnote{This serves as a new \emph{displaced} collisional purification.} The remaining generator depends only on the fluctuations $\Delta L_a(t_j)$, and repeating the proof of Theorem~\ref{thm:qgt-flow} in this new displaced gauge replaces the disconnected correlator $\mathcal C(t)$ by the connected correlator $\mathcal C^{c}(t)$.
\end{proof}

\subparagraph*{Remarks.}
The fact that Proposition~\ref{prop:connected-qgt} represents a tighter purification through the connected correlator is special to Hermitian jumps. If $L_a=L_a^\dagger$, then the mean $\Tr(L_a\rho(t))$ is real and contributes only a bath displacement, which may be removed by an environment gauge. After this subtraction, the collisional QGT flow depends only on the fluctuations $\Delta L_a$. Contrariwise, for non-Hermitian jumps, the disconnected part is not merely a removable bath displacement. Hence subtracting first moments does not suffice to reduce the all-time bound to the connected correlator $\mathcal C^c(t)$. Indeed, this replacement is false in general: Estimating bosonic loss, with jump operator $L=a$, with a coherent-state probe provides a counterexample, since the centered correlator vanishes whereas the QFI is nonzero for all $t$.

\subsection{Proof of Theorem~\ref{thm:gen-Heisenberg}}
\label{app:thm-gen-heisenberg}

In this Appendix, we prove the generalized Heisenberg limit, Theorem~\ref{thm:gen-Heisenberg} of the main text, which we restate here for convenience.

%average QFI per parameter satisfies

\GHL*

\begin{proof}
    To begin, consider a single shot of duration $t$. Take the parameter trace of~\eqref{eq:qfim-flow} (by Definition~\ref{def:param-avg}) to obtain the average QFI flow for $s\in[0,t]$:
    \begin{equation}
        \dot{\bar{F}}_Q(s)
        = \Re\Tr(\bar{\mathcal{K}} \mathcal{C}(s)).
    \end{equation}
    Integrate from $0$ to $t$ and use $\abs{\Re z}\le \abs{z}$ to obtain
    \begin{align}
        \bar{F}_Q(t)
        &=\int_0^t\dd{s}\,\Re\,\Tr(\bar{\mathcal{K}} \mathcal{C}(s))\nonumber\\
        &\le t\,\sup_{s\in[0,t]} \abs{\Tr(\bar{\mathcal{K}} \mathcal{C}(s))}.
        \label{eq:avg-qfi-int}
    \end{align}
    For fixed $s$, expand entrywise and bound terms by the max-norm (Definition~\ref{def:max-norms}):
    \begin{align}
        \abs{\Tr(\bar{\mathcal{K}} \mathcal{C}(s))}
        &= \abs{\sum_{a,b=1}^R \bar{\mathcal{K}}_{ba}\,\mathcal{C}_{ab}(s)} \nonumber\\
        &\le \sum_{a,b=1}^R \abs{\bar{\mathcal{K}}_{ba}}\,\abs{\mathcal{C}_{ab}(s)} \nonumber\\
        &\le \norm{\bar{\mathcal{K}}}_{\max}\,\norm{\mathcal{C}(s)}_{\max}\,
        \abs{\mathsf{I}(s)},
        \label{eq:KC-supp}
    \end{align}
    where we define the set $\mathsf{I}(s)\triangleq {\rm supp}(\bar{\mathcal{K}})\cap{\rm supp}(\mathcal{C}(s))$ and $\abs{\mathsf{I}(s)}$ its cardinality. Moreover, since $\mathcal{C}_{ab}(s)=\Tr(L_a^\dagger L_b\,\rho(s))$ with $\Tr\rho(s)=1$,
    \begin{equation}
        \abs{\mathcal{C}_{ab}(s)}
        =\abs{\Tr(L_a^\dagger L_b\,\rho(s))}
        \le \norm{L}_{\max}^2,
    \end{equation}
    and thus $\norm{\mathcal{C}(s)}_{\max}\le \norm{L}_{\max}^2$ (see Definition~\ref{def:max-norms} for a definition of $\norm{L}_{\max}$).
    Finally, for each $a\in[R]$,
    \begin{equation}
        \abs{\{b\in[R]:(a,b)\in\mathsf{I}(s)\}}
        \le \min(r_{\bar{\mathcal{K}}},\,r_{\mathcal{C}(s)}),
    \end{equation}
    with $r_M$ the row connectivity of matrix $M$ (Definition~\ref{def:row-connect}), 
    so summing over $a$ bounds the cardinality of $\mathsf{I}(s)$
    \begin{equation}
        \abs{\mathsf{I}(s)}
        \le R\,\min(r_{\bar{\mathcal{K}}},\,r_{\mathcal{C}(s)})
        \le R\,\min(r_{\bar{\mathcal{K}}},\,r_{\mathcal{C}}),
    \end{equation}
    with $r_{\mathcal{C}}\ \triangleq\ \sup_{s\in[0,t]} r_{\mathcal{C}(s)}$. Substitute into~\eqref{eq:KC-supp} and then into~\eqref{eq:avg-qfi-int} to obtain 
    \begin{equation}
        \bar{F}_Q(t)
        \le t\,\norm{\bar{\mathcal{K}}}_{\max}\,\norm{L}_{\max}^2\,
        R\,\min(r_{\bar{\mathcal{K}}},\,r_{\mathcal{C}}).
    \end{equation}
    Additivity of QFI under independent repetitions implies $\bar{F}_Q(T)=\nu\,\bar{F}_Q(t)$, which produces~\eqref{eq:avg-qfi-bound} since $T=\nu t$.
\end{proof}

%==========
\subsection{Proof of Theorem~\ref{thm:rpm-qcrb}}
\label{app:thm-rpm-qcrb}

In this Appendix, we prove that the RPM protocol achieves the eigenrate QCRB, Theorem~\ref{thm:rpm-qcrb} of the main text, which we restate here for convenience.

\RPM*

\begin{proof}
    In the eigenrate-only setting, Proposition~\ref{prop:rate-qfi} gives
    \begin{equation}
        (\dot F_Q)_{\alpha\alpha}(t)
        =
        \sum_{k=1}^R \frac{(\partial_\alpha\gamma_k)^2}{\gamma_k}\,\Tr(J_k^\dagger J_k\,\rho(t)).
    \end{equation}
    Taking the per-parameter average [via Eq.~\eqref{eq:param-avg}] yields
    \begin{equation}
        \dot{\bar{F}}_Q(t)=\Tr(A\,\rho(t)),
    \end{equation}
    with
    \begin{equation}
        A\ \triangleq\ \sum_{k=1}^R c_k\,J_k^\dagger J_k, \quad
        c_k \triangleq \frac{1}{d}\sum_{\alpha=1}^d \frac{(\partial_\alpha\gamma_k)^2}{\gamma_k}.
    \end{equation}
    Hence, for any trajectory $\rho(s)$ and any $t>0$,
    \begin{equation}
        \frac{\bar{F}_Q(t)}{t}
        =\frac{1}{t}\int_0^t\dd{s}\,\Tr(A\rho(s))
        \le \sup_{\rho}\Tr(A\rho)=\lambda_{\max}(A).
    \end{equation}
    Conversely,
    $\lim_{t\to 0^+}\bar{F}_Q(t;\psi)/t=\Tr(A\psi)$ for any pure $\psi$, so optimizing gives $\mathcal{R}^\star=\sup_{\psi}\Tr(A\psi)=\lambda_{\max}(A)$.
    
    If $\psi^\star$ achieves $\lambda_{\max}(A)$ and admits a distinguishable jump POVM $\bm{\mathcal{J}}$~\eqref{eq:jump-povm}, then
    Lemma~\ref{lem:rpm-fisher} implies
    \begin{equation}
        \lim_{t\to 0^+}\frac{\bar{F}_{\rm rpm}(t;\psi^\star)}{t}
        =
        \Tr(A\psi^\star)
        =
        \lambda_{\max}(A)
        =
        \mathcal{R}^\star,
    \end{equation}
    proving that RPM attains the average optimal per-parameter rate. Scenario~1 (unlimited shots~\eqref{eq:F-unlim-shots}) implies
    $\bar{F}_Q^{\rm opt}(T)=T\,\mathcal{R}^\star$, hence RPM saturates the eigenrate QCRB.
\end{proof}

%==========
%==========
\section{Uhlmann Extremality}
\label{app:uhlmann-extremality}

In this Appendix, we analyze the gauge freedom in purifications of the reduced state $\rho(\bm{\theta};t)$ and characterize the Uhlmann-extremal (tight) purification. In the main text we primarily work with the canonical collisional purification (Appendix~\ref{app:collisional}), which is convenient and useful for obtaining upper bounds, but is not generically tight for any given $\psi$. To assess tightness formally, we adopt a geometric viewpoint, using the Bures metric $g^{\rm Bures}_{uv}$ on mixed states and the Fubini--Study metric $g^{\rm FS}_{uv}$ on purifications, where the Bures metric is directly related to the QFI via~\cite{Braunstein1994:Bures}
\begin{equation}
  g_{uv}^{\rm Bures}\big[\rho(\bm{\theta};t)\big]
  =\tfrac{1}{4}(F_Q)_{uv}\big[\rho(\bm{\theta};t)\big].
\end{equation}

Fix a time $t$ and a parameter value $\bm{\theta}$. The canonical collisional purification reads
\begin{equation}
\ket{\Psi(\bm{\theta};t)}
=
U(\bm{\theta};t)\,\ket{\psi}\otimes\ket{\varphi_{\rm in}},
\end{equation}
where $\ket{\psi}$ is a pure system probe state and the environment is initialized in the bosonic vacuum
$\ket{\varphi_{\rm in}}=\ket{0}$. The corresponding reduced output states are
\begin{equation}
\rho(\bm{\theta};t)=\Tr_{\rm env}\dyad{\Psi(\bm{\theta};t)}{\Psi(\bm{\theta};t)},
\qquad
\varphi_{\rm out}(\bm{\theta};t)=\Tr_{\rm sys}\dyad{\Psi(\bm{\theta};t)}{\Psi(\bm{\theta};t)}.
\end{equation}

Any other purification of $\rho(\bm{\theta};t)$ (on the same environment Hilbert space) is obtained by a $\bm{\theta}$-dependent environment unitary,
\begin{equation}
\ket{\Psi(\bm{\theta};t)}
\longmapsto
\ket{\Psi_{U_E}(\bm{\theta};t)}
\triangleq
\big(I\otimes U_E(\bm{\theta})\big)\ket{\Psi(\bm{\theta};t)} ,
\end{equation}
which leaves $\rho(\bm{\theta};t)$ invariant. By Uhlmann's theorem~\cite{Uhlmann1976}, a purification variational \emph{Ansatz} admits the Bures metric. That is, for any tangent vectors $u,v$ at $\bm{\theta}$,
\begin{equation}
g^{\rm Bures}_{uv}\big[\rho(\bm{\theta};t)\big]
=
\min_{U_E}\, g^{\rm FS}_{uv}\big[\Psi_{U_E}(\bm{\theta};t)\big],
\end{equation}
where the minimum ranges over $\bm{\theta}$-dependent environment unitaries $U_E(\bm{\theta})$. Let $U_E^\star(\bm{\theta})$ denote an extremal gauge and define the corresponding extremal purification
\begin{equation}
\ket{\Psi^\star(\bm{\theta};t)}\;\triangleq\;
\big(I\otimes U_E^\star(\bm{\theta})\big)\ket{\Psi(\bm{\theta};t)}.
\end{equation}
Then the extremal Fubini--Study metric coincides with the Bures metric of the reduced state,
\begin{equation}
g^\star_{\alpha\beta}(\bm{\theta};t)
\;\triangleq\;
g^{\rm FS}_{\alpha\beta}\big[\Psi^\star(\bm{\theta};t)\big]
\;=\;
g^{\rm Bures}_{\alpha\beta}\big[\rho(\bm{\theta};t)\big].
\end{equation}

Below we characterize Uhlmann extremality through independent Lyapunov equations satisfied by the extremal environment generators $h^\star_\alpha\triangleq i{U_E^\star}^\dagger (\partial_\alpha U_E^\star)$. The canonical collisional purification is Uhlmann-extremal iff $U_E^\star(\bm{\theta})=I_E$, equivalently $h^\star_\alpha=0$ for all $\alpha$. (This extremization is similar in spirit to those used in Refs.~\cite{Fujiwara2008:fibre, Dobrza2012ElusiveHeisenberg, Sekatski2017:FullFast, Demkowicz2017:Adaptive, Zhou2018:HNLS, Zhou2021:Asymptotic, Albarelli2022:Incompatibility, Gorecki2023:CausalAdaptive} to obtain asymptotic precision bounds.)

Throughout we use the shorthand $\Psi \triangleq \dyad{\Psi}$ to denote the density matrix of the global pure state.

%==========
\subsection{Extremality Condition}

Let $\partial_\alpha$ denote differentiation with respect to $\theta_\alpha\in\bm{\theta}$, and define a real direction $u$ with $\partial_u \triangleq \sum_\alpha u_\alpha \partial_\alpha$. Introduce the Hermitian generators for the collisional unitary and the environment gauge,\footnote{Here, $\tilde{G}_\alpha$ relates to $G_\alpha$ of Appendix~\ref{app:thm-collision-qfi} by $G_\alpha= U \widetilde{G}_\alpha U^\dagger$.}
\begin{equation}\label{appeq:generators}
\widetilde{G}_\alpha(\bm{\theta};t)\triangleq i\,(\partial_\alpha U)U^\dagger,
\qquad
h_\alpha(\bm{\theta})\triangleq i\,U_E^\dagger(\bm{\theta})\,(\partial_\alpha U_E(\bm{\theta})),
\end{equation}
and the directional generators $\widetilde{G}_u=\sum_\alpha u_\alpha \widetilde{G}_\alpha$ and $h_u=\sum_\alpha u_\alpha h_\alpha$.
Then
\begin{equation}
\partial_u\ket{\Psi(\bm{\theta};t)} = -i\,\widetilde{G}_u(\bm{\theta};t)\,\ket{\Psi(\bm{\theta};t)}.
\end{equation}
For the gauged purification $\ket{\Psi_{U_E}}=(I\otimes U_E)\ket{\Psi}$,
\begin{equation}
\partial_u\ket{\Psi_{U_E}}
=-i(I\otimes U_E)\left(\widetilde{G}_u+I\otimes h_u\right)\ket{\Psi}.
\end{equation}
Since the Fubini-Study metric is invariant under $(I\otimes U_E)$, the gauge optimization is entirely captured by
the shift $\widetilde{G}_u\mapsto \widetilde{G}_u+I\otimes h_u$ acting on the canonical state $\ket{\Psi}$.

For a smooth parametrization $\Psi_{U_E}(\bm{\theta};t)$, the directional Fubini-Study line element is
\begin{equation}
    \ell_u^2
    = \ip{\partial_u\Psi_{U_E}}{\partial_u\Psi_{U_E}}
    -\abs{\ip{\Psi_{U_E}}{\partial_u\Psi_{U_E}}}^2.
\end{equation}
Using $\partial_u\ket{\Psi_{U_E}}=-i(I\otimes U_E)(\widetilde{G}_u+I\otimes h_u)\ket{\Psi}$,
this becomes the variance of the shifted generator $\widetilde{G}_u+I\otimes h_u$ with respect to the canonical state,
\begin{equation}
\ell_u^2[h_u] = \Var_{\Psi}\!\big(\widetilde{G}_u+I\otimes h_u\big).
\label{eq:FS-variance}
\end{equation}
Expanding gives
\begin{equation}
\Var_{\Psi}(\widetilde{G}_u+I\otimes h_u)
=\Var_{\Psi}(\widetilde{G}_u)
+\Var_{\varphi_{\rm out}}(h_u)
+2\,\Re\Tr\big(\Psi\,\Delta \widetilde{G}_u\,(I\otimes \Delta h_u)\big), 
\end{equation}
where
\begin{equation}
\Delta h_u = h_u - \Tr(\varphi_{\rm out} h_u)\,I_E,
\qquad
\Delta \widetilde{G}_u = \widetilde{G}_u - \Tr(\Psi \widetilde{G}_u)\,I_{SE}. 
\end{equation}
The cross-term reduces to an environment expectation,
\begin{equation}
\Tr\big(\Psi\,\Delta \widetilde{G}_u\,(I\otimes \Delta h_u)\big)
=\Tr_{\rm env}\!\Big(B_u\,\Delta h_u\Big), \qquad B_u \triangleq \Tr_{\rm sys}(\Psi\,\Delta \widetilde{G}_u). \label{eq:Bu}
\end{equation}
Because $\Delta h_u$ is Hermitian, only the Hermitian part $(B_u)_H=(B_u+B_u^\dagger)/2$ contributes.
Thus we obtain the quadratic functional
\begin{equation}
\ell_u^2[h_u]
=
\Var_{\Psi}(\widetilde{G}_u)
+\Tr\big(\varphi_{\rm out}\,(\Delta h_u)^2\big)
+2\,\Tr\big((B_u)_H\,\Delta h_u\big).
\label{eq:quadratic-functional}
\end{equation}
For each direction $u$, minimizing the Fubini-Study line element over $U_E$ is equivalent to minimizing
Eq.~\eqref{eq:quadratic-functional} over $\Delta h_u$.

\begin{thm}[Lyapunov-Uhlmann Extremality]
    \label{thm:lyapunov}
    Given a purification $\ket{\Psi_{U_E}}=(I\otimes U_E)\ket{\Psi}$, with environment gauge $U_E$ and local generator $h_u=i U_E^\dagger(\partial_u U_E)$, the gauge is Uhlmann extremal ($U_E=U^\star_E$) iff the centered generator $\Delta h_u^\star$ obeys the Lyapunov equation
    \begin{equation}\label{eq:uhlmann-lyapunov}
      \{\varphi_{\rm out},\Delta h_u^\star\} = -2(B_u)_H,
    \end{equation}
    with $B_u$ from Eq.~\eqref{eq:Bu}. The solution is unique up for each $u$ to $h^\star_u \mapsto h^\star_u+c I$.
\end{thm}

\begin{proof}
Consider a Hermitian, centered variation $\Delta h_u\mapsto \Delta h_u+\varepsilon X$ with $\Tr(\varphi_{\rm out} X)=0$. Differentiating Eq.~\eqref{eq:quadratic-functional} at $\varepsilon=0$ gives
\begin{align}
  \lim_{\epsilon\to 0}\frac{\ell^2_u[h_u+\epsilon X]-\ell^2_u[h_u]}{\epsilon}
  &= \Tr\Big(\varphi_{\rm out}(\Delta h_u X + X\Delta h_u)\Big) + 2\,\Tr\big((B_u)_H X\big) \\
  &= \Tr\Big(\big(\{\varphi_{\rm out},\Delta h_u\}+2(B_u)_H\big)\,X\Big).
\end{align}
On $\mathrm{supp}(\varphi_{\rm out})$, extremality for all $X$ implies
\begin{equation}
  \{\varphi_{\rm out},\Delta h_u^\star\} = -2(B_u)_H.
  \label{eq:sylvester}
\end{equation}
Conversely, if $\Delta h_u^\star$ satisfies \eqref{eq:sylvester}, then completing the square in Eq.~\eqref{eq:quadratic-functional} yields
\begin{equation*}
  \ell^2_u[h_u] = \ell^2_u[h_u^\star] +\Tr\Big(\varphi_{\rm out}(\Delta h_u-\Delta h_u^\star)^2\Big)
\end{equation*}
so $\ell^2_u[h_u]\ge \ell^2_u[h_u^\star]$ with equality iff $\Delta h_u=\Delta h_u^\star$ on ${\rm supp}(\varphi_{\rm out})$. Finally, since $X\mapsto\{\varphi_{\rm out},X\}$ is positive and invertible on ${\rm supp}(\varphi_{\rm out})$ and the variance functional~\eqref{eq:quadratic-functional} is strictly convex, the solution $\Delta h_u^\star$ is unique and the global minimizer.
\end{proof}

\subparagraph*{Remarks.}
Extremality of the canonical gauge, such that $U_E^\star=I$ and $\Delta h^\star_u =0$, implies $(B_u)_H = 0$ on $\mathrm{supp}(\varphi_{\rm out})$. Therefore, the canonical purification is Uhlmann-extremal in direction $u$ iff the Hermitian part of $B_u$ vanishes on $\mathrm{supp}(\varphi_{\rm out})$.

Let $\varphi_{\rm out}=\sum_m p_m \dyad{m}$ on its support. Taking matrix elements of Eq.~\eqref{eq:sylvester} yields, for $m,n$ in the support,
\begin{equation}
  (\Delta h_u^\star)_{mn} = -\frac{2\big((B_u)_H\big)_{mn}}{p_m+p_n},
  \qquad (m\neq n),
\end{equation}
while the diagonal is fixed up to an arbitrary scalar shift. This makes manifest that $\Delta h_u^\star=0$ holds iff $(B_u)_H$ has no off-diagonal components relative to $\varphi_{\rm out}$.

%==========
\subsection{Short-Time Extremality}
\label{app:short-time-extremality}
Here we analyze extremality in the short-time limit ($\delta t \to 0^+$), which is relevant to RPM protocols and to optimal precision limits more broadly. 

Consider the first collision within an arbitrarily small interval $\delta t$. In this limit, the input-output bath elements are confined to the vacuum and single-excitation sectors, $\mathscr{H}_E\simeq \mathscr{H}_0 \oplus \mathscr{H}_1$, with $\mathscr{H}_0={\rm span}(\ket{0})$ and
$\mathscr{H}_1={\rm span}(\{\ket{1_r}\}_{r=1}^R)$. The initial joint state is
\begin{equation}
\Psi(0)=\dyad{\psi}\otimes\dyad{0}{0},
\end{equation}
where the initial probe $\psi$ has mean and two-point correlator
\begin{equation}
\label{eq:init-moments}
m_r = \expval{J_r}{\psi},
\qquad
\mathcal{C}'_{rs} = \expval{J_r^\dagger J_s}{\psi},
\end{equation}
in the canonical jump basis.

We layout the perturbative expansion. The canonical purification admits the standard short-time expansion
\begin{equation}
\ket{\Psi(\delta t)}
=
\Big(I-\tfrac{\delta t}{2}\sum_r \gamma_r J_r^\dagger J_r\Big)\ket{\psi}\otimes\ket{0}
\;-\; i\sqrt{\delta t}\sum_r \sqrt{\gamma_r}\,\big(J_r\ket{\psi}\big)\otimes\ket{1_r}
\;+\;\order{\delta t^{3/2}},
\label{eq:Psi-firstbin}
\end{equation}
where $\ket{1_r} \triangleq  \,e_r^\dagger\ket{0}$.
The reduced state of the first collisional bath element is
\begin{equation}
\varphi_{\rm out}(\delta t)
=
\begin{pmatrix}
  1-\delta t\,\Tr(\mathfrak{C}) & i\sqrt{\delta t}\,\mathfrak{m}^\dagger\\[2pt]
  -i\sqrt{\delta t}\,\mathfrak{m} & \delta t\,\mathfrak{C}
\end{pmatrix}
+\order{\delta t^{3/2}},
\label{eq:rhoE-firstbin}
\end{equation}
in the basis $\mathscr{H}_0\oplus\mathscr{H}_1$, and we define the dressed mean and correlator as
\begin{equation}
\mathfrak{m}\triangleq \sqrt{\Lambda}\,m,
\qquad
\mathfrak{C}\triangleq \sqrt{\Lambda}\,\mathcal{C}'^{\top}\sqrt{\Lambda}.
\label{eq:cD-def}
\end{equation}
It is convenient to introduce the (dressed) connected correlator,
\begin{equation}
\mathfrak{C}_c\triangleq \mathfrak{C}-\mathfrak{m}\mathfrak{m}^\dagger
=\sqrt{\Lambda}\mathcal{C}_c'^{\top}\sqrt{\Lambda}, \qquad \mathcal{C}'_c \triangleq \mathcal{C}'-m^* m^\top
    \label{eq:dressed-moments}
\end{equation}

Fix a direction $u=\sum_\alpha u_\alpha \partial_\alpha$ and define $F_u\triangleq \sum_\alpha u_\alpha F_\alpha$, with $F_\alpha$ given by Eq.~\eqref{eq:Falpha}. In the short-time limit, the generator~\eqref{appeq:generators} reads
\begin{equation}
\widetilde{G}_u^{(\delta t)} =
\frac{\sqrt{\delta t}}{2}\left(\sum_{q,r}\sqrt{\gamma_r}
(F_u)_{qr}J_q\otimes\ket{1_r}\!\bra{0}+\text{h.c.}\right)
+\order{\delta t}.
\label{eq:Gu-firstbin}
\end{equation}
Given the canonical purification $\Psi(\delta t)=U(\delta t)\Psi(0)U^\dagger(\delta t)$, let
\begin{equation}
B_u = \Tr_{\rm sys}\left(\Psi(\delta t)\,\Delta \widetilde{G}_u^{(\delta t)}\right),
\qquad
\Delta \widetilde{G}_u^{(\delta t)} = \widetilde{G}_u^{(\delta t)}-\Tr(\Psi\widetilde{G}_u^{(\delta t)})\,I,
\end{equation}
similar to~\eqref{eq:Bu}. A direct evaluation using \eqref{eq:Psi-firstbin}--\eqref{eq:Gu-firstbin} gives
\begin{equation}
B_u = \frac{\sqrt{\delta t}}{2}\dyad{0}{w_u} +\delta t \sum_{r,s}\left(M_u\right)_{sr}\dyad{1_s}{1_r} - \delta t\,\beta \dyad{0}
+ \order{\delta t^{3/2}},
\label{eq:Bu-expand}
\end{equation}
where $\beta=\Tr(M_u)$ is an irrelevant constant,
\begin{equation}
w_u=\sqrt{\Lambda}F_u^\top m, \qquad \ket{w_u}\triangleq \sum_{r}\left(w_u\right)_r\ket{1_r},
\label{eq:w-def}
\end{equation}
and we define the mixing matrix
\begin{equation}
M_u \triangleq
-\frac{i}{2}\sqrt{\Lambda}\mathcal{C}'^\top F_u^*\sqrt{\Lambda},
\label{eq:Cu-def}
\end{equation}
with Hermitian and anti-Hermitian parts
\begin{align}
(M_u)_H &= \frac{i}{4}\sqrt{\Lambda}\left(F_u^\top\mathcal{C}'^\top - \mathcal{C}'^\top F_u^*\right)\sqrt{\Lambda}, \\
(M_u)_A &= -\frac{i}{4}\sqrt{\Lambda} \left(F_u^\top\mathcal{C}'^\top + \mathcal{C}'^\top F_u^*\right) \sqrt{\Lambda}.
\end{align}
With this perturbative construction in hand, we prove the following proposition.

\begin{prop}[Short-Time Extremal \emph{Ansatz}]
\label{prop:ansatz}
    Let $m$ and $\mathcal{C}'$ denote the mean and two-point correlator of the initial probe in the canonical basis per Eq.~\eqref{eq:init-moments} and $\mathcal{C}'_c \triangleq \mathcal{C}'-m^* m^\top$ denote the connected correlator. Choose an environment gauge $U_E(\bm{\theta})$ with generator $h_u = iU_E^\dagger\partial_u U_E$ and $\Delta h_u = h_u-\Tr(\varphi_{\rm out} h_u)I$. Then, the block \emph{Ansatz}
    \begin{equation}
      \Delta h_u^\star =
      \begin{pmatrix}
        0 & \sqrt{\delta t}\,{\xi_u^\star}^\dagger\\[2pt] \sqrt{\delta t}\,\xi_u^\star & \Xi_u^\star
      \end{pmatrix}
      + \order{\delta t}
      \label{eq:h-ansatz}
    \end{equation}
    is Uhlmann extremal, in the sense of satisfying equation~\eqref{eq:uhlmann-lyapunov} of Theorem~\ref{thm:lyapunov}, when (\textit{i}) the displacement vector $\xi_u^\star$ obeys
    \begin{equation}
        \label{eq:extremal-displace}
      \xi_u^\star = \left(i \Xi_u^\star \sqrt{\Lambda}- \frac{1}{2}\sqrt{\Lambda}F_u^\top\right)m
    \end{equation}
    and (\textit{ii}) the $R\times R$ Hermitian matrix $\Xi_u^\star$ obeys the Lyapunov equation
    \begin{equation}
        \label{eq:extremal-rotate}
      \acomm{\mathfrak{C}_{c}}{\,\Xi_u^\star} = -\frac{i}{2} \sqrt{\Lambda}\left(F_u^\top {\mathcal{C}'_{c}}^\top - {\mathcal{C}'_{c}}^\top F_u^*\right) \sqrt{\Lambda},
    \end{equation}
    where $\mathfrak{C}_{c} = \sqrt{\Lambda}\, \mathcal{C}_c'^\top \sqrt{\Lambda}$.
\end{prop}

\begin{proof}
    By Theorem~\ref{thm:lyapunov}, a gauge is  Uhlmann-extremal iff
    \begin{equation}
      \{\varphi_{\rm out}(\delta t),\Delta h_u^\star\}=-2(B_u)_H.
      \label{eq:sylvester-firstbin}
    \end{equation}
    Motivated by the scalings in \eqref{eq:rhoE-firstbin} and \eqref{eq:Bu-expand}, we posit the block \emph{Ansatz}
    \begin{equation}
      \Delta h_u
      =
      \begin{pmatrix}
        0 & \sqrt{\delta t}\,\xi_u^\dagger\\[2pt]
        \sqrt{\delta t}\,\xi_u & \Xi_u
      \end{pmatrix}
      +\order{\delta t},
      \label{eq:h-ansatz2}
    \end{equation}
    where $\sqrt{\delta t}\,\xi_u$ denotes a displacement vector acting on the $\dyad{0}{1}$ coherence block and $\Xi_u$ denotes a (Hermitian) generator of rotations on $\mathscr{H}_1$.
    
    Using \eqref{eq:rhoE-firstbin} and \eqref{eq:h-ansatz2}, we find that
    \begin{equation}
      \varphi_{\rm out} \Delta h_u =
      \begin{pmatrix}
        0 + \order{\delta t} & \sqrt{\delta t}\left(\xi_u^\dagger + i \mathfrak{m}^\dagger\, \Xi_u\right) + \order{\delta t} \\
        0 + \order{\delta t^{3/2}} & \delta t\left(-i \mathfrak{m}\, \xi_u^\dagger + \mathfrak{C}\,\Xi_u\right) + \order{\delta t^{3/2}}
      \end{pmatrix},
    \end{equation}
    such that
    \begin{align}
      \left(\{\varphi_{\rm out},\Delta h_u\}\right)_{01}
      &=
      \sqrt{\delta t}\left(\xi_u^\dagger+i\,\mathfrak{m}^\dagger \Xi_u\right)+\order{\delta t},
      \label{eq:block01-LHS}\\
      \left(\{\varphi_{\rm out},\Delta h_u\}\right)_{11}
      &=
      \delta t\left(\{\mathfrak{C},\Xi_u\}+i(\xi_u \mathfrak{m}^\dagger-\mathfrak{m}\,\xi_u^\dagger)\right)+\order{\delta t^{3/2}}.
      \label{eq:block11-LHS}
    \end{align}
    Imposing extremality, such that $\Delta h^\star_u$ satisfies Eq.~\eqref{eq:uhlmann-lyapunov}, we equate~\eqref{eq:block01-LHS} to the $\dyad{0}{1}$ block of $-2(B_u)_H$ from \eqref{eq:Bu-expand}, and similarly equate~\eqref{eq:block11-LHS} to the $\dyad{1}$ block $-2(B_u)_H$, resulting in the coupled equations
    \begin{align}
      \xi_u^\star &=i \Xi_u^\star\, \mathfrak{m}
      -\frac{1}{2} w_u,
      \label{eq:xi-from-H}\\
      \{\mathfrak{C},\Xi_u\} &=
      -2 \left((M_u)_H + \frac{i}{2}\left(\xi_u^\star \mathfrak{m}^\dagger-\mathfrak{m}{\xi_u^\star}^\dagger\right)\right).
      \label{eq:H-coupled}
    \end{align}
    Given $i(\xi_u^\star  \mathfrak{m}^\dagger-\mathfrak{m}{\xi_u^\star}^\dagger)=-\{\mathfrak{m} \mathfrak{m}^\dagger,\Xi_u\}-\frac{i}{2}(w_u \mathfrak{m}^\dagger- \mathfrak{m} w_u^\dagger)$, we recover the closed Lyapunov equation
    \begin{equation}
      \{\mathfrak{C}_{c},\,\Xi_u^\star\} = -2 \left((M_u)_H-\frac{i}{4}\left(w_u\mathfrak{m}^\dagger-\mathfrak{m}w_u^\dagger\right)\right),
      \label{eq:H-final}
    \end{equation}
    with $\mathfrak{C}_{c}=\mathfrak{C}-\mathfrak{m}\mathfrak{m}^\dagger$ [Eq.~\eqref{eq:dressed-moments}], which concludes the proof.
\end{proof}

\begin{cor}
    \label{cor:shorttime-no-mixing}
    Up to a counter-displacement $\xi^\star$, the canonical gauge is Uhlmann-extremal at order $\order{\delta t}$, i.e., $\Xi^\star=0$, iff
    \begin{equation}
      \sqrt{\Lambda}\left(F_u^\top {\mathcal{C}'_{c}}^\top - {\mathcal{C}'_{c}}^\top F_u^*\right) \sqrt{\Lambda}=0
      \quad\Longleftrightarrow\quad
      F_u^\dagger\,\mathcal{C}'_{c}=\mathcal{C}'_{c}\,F_u.
      \label{eq:no-mixing-condition}
    \end{equation}
\end{cor}

\subparagraph*{Remarks.}
For \emph{rate-only encoding} ($K'_u=0$, hence $F_u=\Lambda^{-1}\partial_u\Lambda$
is real diagonal in the canonical basis), Eq.~\eqref{eq:no-mixing-condition} reduces to $\comm{F_u}{\mathcal{C}'_{c}}=0$, which holds
whenever $\mathcal{C}'_{c}$ is block-diagonal in the canonical basis.

If the probe has vanishing connected correlators in the canonical jump basis, such that $\mathcal{C}'_{c}=0$, then $\Xi^\star_u=0$ for any $F_u$, and thus only the counter-displacement $\xi^\star$ is needed.

Finally, we may also write a formal solution for $\Xi_u^\star$. Let $\mathscr{A}_{\mathfrak{C}_{c}}$ denote the anticommutator superoperator with respect to the (dressed) connected correlator,
\begin{equation}
\mathscr{A}_{\mathfrak{C}_{c}}(X)\triangleq \{\mathfrak{C}_{c},X\}.
\end{equation}
On $\mathrm{supp}(\mathfrak{C}_{c})$ the map $\mathscr{A}_{\mathfrak{C}_{c}}$ is invertible and \eqref{eq:H-final} yields
\begin{equation}
\Xi_u^\star
=
-2\mathscr{A}_{\mathfrak{C}_{c}}^{-1}\left(M_u^{\rm eff}\right),
\qquad M_u^{\rm eff} = M_u-\frac{i}{4}\left(w_u\mathfrak{m}^\dagger-\mathfrak{m}w_u^\dagger\right),
\label{eq:H-solution}
\end{equation}
which is valid in the short-time regime. Equivalently, in the eigenbasis $\mathfrak{C}_{c}=\sum_k s_k\dyad{k}$,
\begin{equation}
    (\Xi_u^\star)_{k\ell}
    =
    -2\,\frac{(M_u^{\rm eff}(0))_{k\ell}}{s_k+s_\ell}.
    \label{eq:H-eigs}
\end{equation}

%==========
\subsection{Gap Between Canonical and Extremal Purifications}
\label{sec:qgt-gap-t0}

Here we assess the QGT/metric gap between the canonical collisional purification and the Uhlmann-extremal purification in the short-time limit.

Let $\mathcal Q_{uv}$ denote the QGT of the canonical purification $\ket{\Psi}=U(\bm{\theta};t)\ket{\psi}\otimes \ket{0}$. Appending an environment unitary $I\otimes U_E$, viz., $\ket{\Psi} \mapsto (I\otimes U_E)\ket{\Psi}$, shifts the centered generator of $U$ by $\Delta\widetilde G_u\mapsto \Delta\widetilde G_u+I\otimes \Delta h_u$ and, hence, the QGT of the new gauge becomes
\begin{align}
  \mathcal Q_{uv}(h)
  &= \ip{\partial_u \Psi_{U_E}}{\partial_v \Psi_{U_E}} - \ip{\partial_u \Psi_{U_E}}{\Psi_{U_E}}\ip{\Psi_{U_E}}{\partial_v \Psi_{U_E}} \nonumber \\ 
  &=\Tr(\Psi \left(\Delta\widetilde{G}_u+I\otimes\Delta h_u\right)\left(\Delta \widetilde{G}_v+I\otimes\Delta h_v\right)),
  \label{eq:QGT-gauge}
\end{align}
where we emphasize that $h$ parametrizes the gauge. Expanding \eqref{eq:QGT-gauge} then yields
\begin{equation}
  \mathcal Q_{uv}(h)
  =
  \mathcal Q_{uv}
  +\Tr(B_u\,\Delta h_v)
  +\Tr(\Delta h_u B_v^\dagger)
  +\Tr(\varphi_{\rm out}\,\Delta h_u\,\Delta h_v),
  \label{eq:QGT-expand}
\end{equation}
where $B_u\triangleq \Tr_{\rm sys}(\Psi\,\Delta\widetilde G_u)$ as before [Eq.~\eqref{eq:Bu}] and $\varphi_{\rm out}=\Tr_{\rm sys}(\Psi)$. The Uhlmann-extremal purification $\ket{\Psi^\star}$, i.e., the gauge choice $U_E^\star$ (and therefore $h^\star$) that minimizes the Fubini-Study metric, has QGT $\mathcal Q^{\star}_{uv} \triangleq \mathcal Q_{uv}(h^\star)$. We now quantify how close the canonical purification is to the extremal purification in terms of the gap between their QGTs.

\begin{prop}[QGT Gap]
  \label{prop:qgt-gap}
  Assume $\Delta h_u^\star$ solves the Uhlmann extremality equation~\eqref{eq:uhlmann-lyapunov} of Theorem~\ref{thm:lyapunov}. Then the QGT gap between the canonical and extremal purifications admits the decomposition
  \begin{equation}
    \mathcal Q_{uv}-\mathcal{Q}^\star_{uv}
    =
    \Tr\big(\varphi_{\rm out}\,\Delta h_v^\star\,\Delta h_u^\star\big)
    -\Tr((B_u)_A\,\Delta h_v^\star)
    +\Tr(\Delta h_u^\star (B_v)_A),
    \label{eq:qgt-gap}
  \end{equation}
  where $\mathcal{Q}^\star_{uv}\triangleq \mathcal Q_{uv}(h^\star)$ and $(B_u)_A=\tfrac12(B_u-B_u^\dagger)$ is the anti-Hermitian part of $B_u$ [defined in Eq.~\eqref{eq:Bu}]. Therefore, the metric gap, which relates to the real part of the QGT, is
  \begin{equation}
    g_{uv}-g^\star_{uv}
    =
    \frac{1}{2}\,\Tr(\varphi_{\rm out}\{\Delta h_u^\star,\Delta h_v^\star\}) \geq 0,
    \label{eq:metric-gap}
  \end{equation}
  and the curvature gap, which relates the imaginary part of the QGT, is
  \begin{equation}
      \Omega_{uv} - \Omega^\star_{uv} = -i\Tr(\varphi_{\rm out} \comm{\Delta h_u^\star}{\Delta h_v^\star}) + 2i\Tr(\Delta h_u^\star (B_v)_A - (B_u)_A\Delta h_v^\star).
    \label{eq:curvature-gap}
  \end{equation}
\end{prop}

\begin{proof}
  Evaluate Eq.~\eqref{eq:QGT-expand} at $h^\star$,
  \[
    \mathcal Q_{uv}-\mathcal{Q}^\star_{uv}
    =
    -\Tr(B_u\,\Delta h_v^\star)-\Tr(B_v^\dagger\,\Delta h_u^\star)-\Tr(\varphi_{\rm out}\,\Delta h_v^\star\Delta h_u^\star).
  \]
  Use $B_u=(B_u)_H+(B_u)_A$ and $B_v^\dagger= (B_v)_H - (B_v)_A$ to rewrite the gap as 
  \begin{equation}
      \mathcal Q_{uv}-\mathcal{Q}^\star_{uv} = - \Tr(\varphi_{\rm out}\,\Delta h_v^\star\Delta h_u^\star) - \Tr((B_u)_H \Delta h_v^\star + \Delta h_u^\star(B_u)_H) \\ - \Tr((B_u)_A \Delta h_v^\star - \Delta h_u^\star(B_u)_A).
  \end{equation}
  Now consider the Lyapunov equation~\eqref{eq:uhlmann-lyapunov}, multiply it by $\Delta h_v^\star$, and take the trace:
  \[
    \Tr(\{\varphi_{\rm out},\Delta h_u^\star\}\Delta h_v^\star)=-2\Tr((B_u)_H\Delta h_v^\star),
  \]
  By cyclicity,
  \[
    \Tr(\{\varphi_{\rm out},\Delta h_u^\star\}\Delta h_v^\star)
    =
    \Tr\big(\varphi_{\rm out}(\Delta h_u^\star\Delta h_v^\star+\Delta h_v^\star\Delta h_u^\star)\big)
    =
    \Tr\big(\varphi_{\rm out}\{\Delta h_u^\star,\Delta h_v^\star\}\big),
  \]
  hence
  \[
    \Tr((B_u)_H\Delta h_v^\star)=-\frac12\Tr(\varphi_{\rm out}\{\Delta h_u^\star,\Delta h_v^\star\}),
  \]
  and similarly for $u \leftrightarrow v$. Substituting produces \eqref{eq:qgt-gap}. Finally, terms like $\Tr((B_u)_A \Delta h_{v}^\star)\in i\mathbb{R}$ and only contribute to the curvature. Therefore, Eqs.~\eqref{eq:metric-gap} and~\eqref{eq:curvature-gap} follow.
\end{proof}

We now express the QGT gap in the short-time regime, where the single-excitation sector of the environment is most informative (see Sec.~\ref{app:short-time-extremality}). 
\begin{cor}[Short-time gap]
    In the short-time regime, from $t=0$ to $t=\delta t$, the metric gap decomposes to
    \begin{equation}
        g_{uv}-g^\star_{uv}
        =\frac{\delta t}{2} \left(\frac{1}{4} \Tr(m^* m^\top\left(F_u\Lambda F_v^\dagger+F_v\Lambda F_u^\dagger\right))
        +\Tr({\mathcal C_c'}^{\top}\sqrt{\Lambda}\, \acomm{\Xi_u^\star}{\Xi_v^\star}\, \sqrt{\Lambda})\right)
        +\, \order{\delta t^{3/2}},
        \label{eq:shorttime-metric-gap}
    \end{equation}
    where $m$ and $\mathcal{C}'_c$ are the mean and (connected) correlator of the initial probe state in the canonical basis [Eq.~\eqref{eq:init-moments}] and $F_u$ is defined in Eq.~\eqref{eq:Falpha}.
\end{cor}

\begin{proof}
   From Eq.~\eqref{eq:metric-gap}, the metric gap depends only on $\Tr(\varphi_{\rm out}\,\Delta h_u^\star\,\Delta h_v^\star)$. To calculate this, use the block \emph{Ansatz} from Proposition~\ref{prop:ansatz}:
    \begin{equation}
      \Delta h_u^\star
      =
      \begin{pmatrix}
        0 & \sqrt{\delta t}\,\xi_u^{\star\dagger}\\[2pt]
        \sqrt{\delta t}\,\xi_u^\star & \Xi_u^\star
      \end{pmatrix}
      +\order{\delta t}.
    \end{equation}
    Direct multiplication using the short-time form of $\varphi_{\rm out}(\delta t)$ in Eq.~\eqref{eq:rhoE-firstbin} gives, at leading order,
    \begin{equation}
      \Tr(\varphi_{\rm out}\,\Delta h_u^\star\,\Delta h_v^\star)
      =
      \delta t\left[
        \Tr(\left(\xi_v^\star-i\Xi_v^\star\mathfrak{m}\right)\left(\xi_u^\star-i\Xi_u^\star\mathfrak{m}\right)^\dagger)
        +
        \Tr(\mathfrak{C}_{c}\,\Xi_u^\star \Xi_v^\star)
      \right]
      +\order{\delta t^{3/2}}.
      \label{eq:additivity-general}
    \end{equation}
    The extremal displacement obeys $\xi_u^\star-i\Xi_u^\star\mathfrak{m}=-\tfrac12 w_u$, with $w_u=\sqrt{\Lambda}\,F_u^\top m$.
    Use $\mathfrak{C}_c=\sqrt{\Lambda} {\mathcal{C}_c'}^\top \sqrt{\Lambda}$ [Eq.~\eqref{eq:dressed-moments}] and substitute these expressions into~\eqref{eq:additivity-general} to find
    \begin{equation}
      \Tr(\varphi_{\rm out}\,\Delta h_u^\star\,\Delta h_v^\star)
      =
      \delta t\left[
        \frac14\Tr(w_v w_u^\dagger)
        +
        \Tr({\mathcal{C}'_{c}}^\top\,\sqrt{\Lambda}\Xi_u^\star \Xi_v^\star\sqrt{\Lambda})
      \right]
      +\order{\delta t^{3/2}}.
      \label{eq:additivity-opt}
    \end{equation}
    Expand the displacement term,
    \begin{equation}
      \Tr(w_v w_u^\dagger)
      = \Tr(\sqrt{\Lambda} F_v^\top m m^\dagger F_u^* \sqrt{\Lambda}) = \Tr(m^*m^\top\, F_v\Lambda F_u^\dagger),
    \end{equation}
    where we use $\Tr(A^\top)=\Tr(A)$. Substituting and taking the real part validates~\eqref{eq:shorttime-metric-gap}.
\end{proof}

%==========
%==========
\section{Details on Learning Pauli Noise}
\label{app:pauli-noise}

In this Appendix, we derive a precision bound for the RPM protocol that indicates exponential complexity overhead for learning Pauli noise rates without quantum memory. The arguments rely on general system-only POVMs, as opposed to the simplification using projective measurements (PVMs) in the main text (Sec.~\ref{sec:pauli-learning}). This analysis applies to the short-time, or equivalently weak-noise, Pauli noise rate problem; cf.~\cite{Flammia2020:PauliLearn} for a complementary approach in the weak-noise setting and Refs.~\cite{Tu2025:LearnMixedU, Kwon2026:FisherComplexity} for related Fisher-information analyses and more broadly Refs.~\cite{Chen2022:ExpSepMemory, Chen2022:PauliChEst, Chen2024:TightPauliLearn, Caro2024:PTMlearning, Kim2025:LearnResources}.

Consider an $N$-qubit pure probe state $\ket{\psi}$ evolving for a short time bin $\delta t$ under the Pauli-noise channel:
\begin{equation}
    \label{appeq:pauli-short-time}
    \rho(\bm{\theta};\delta t)
    =
    \left(1-\delta t \sum_{a=1}^{R}\gamma_a\right)\psi
    +
    \delta t \sum_{a=1}^{R}\gamma_a \, P_a \psi P_a
    +
    \order{\delta t^2},
\end{equation}
where $\{P_a\}_{a=1}^{R}$ are the non-identity Pauli strings, $R=4^N-1$, $\bm{\theta}=(\gamma_a)_{a=1}^R$ are the estimands with $\gamma_a \ge 0$ the Pauli noise rates. Let $D \triangleq 2^N$ denote the system dimension, and define $\gamma_{\min}\triangleq \min_a \gamma_a$.

We consider an arbitrary system-only POVM $\{E_x\}_x$ on the system Hilbert space $\mathscr{H}_S$, with outcome probabilities
\begin{equation}
    \label{eq:px}
    p_x(\bm{\theta};\delta t)
    =
    \Tr(E_x\,\rho(\bm{\theta};\delta t)).
\end{equation}
Let $F(\delta t)$ denote the associated classical Fisher information matrix with respect to the parameters $\bm{\theta}$, and consider the short-time Fisher rate $\dot F(0^+) = \lim_{\delta t\to 0} F(\delta t)/\delta t$ [see Eq.~\eqref{eq:F-short-time}]. We now establish the following complexity bound on the average precision for any RPM protocol.

\begin{prop}[RPM Complexity: Learning Pauli Noise Without Quantum Memory]
    \label{prop:exp-pauli-learn}
    Consider any short-time RPM protocol of total sensing time $T$, built from an arbitrary sequence of rapid prepare-and-measure rounds on the system alone (no quantum memory); a classical register and adaptivity are allowed. Then the total Fisher information satisfies
    \begin{equation}
        \Tr(F_{\rm tot})
        \le
        T\,\frac{D-1}{\gamma_{\min}},
    \end{equation}
    at leading order in $\delta t$. Consequently, the average absolute precision (per parameter) obeys
    \begin{equation}
        \bar{\tau}
        \le
        \frac{T}{\gamma_{\min}(D+1)}
        =
        \Theta\!\left(T\gamma_{\min}^{-1}2^{-N}\right).
    \end{equation}
\end{prop}

\begin{proof}
    Consider a single short-time round with probe state $\psi$, POVM $\{E_x\}_x$, and duration $\delta t$. For each measurement outcome $x$, define
    \begin{equation}
        e_x \triangleq \Tr(E_x\psi),
        \qquad
        q_{x,a} \triangleq \Tr(E_x P_a \psi P_a).
    \end{equation}
    Expanding the outcome probabilities~\eqref{eq:px} to first order in $\delta t$ gives
    \begin{equation}
        p_x(\bm{\theta};\delta t)
        =
        e_x
        +
        \delta t \sum_{a=1}^{R}\gamma_a\big(q_{x,a}-e_x\big)
        +
        \order{\delta t^2}.
    \end{equation}
    Hence
    \begin{equation}
        \partial_{\gamma_b}p_x
        =
        \delta t \big(q_{x,b}-e_x\big)
        +
        \order{\delta t^2}.
    \end{equation}
    Now split the outcomes into two classes:

    If $e_x>0$, then $p_x(\bm{\theta};\delta t)=e_x+\order{\delta t}$ is order unity, while $\partial_{\gamma_b}p_x=\order{\delta t}$. Therefore the contribution of these outcomes to $F(\delta t)/\delta t$ vanishes as $\delta t\to 0$.

    If $e_x=0$, then positivity implies $E_x\ket{\psi}=0$, and
    \begin{equation}
        p_x(\bm{\theta};\delta t)
        =
        \delta t \sum_{a=1}^{R}\gamma_a q_{x,a}
        +
        \order{\delta t^2},
        \qquad
        \partial_{\gamma_b}p_x
        =
        \delta t\, q_{x,b}
        +
        \order{\delta t^2}.
    \end{equation}
    Thus only these outcomes contribute to the short-time Fisher rate, and
    \begin{equation}
        \dot F_{bc}(0^+)
        =
        \sum_{x:\,e_x=0}
        \frac{q_{x,b}q_{x,c}}{\sum_{a=1}^{R}\gamma_a q_{x,a}}.
    \end{equation}
    Taking the trace,
    \begin{equation}
        \Tr(\dot F(0^+))
        =
        \sum_{x:\,e_x=0}
        \frac{\sum_{b=1}^{R} q_{x,b}^2}{\sum_{a=1}^{R}\gamma_a q_{x,a}}
        \le
        \frac{1}{\gamma_{\min}}
        \sum_{x:\,e_x=0}
        \frac{\sum_{b=1}^{R} q_{x,b}^2}{\sum_{a=1}^{R} q_{x,a}}.
    \end{equation}
    For each $x$ and $a$,
    \begin{equation}
        0 \le q_{x,a} \le \Tr(E_x),
    \end{equation}
    since $P_a\psi P_a$ is a rank-$1$ projector. Therefore
    \begin{equation}
        \sum_{b=1}^{R} q_{x,b}^2
        \le
        \Tr(E_x)\sum_{a=1}^{R} q_{x,a},
    \end{equation}
    and so
    \begin{equation}
        \Tr(\dot F(0^+))
        \le
        \frac{1}{\gamma_{\min}}
        \sum_{x:\,e_x=0}\Tr(E_x).
    \end{equation}

    Define
    \begin{equation}
        E_\perp \triangleq \sum_{x:\,e_x=0} E_x.
    \end{equation}
    Because each $E_x\ge 0$ and $e_x=\expval{E_x}{\psi}=0$, we have $E_x\ket{\psi}=0$ for every such $x$, hence
    \begin{equation}
        E_\perp\ket{\psi}=0.
    \end{equation}
    Moreover, $0\le E_\perp \le \sum_x E_x = I$, so $E_\perp$ is supported on the $(D-1)$-dimensional orthogonal complement of $\ket{\psi}$. Therefore
    \begin{equation}
        \sum_{x:\,e_x=0}\Tr(E_x)
        =
        \Tr(E_\perp)
        \le
        D-1.
    \end{equation}
    Combining the above inequalities yields
    \begin{equation}
        \Tr(\dot F(0^+))
        \le
        \frac{D-1}{\gamma_{\min}}.
    \end{equation}

    Now consider an arbitrary RPM protocol, consisting of many rounds and taking a total time $T=\sum_j \delta t_j$. In each round $j$, possibly conditioned on the entire previous classical record, the same one-round bound applies. Hence, within an interval $\delta t_j$,
    \begin{equation}
        \Tr(F^{(j)})= \delta t_j \Tr(\dot F(0^+)) + \order{\delta t_j^2}
        \le
        \delta t_j\,\frac{D-1}{\gamma_{\min}}+\order{\delta t_j^2}.
    \end{equation}
    Summing over all rounds gives
    \begin{equation}
        \Tr(F_{\rm tot})
        \le
        T\,\frac{D-1}{\gamma_{\min}},
    \end{equation}
    at leading order. Finally, using $R=D^2-1$ and the average absolute precision bound $\bar{\tau}\le \Tr(F_{\rm tot})/R$, which follows from the Cram\'er-Rao bound, we obtain
    \begin{equation}
        \bar{\tau}
        \le
        T\,\frac{D-1}{\gamma_{\min}(D^2-1)}
        =
        \frac{T}{\gamma_{\min}(D+1)}
        =
        \Theta\!\left(T\,\gamma_{\min}^{-1}2^{-N}\right),
    \end{equation}
    which concludes the proof.
\end{proof}

Thus, without quantum memory, the average precision in the RPM protocol is exponentially suppressed in $N$, even when allowing arbitrary system-only POVMs.

\subparagraph*{Remarks.}
Proposition~\ref{prop:exp-pauli-learn} holds even if we append an initial classical register and perform joint measurements on the register and system. Indeed, this classical-memory-assisted protocol~\cite{Chen2024:TightPauliLearn} is equivalent to a classical randomization over system-only experiments: given
\begin{equation}
    \rho_{CS}(0)=\sum_c p_c \dyad{c}\otimes \psi_c \longrightarrow \rho_{CS}(\bm{\theta}; \delta t)=\sum_c p_c \dyad{c}\otimes \rho_c(\bm{\theta};\delta t), 
\end{equation}
a final joint POVM induces conditional system POVMs $\{E_{c,x}\}_x$ and outcome probabilities
\begin{equation}
    p_{c,x}(\bm{\theta};\delta t)=p_c\,\Tr(E_{c,x}\, \rho_c(\bm{\theta};\delta t)).
\end{equation}
Since Proposition~\ref{prop:exp-pauli-learn} applies to each branch separately, and coarse-graining over the classical label cannot increase Fisher information (data processing), the trace bound follows immediately.

Proposition~\ref{prop:exp-pauli-learn} is formulated for the short-time estimation of Pauli noise rates $\bm{\theta}=(\gamma_a)_{a=1}^R$. The exponential complexity barrier persists when learning the discrete Pauli error probabilities $\{r_a\}$~\cite{Chen2024:TightPauliLearn}, since both models are governed by the same underlying Pauli-channel geometry. The proof given here, however, exploits the special short-time (weak-noise) structure of the rate-estimation problem, in which only rare jump events contribute at leading order. For the discrete model with arbitrary $\{r_a\}$, this simplification no longer seems available, and one needs slightly different proof techniques~\cite{Kwon2026:FisherComplexity, Tu2025:LearnMixedU}.

Here we bound the \emph{absolute} precision per parameter. We can likewise bound the \emph{relative} precision per parameter, $\bar{\tau}_{\rm rel}$, which is dimensionless and defined as $\bar{\tau}_{\rm rel} \triangleq (\sum_{a=1}^R \gamma_a^2 \tau_a)/R$. In fact, using similar arguments as above, one can show that $\bar{\tau}_{\rm rel} \leq T\gamma_{\rm max} /(D+1)=\Theta(T\gamma_{\rm max} 2^{-N})$, implying an exponential time complexity to achieve a target relative precision without quantum memory, $T=\Theta(\bar{\tau}_{\rm rel}\gamma_{\rm max}^{-1} 2^N)$.

%==========
%==========
\section{Details on Subdiffraction Quantum Imaging}
\label{app:subdiff-imaging}

In this Appendix, we provide details and derivations underlying the subdiffraction imaging problem in Sec.~\ref{sec:subdiff} of the main text.

\paragraph*{Problem setup.}
We consider two equally bright incoherent emitters at image-plane positions $p_{1,2}=\bar{x}\mp d/2$ with a Gaussian point-spread function (PSF) of width $\sigma$. Let $\phi_0(u)$ denote the normalized Gaussian mode centered at the origin,
\begin{equation}
  \phi_0(u)=\frac{1}{(2\pi\sigma^2)^{1/4}}\exp(-\frac{u^2}{4\sigma^2}),
\end{equation}
so that $|\phi_0(u)|^2$ has variance $\sigma^2$.
Define the displaced PSFs
\begin{equation}
  \phi_1(u)=\phi_0\big(u-(\bar{x}-d/2)\big),\qquad
  \phi_2(u)=\phi_0\big(u-(\bar{x}+d/2)\big).
\end{equation}
In the image-plane coordinate basis, the mutual coherence matrix of the fluorescence scene\footnote{Physically, the mutual coherence $\Gamma_{uv}$ is just the two-point correlator of the thermal emission field in image-plane coordinates. } has elements~\cite{Lupo2016:QuImaging}
\begin{equation}
  \Gamma_{uv}(\bar{x},d)=\frac{\varepsilon}{2}
  \big(\phi_1^*(v)\phi_1(u)+\phi_2^*(v)\phi_2(u)\big),
  \qquad
  \Tr{\Gamma}=\varepsilon,
  \label{eq:Gamma-two-source}
\end{equation}
where $\varepsilon$ equals the detected brightness per bin. In the weak-signal regime, $\varepsilon\ll 1$, the optical state in a single bin is well approximated by
\begin{equation}
  \rho(\bar{x},d)\approx (1-\varepsilon)\dyad{0}
  +\varepsilon\,\rho_1(\bar{x},d),
\end{equation}
where the one-photon component is an incoherent mixture
\begin{equation}
  \rho_1(\bar{x},d)
  =\varepsilon^{-1}\sum_{u,v}\Gamma_{uv}(\bar{x},d)\,
  \hat a_u^\dagger \dyad{0}\hat a_v
  =\frac{1}{2}\sum_{k=1}^2 \dyad{\phi_k},
\end{equation}
with
\begin{equation}
  \ket{\phi_k}\triangleq \sum_u \phi_k(u)\hat a_u^\dagger\ket{0}.
\end{equation}
Here $\hat a_u^\dagger$ creates a photon at image-plane coordinate $u$.

To place this in our Lindblad framework, we interpret $\Gamma(\bar{x},d)$ as the dissipation matrix of a weak, Markovian excitation process with jump operators $L_u=\hat{a}_u^\dagger$ acting on the vacuum. In the weak-signal regime [Scenario 3, Eq.~\eqref{eq:F-short-time}], the per-bin QFI obeys $F_Q(\delta t)\approx \delta t\,\dot F_Q(0^+)$, which we evaluate via Theorem~\ref{thm:qfi-matrix}. In particular, the quantum correlator matrix in this regime satisfies
\begin{equation}
  \mathcal{C}_{uv}(0^+)=\Tr(L_u^\dagger L_v\,\rho(0^+))
  =\Tr(\hat{a}_u \hat{a}_v^\dagger\dyad{0}{0})
  =\delta_{uv},
\end{equation}
so $\mathcal{C}(0^+)=I$ to leading order.

\paragraph*{Spectrum of $\Gamma$.}
The identical-emitter assumption implies global translations act unitarily on the mode space,
\begin{equation}
  \Gamma(\bar{x},d)=T(\bar{x})\,\Gamma(0,d)\,T^\dagger(\bar{x}),
  \label{eq:Gamma-translation}
\end{equation}
where $(T(\bar{x})f)(u)=f(u-\bar{x})$.
At $\bar{x}=0$, define the overlap of the mode functions,
\begin{equation}
  S(d)\triangleq \braket{\phi_1}{\phi_2}=e^{-d^2/(8\sigma^2)}.
\end{equation}
The rank-$2$ operator $\Gamma(0,d)$ then has eigenvalues~\cite{Lupo2016:QuImaging}
\begin{equation}
  \lambda_\pm(d)=\frac{\varepsilon}{2}\big(1\pm S(d)\big),
\end{equation}
and normalized eigenvectors
\begin{equation}
  \ket{\phi_\pm(d)}=\frac{\ket{\phi_1}\pm\ket{\phi_2}}{\sqrt{2(1\pm S(d))}}.
\end{equation}
Thus $\Gamma(0,d)=V_0(d)\Lambda(d)V_0^\dagger(d)$ where
$V_0(d)=(\ket{\phi_+(d)},\ket{\phi_-(d)})$ and $\Lambda(d)=\mathrm{diag}\{\lambda_+(d),\lambda_-(d)\}$, and
\begin{equation}
  V(\bar{x},d)=T(\bar{x})V_0(d),\qquad
  \Gamma(\bar{x},d)=V(\bar{x},d)\Lambda(d)V^\dagger(\bar{x},d).
\end{equation}
In the subdiffraction regime $d/\sigma\to 0$, the dependence of $V_0(d)$ on $d$ is subleading compared to the eigenvalue dependence of $\Lambda(d)$; in particular we may use the approximations
\begin{equation}
  V(\bar{x},d)\approx V(\bar{x})=T(\bar{x}) V_0(0), \qquad \lambda_+(d)\approx \varepsilon\left(1-\frac{d^2}{16\sigma^2}\right),\qquad
  \lambda_-(d)\approx \frac{\varepsilon d^2}{16\sigma^2}.
  \label{eq:lambda-small-d}
\end{equation}

\paragraph*{Subdiffraction QFI matrix.}
We now evaluate the QFI matrix using Theorem~\ref{thm:qfi-matrix}. Since $\mathcal{C}(0^+)=I$, the QFI is determined completely by the dissipation kernel $\mathcal{K}_{\alpha\beta}$.

We first calculate the QFI for $\bar{x}$. From $V(\bar{x},d)=T(\bar{x})V_0(d)$, the centroid dependence enters through the translation unitary. Writing $T(\bar{x})=e^{-i\bar{x}\hat p}$ with generator $\hat p=-i\partial_u$ on mode functions, the eigenframe connection is
\begin{equation}
  K_{\bar{x}}=(\partial_{\bar{x}}V)V^\dagger=(\partial_{\bar{x}}T)T^\dagger=-i\hat p,
  \qquad K_{\bar{x}}^\dagger=-K_{\bar{x}}.
  \label{eq:Kx-translation}
\end{equation}
Combining this with the rate equation [Eq.~\eqref{eq:gamma-rate-eq}], we have $\partial_{\bar{x}}\Lambda=0$ and hence $\partial_{\bar{x}}\Gamma=[K_{\bar{x}},\Gamma]$. Therefore,
\begin{equation}
  \partial_{\bar{x}}\Gamma+\{K_{\bar{x}},\Gamma\}
  =[K_{\bar{x}},\Gamma]+\{K_{\bar{x}},\Gamma\}
  =2K_{\bar{x}}\Gamma,
\end{equation}
and application of Theorem~\ref{thm:qfi-matrix} yields
\begin{equation}
  \mathcal{K}_{\bar{x}\bar{x}}=4K_{\bar{x}}\Gamma K_{\bar{x}}^\dagger.
\end{equation}
Contracting with $\mathcal{C}(0^+)=I$ gives
\begin{equation}
  (F_Q)_{\bar{x}\bar{x}}(\delta t)=4\Tr(K_{\bar{x}}\Gamma K_{\bar{x}}^\dagger)
  =4\Tr(\hat p^{\,2}\Gamma).
  \label{eq:Fxx-eval}
\end{equation}
In the limit $d/\sigma\to 0$, a direct calculation for the Gaussian PSF gives $\Tr(\hat p^{\,2}\Gamma)=\varepsilon/(4\sigma^2)$~\cite{Lupo2020:LinearOptLimits}, hence
\begin{equation}
  \lim_{d/\sigma \to 0} (F_Q)_{\bar{x}\bar{x}}(\delta t) = \frac{\varepsilon}{\sigma^2}.
\end{equation}

We now calculate the QFI for $d$. In the subdiffraction regime $V(\bar{x},d)\approx V(\bar{x})$, we have $\partial_d\Gamma\approx V(\partial_d\Lambda)V^\dagger$. Thus
\begin{equation}
  \mathcal{K}_{dd}\approx V\big((\partial_d\Lambda)\Lambda^{-1}(\partial_d\Lambda)\big)V^\dagger,
\end{equation}
so with $\mathcal{C}(0^+)=I$,
\begin{equation}
  (F_Q)_{dd}(\delta t)
  \approx \frac{(\partial_d\lambda_+)^2}{\lambda_+}+\frac{(\partial_d\lambda_-)^2}{\lambda_-}.
  \label{eq:Fdd-exact-form}
\end{equation}
Using $S'(d)=-(d/(4\sigma^2))S(d)$ and $\partial_d\lambda_\pm=\pm(\varepsilon/2)S'(d)$, we obtain
\begin{equation}
  \lim_{d/\sigma\to 0}(F_Q)_{dd}(\delta t) = \frac{\varepsilon}{4\sigma^2}.
\end{equation}

We now calculate the cross terms of the QFI matrix. To leading order, the off-diagonal term of the dissipation kernel takes the form
\begin{equation}
  \mathcal{K}_{\bar{x}d}\ \propto\ V\,(\partial_d\Lambda^{1/2})\Lambda^{1/2}\,V^\dagger\,K_{\bar{x}}^\dagger,
\end{equation}
so that, in the $\Gamma$-eigenbasis,
\begin{equation}
  \Tr\!\big(\mathcal{K}_{\bar{x}d}\big)\ \propto\ \Tr\!\big((\partial_d\Lambda^{1/2})\Lambda^{1/2}\,(K'_{\bar{x}})^\dagger\big).
\end{equation}
Here $(\partial_d\Lambda^{1/2})\Lambda^{1/2}$ is diagonal, while $K'_{\bar{x}}$ has vanishing diagonal entries in the real eigenbasis $\phi_\pm$. Hence
\begin{equation}
  \Tr\!\big(\mathcal{K}_{\bar{x}d}\big)=0
  \quad\Rightarrow\quad
  (F_Q)_{\bar{x}d}(\delta t)=0.
\end{equation}

Collecting everything together yields the subdiffraction QFI matrix for $(\bar{x}, d)$ after $\nu$ measurements:
\begin{equation}
  \lim_{d/\sigma\to 0}F_Q^{\rm tot}=\frac{\nu\varepsilon}{\sigma^2}
  \begin{pmatrix}
    1 & 0\\
    0 & \tfrac14
  \end{pmatrix},
\end{equation}
where $\nu\varepsilon$ is the total mean-number of photons detected. This agrees with the fundamental quantum limits for centroid and separation estimation in the
weak-thermal, Gaussian-PSF model~\cite{Tsang2016:SPADE, Rehacek2017:MultiImaging, Lupo2020:LinearOptLimits}.

\end{document}